\documentclass[a4paper,UKenglish,cleveref,autoref,thm-restate,nolineno]
{socg-lipics-v2021}

\pdfoutput=1 
\hideLIPIcs  

    \title{Parameterized Complexity of Temporal Connected Components} 

\titlerunning{Parameterized Complexity of Temporal Connected Components} 

\author{Argyrios Deligkas}
    {Royal Holloway, University of London, Egham, United Kingdom \and \url{https://sites.google.com/view/deligkas/home}}
    {argyrios.deligkas@rhul.ac.uk}
    {https://orcid.org/0000-0002-6513-6748}
    {EPSRC Grant EP/X039862/1 ``NAfANE: New Approaches for Approximate Nash Equilibria''}

    \author{Michelle D\"oring}
    {Hasso Plattner Institute, University of Potsdam, Potsdam, Germany \and \url{https://michelledoering.notion.site/}}
    {michelle.doering@hpi.de}
    {https://orcid.org/0000-0001-7737-3903}
    {HPI Research School on Foundations of AI (FAI)}

    \author{Eduard Eiben}
    {Royal Holloway, University of London, Egham, United Kingdom \and \url{https://pure.royalholloway.ac.uk/en/persons/eduard-eiben}}
    {eduard.eiben@rhul.ac.uk}
    {https://orcid.org/0000-0003-2628-3435}
    {}

    \author{Tiger-Lily Goldsmith}
    {Royal Holloway, University of London, Egham, United Kingdom}
    {tigerlily.goldsmith@gmail.com}
    {https://orcid.org/0000-0003-0458-6267}
    {}

    \author{George Skretas}
    {Hasso Plattner Institute, University of Potsdam, Potsdam, Germany}
    {georgios.skretas@hpi.de}
    {https://orcid.org/0000-0003-2514-8004}
    {}

    \author{Georg Tennigkeit}
    {Hasso Plattner Institute, University of Potsdam, Potsdam, Germany}
    {georg.tennigkeit@hpi.de}
    {https://orcid.org/0000-0003-0734-0684}
    {HPI Research School on Foundations of AI (FAI)}

\authorrunning{A. Deligkas, M. Döring, E. Eiben, T. Goldsmith, G. Skretas and G. Tennigkeit} 

\Copyright{Jane Open Access and Joan R. Public} 



\category{} 

\relatedversion{} 

\EventEditors{John Q. Open and Joan R. Access}
\EventNoEds{2}
\EventLongTitle{42nd Conference on Very Important Topics (CVIT 2016)}
\EventShortTitle{CVIT 2016}
\EventAcronym{CVIT}
\EventYear{2016}
\EventDate{December 24--27, 2016}
\EventLocation{Little Whinging, United Kingdom}
\EventLogo{}
\SeriesVolume{42}
\ArticleNo{23}

\usepackage[dvipsnames]{xcolor}
\usepackage{xspace}
\usepackage{fancyhdr}
\usepackage{needspace} 
\usepackage{tcolorbox}
    
\usepackage{graphicx}
\usepackage{tikz} 

\usepackage{amssymb}
\usepackage{amsthm}

\usepackage{amsmath,mathtools}
\usepackage{bbm}

\usepackage{stmaryrd} 
\usepackage{stackengine}
\usepackage[english,ruled,vlined,algo2e,linesnumbered]{algorithm2e}
\usepackage{bm}

\usepackage{multirow}
\usepackage{tabularray}
\usepackage{colortbl}

\usepackage{anyfontsize}
\usepackage[percent]{overpic}

\usepackage{dsfont}

\newtheorem{result}{Result}
\crefname{result}{result}{results}
\Crefname{result}{Result}{Results}
\crefname{observation}{observation}{observations}
\Crefname{observation}{Observation}{Observations}

\definecolor{darkgray}{gray}{0.25}
    \newenvironment{construction*}{%
        \par\vspace{0.25\baselineskip}
        \pushQED{\qed}%
        \noindent\textcolor{darkgray}{$\triangleright$}\;%
        \noindent\textbf{\textcolor{darkgray}{Construction~\theconstruction.}}\ %
    }{%
      \popQED\par\vspace{0.25\baselineskip}
    }
\definecolor{MAGENTA}{cmyk}{0,1,0,0}
\theoremstyle{remark}

\declaretheoremstyle[
  headfont=\bfseries,
  headformat=$\lozenge$\ \NAME\NOTE, 
  headpunct={.},
  notefont=\bfseries,
  bodyfont=\normalfont,
  spaceabove=6pt, spacebelow=6pt
]{cleancons}
\theoremstyle{cleancons}

\newcommand{\ie}{i.\,e.,\xspace}

\newcommand{\rcal}{\ensuremath{\mathcal{R}}\xspace}
\newcommand{\hcal}{\ensuremath{\mathcal{H}}\xspace}

\newcommand{\pcal}{\ensuremath{\mathcal{P}}\xspace}
\newcommand{\gcal}{\ensuremath{\mathcal{G}}\xspace}
\newcommand{\ecal}{\ensuremath{\mathcal{E}}\xspace}

\newcommand{\tuple}[1]{\ensuremath{\langle {#1} \rangle}\xspace}

\newcommand{\smallparagraph}[1]{\vspace{0.3em}\noindent\emph{#1}\quad}
\newcommand{\bigparagraph}[1]{\vspace{0.4em}\noindent\textbf{#1}}

\newcommand{\tdegree}{\ensuremath{\Delta^t}\xspace}

\newcommand{\lifetime}{\ensuremath{\Lambda}\xspace}

\newcommand{\leftreach}{\ensuremath{\curvearrowleft}\xspace}
\newcommand{\rightreach}{\ensuremath{\curvearrowright}\xspace}

\newcommand{\kpathgraph}[1]{\ensuremath{#1}-path graph\xspace}

\newcommand{\kpathgraphs}[1]{\ensuremath{#1}-path graphs\xspace}

\newcommand{\kwalkgraphs}[1]{\ensuremath{#1}-walk graphs\xspace}

\newcommand{\openTCC}{\textup{\textsc{openTCC}}\xspace}
\newcommand{\closedTCC}{\textup{\textsc{closedTCC}}\xspace}
\newcommand{\otcc}{open tcc\xspace}
\newcommand{\ctcc}{closed tcc\xspace}

\newcommand{\paraCliqueSize}{\ensuremath{s}\xspace}
\newcommand{\paraCompSize}{\ensuremath{s}\xspace}
\newcommand{\paraPaths}{\ensuremath{k}\xspace}

\newcommand{\tpn}{\ensuremath{\mathsf{tpn}}\xspace}
\newcommand{\twn}{\ensuremath{\mathsf{twn}}\xspace}

\DeclareMathOperator{\adj}{adj}
\DeclareMathOperator{\inc}{inc}
\DeclareMathOperator{\source}{source}
\DeclareMathOperator{\target}{target}
\DeclareMathOperator{\timeTedge}{time}
\DeclareMathOperator{\edgeTedge}{edge}
\DeclareMathOperator{\pres}{pres}
\DeclareMathOperator{\psuc}{pos\_suc}
\DeclareMathOperator{\pnsuc}{pos\_nsuc}

\DeclareMathOperator{\tadj}{adj_t}
\DeclareMathOperator{\pathPred}{path}
\DeclareMathOperator{\pathPredX}{path_X}

\DeclareMathOperator*{\argmax}{arg\,max}

\newcommand{\NP}{\ensuremath{\mathtt{NP}}\xspace}
\newcommand{\coNP}{\ensuremath{\mathtt{coNP}}\xspace}
\newcommand{\paraNP}{\ensuremath{\mathtt{paraNP}}\xspace}

\newcommand{\Wone}{\ensuremath{\mathtt{W}[1]}\xspace}

\newcommand{\PP}{\ensuremath{\mathtt{P}}\xspace}

\newcommand{\FPT}{\ensuremath{\mathtt{FPT}}\xspace}
\newcommand{\XP}{\ensuremath{\mathtt{XP}}\xspace}

\newcommand{\poly}{\ensuremath{\mathtt{poly}}\xspace}

\newcommand{\bigoh}{\mathcal{O}}

    \newcommand{\tw}{\ensuremath{\mathsf{tw}}\xspace}
    \newcommand{\pw}{\ensuremath{\mathsf{pw}}\xspace}
    \newcommand{\fvn}{\ensuremath{\mathsf{fvn}}\xspace}
    
    \newcommand{\vc}{\ensuremath{\mathsf{vcn}}\xspace}

    \newcommand{\vcdimension}{\ensuremath{\mathsf{VC}}-dimension\xspace}
    \newcommand{\Copen}[1]{\ensuremath{\mathcal C^{open}(#1)}\xspace}
    \newcommand{\Cclosed}[1]{\ensuremath{\mathcal C^{closed}(#1)}\xspace}

\newcommand{\clique}{\textup{\textsc{Clique}}\xspace}
\newcommand{\MCclique}{\textup{\textsc{Multi-Colored Clique}}\xspace}

\usepackage{tabularx}
\newcommand{\problemtitle}[1]{\gdef\@problemtitle{#1}}
\newcommand{\probleminput}[1]{\gdef\@probleminput{#1}}
\newcommand{\problemquestion}[1]{\gdef\@problemquestion{#1}}
\newif\iflong
\newif\ifshort
\longtrue

\iflong
\else
\shorttrue
\fi

\usepackage{todonotes}

\newcommand{\vcdim}{\ensuremath{\operatorname{\mathsf{VC}\text{-}dim}}\xspace}

\newenvironment{tightcenter}
 {\parskip=0pt\par\nopagebreak\centering}
 {\par\noindent\ignorespacesafterend}

\usepackage{tikz}
\usetikzlibrary{calc}
\usepackage{xargs}
\usepackage{xifthen}
\usepackage{framed}

\usepackage{ctable}

\newlength{\RoundedBoxWidth}
\newsavebox{\GrayRoundedBox}
\newenvironment{GrayBox}[1]%
   {\setlength{\RoundedBoxWidth}{\textwidth-10.5ex}
    \def\boxheading{#1}
    \begin{lrbox}{\GrayRoundedBox}
       \begin{minipage}{\RoundedBoxWidth}%
   }{%
       \end{minipage}
    \end{lrbox}%
    \begin{tightcenter}%
    \begin{tikzpicture}%
       \node(Text)[draw=black!90,fill=white,rounded corners,%
             inner sep=2ex,text width=\RoundedBoxWidth]%
             {\usebox{\GrayRoundedBox}};
        \coordinate(x) at (current bounding box.north west);
        \node [draw=white,rectangle,inner sep=3pt,anchor=north west,fill=white] 
        at ($(x)+(6pt,.75em)$) {\boxheading};
    \end{tikzpicture}
    \end{tightcenter}\vspace{0pt}%
    \ignorespacesafterend
}

\begin{document}

\maketitle

\begin{abstract}
We study the parameterized complexity of maximum temporal connected components (tccs) in temporal graphs, that is, graphs whose edges are available only at specific points in time.
In a tcc, every pair of vertices must be able to reach one another via time-respecting paths.
We consider both maximum {\em open}~tccs (\openTCC), which allow temporal paths through vertices outside the component, and {\em closed}~tccs (\closedTCC), which require at least one temporal path entirely within the component for every pair of vertices.
We perform a comprehensive study of the \openTCC and \closedTCC problems with respect to both structural parameters (treewidth, pathwidth, vertex cover number) and a temporal parameter
(temporal path number).
We show that the exact complexity, i.e., \paraNP-hardness vs \XP-tractability, depends on both whether we seek an open or closed tcc and on whether the temporal graph is directed or not. 
Vertex cover number suffices for \XP algorithms for both \openTCC and \closedTCC on undirected temporal graphs only, while temporal path number suffices only for \openTCC in both directed and undirected temporal graphs.
Our results are tight: every \XP algorithm is complemented by a matching \Wone-hardness result, and for every other case we prove \NP-hardness for small constant values of the parameters even on planar graphs.
Finally, we prove that both problems become fixed-parameter tractable on both directed and undirected graphs when parameterized by treewidth and temporal path number together.

\end{abstract}

\newpage
\section{Introduction}

Connected components are one of the fundamental notions of graph theory. Their temporal counterparts were introduced by Bhadra and Ferreira \cite{bhadra_ComputingMulticast_2002,bhadra_ComplexityConnected_2003}, who distinguished between open and closed temporal connected components (tccs). In both variants, every pair of vertices must be mutually reachable via temporal paths; for a closed tcc, these paths must remain within the component, whereas for an open tcc they may use vertices outside of it. A fundamental difference from the static setting is that temporal reachability is not transitive: temporal paths from \(u\) to \(v\) and from \(v\) to \(w\) cannot necessarily be concatenated into a temporal path from \(u\) to \(w\). As a result, temporally connected components do not necessarily partition the vertex set, and computing a component of maximum size is NP-hard~\cite{bhadra_ComputingMulticast_2002,bhadra_ComplexityConnected_2003}.

This hardness has motivated the study of tccs from the perspective of
parameterized complexity. Casteigts \cite{casteigts_FindingStructure_2018} established \Wone-hardness when
parameterized by the solution size, and this was subsequently strengthened to NP-hardness even for closed components of size three
\cite{casteigts_hardnessfinding_2026}. Costa, Lopes, Marino, and Silva \cite{costa_ComputingLarge_2023} studied parameterizations by lifetime, solution size, and their combination, obtaining an \FPT algorithm only for the combined parameter on undirected non-strict temporal graphs. Casteigts,
Morawietz, and Wolf introduced the notion of a \emph{transitivity modulator} and obtained \FPT algorithms for \textup{\textsc{openTCC}} parameterized by its size \cite{casteigts_DistanceTransitivity_2024}. This parameter, however, measures how far the graph is from one in which temporal reachability is fully transitive and therefore controls the problem by bounding precisely the non-transitivity that
distinguishes temporal reachability from its static counterpart. Moreover, \textup{\textsc{closedTCC}} remains NP-hard even when the parameter is one \cite[Thm. 18]{casteigts_DistanceTransitivity_2024}.

These results leave open how broadly tractability can be obtained without crudely eliminating the nature of temporal transitivity. From a theoretical perspective, we seek the largest graph classes on which the problem remains tractable, or, equivalently, the least restrictive parameterizations that yield efficient algorithms.
At the same time, we are particularly interested in parameters that can reasonably be expected to be small in real-world temporal graphs. This leads to our central question: \textbf{Which structural or temporal restrictions are sufficient to make
the computation of maximum temporal connected components tractable?}

\subsection*{Overview of Our Results}\addcontentsline{toc}{subsection}{Overview of Our Results}

\begin{figure}[t]
    \centering
    \input{Figures/results-detailed10-overlay}
    \caption{Overview of the parameterized complexity of \textup{\textsc{Maximum}} \openTCC and \closedTCC in directed and undirected temporal graphs. With respect to the corresponding parameter, a red solid box denotes \paraNP-hardness; a blue dotted box denotes \Wone-hardness and \XP; a green dashed box denotes \FPT. Each box consists of four parts depicting the four combinations of open vs closed and directed vs undirected as illustrated in the top right.
    The label $\tw+\mathsf{x}$ represents the three results for $\mathsf{x}\in\{\tpn, \lifetime,\tdegree\}$.}
    \label{fig:result-diagram}
\end{figure}

We study the parameterized complexity of computing maximum open and closed temporal connected components with respect to two families of parameters: \begin{enumerate}
    \item classical structural parameters treewidth \tw, pathwidth~\pw, feedback vertex number~\fvn, and vertex cover number~\vc, all measured on the underlying undirected footprint;
    \item temporal path number \tpn and temporal walk number \twn, which 
    capture the structure of the footprint and the temporal arrangement of the edges simultaneously.
\end{enumerate}
We also investigate the effect of restricting the footprint to planar graphs. 
An overview of the parameterized complexity landscape is illustrated in \Cref{fig:result-diagram}.

We consider directed and undirected temporal graphs under both strict and non-strict reachability.
For the hardness results, we aim for constructions that have strong restrictions on the temporal labeling: A temporal graph is \textit{simple} if each edge appears at only one time and \textit{proper} if incident edges have distinct time labels. Properness makes strict and non-strict reachability coincide, and temporal graphs that are both simple and proper\footnote{Simple and proper temporal graphs are also referred to as \textit{happy}~\cite{casteigts_SimpleStrict_2024}.} form a particularly restricted temporal graph class~\cite{doring_SimpleStrict_2025}.

We begin our investigation with the footprint parameters. Both problems remain \NP-hard even when treewidth, pathwidth, and feedback vertex number are bounded by small constants and the footprint is planar. On directed temporal graphs, hardness also persists at constant vertex cover number.

\begin{result} \label{res:planar-hardness}
    Both \openTCC and \closedTCC are \NP-hard on directed and undirected proper temporal graphs even when the footprint $G$ is planar and $\tw(G)=2$, $\pw(G)=3$, and $\fvn(G)\leq2$.
\end{result}

\begin{result} \label{res:directed-vc-four}
    Both \openTCC and \closedTCC are \NP-hard on directed temporal graphs even when $\vc(G)=4$ and \gcal is simple and proper.
\end{result}
   
For undirected temporal graphs, the vertex cover number instead yields XP algorithms for both problems. Our algorithm for \textup{\textsc{openTCC}} bounds the VC-dimension of the family of maximal open temporal connected components. Then, we use an enumeration strategy to list all maximal connected components.

\begin{result}\label{res:undirected-vc-xp-open}
    \openTCC on undirected, (non-)strict temporal graphs with vertex cover number $\vc(G)=k$ can be solved in time $n^{2^{\mathcal O(k\log k)}}$.
\end{result}

Unfortunately, closed temporal connectivity is not hereditary: deleting vertices can destroy the paths witnessing connectivity among those that remain. Thus, the previous enumeration strategy does not directly extend. Instead, our algorithm for \textup{\textsc{closedTCC}} uses temporal out-tree certificates. This idea was introduced by Casteigts, Döring, Molter, and Zehavi \cite{casteigts2026minimumtemporalspannershappy}, who used out-tree certificates to find minimum temporal spanners in happy graphs when parameterizing by vertex cover. We generalize their algorithm to undirected temporal graphs under both strict and non-strict semantics, thereby answering one of their open questions, and also extend it to compute closed connected components.

\begin{result}\label{res:undirected-vc-xp-closed}
    \closedTCC on undirected temporal graphs with vertex cover
    number $\vc(G)=k$ can be solved in time
    $M^{\mathcal O(k^2)}(n+M)^{\mathcal O(1)}$ and
    $M^{2^{\mathcal O(k\log k)}}(n+M)^{\mathcal O(1)}$
    in the non-strict and strict settings, respectively.
\end{result}

Finally, we show that \openTCC and \closedTCC in undirected graphs are \Wone-hard when parameterized by vertex cover. For that we show, in the appendix, that \textup{\textsc{Mixed 2-Sided Orthant Depth (M2SOD)}} is \Wone-hard when parameterized by the dimension (see \Cref{thm:M2SOD-wone-hard}) and then reduce to \openTCC and \closedTCC. 

\begin{result}\label{res:vc-whard}
    \openTCC and \closedTCC in undirected proper simple temporal graphs are \Wone-hard when parameterized by $\vc(G)$.
\end{result}

In total, our results almost tightly resolve the parameterized landscape of the problems with respect to these footprint parameters. They also establish a sharp threshold at treewidth two: both problems are \NP-hard already on temporal graphs whose footprint has treewidth two, whereas they are polynomial-time solvable when the footprint has treewidth at most one \cite[Section~4]{jarry_ConnectivityEvolving_2004}. Among the footprint parameters considered, only the vertex cover number admits XP algorithms, and only for undirected temporal graphs, while treewidth, pathwidth, and feedback vertex number leave both problems \NP-hard even for constant values. However, the class of undirected temporal graphs with constant vertex cover is very restricted.

These limitations motivate the combination of structural and temporal restrictions. Unsurprisingly, we show using MSO formulations that both problems are \FPT when parameterized by treewidth together with either the lifetime~\lifetime, the largest time label, or the maximum temporal total degree~\tdegree, the maximum number of temporal edges incident with a vertex. However, as with the vertex cover number, these temporal parameters tend to be large in practice. Therefore, the structural insights we gain from these positive results are unlikely to be useful in practice.

This motivates considering a different way of restricting temporal structure: temporal path number, implicitly introduced by Deligkas et al.~\cite{deligkas_HowMany_2025}. This parameter is the minimum number of temporal paths whose temporal-edge sets partition the graph's temporal edges. We call the paths in such a partition \emph{base-paths}. Allowing walks instead of paths gives the \emph{temporal walk number}~\twn, with each temporal edge occurring exactly once across the resulting \emph{base-walks}. Since every temporal path is a temporal walk, $\twn\leq\tpn$. 
These parameters capture both the footprint structure and the temporal arrangement of its edges, and naturally describe networks organized around a small number of routes, such as scheduled journeys in a transportation network.
Although this abstraction does not capture all operational constraints, such as the maximum number of trains that can occupy a station simultaneously, we obtain an XP algorithm for \openTCC even without imposing additional restrictions.

\begin{result}\label{res:open-tpn-xp}
    \openTCC on (un)directed, (non-)strict temporal graphs can be solved in time $n^{2\twn+\mathcal{O}(1)}$, and consequently in time $n^{2\tpn+\mathcal{O}(1)}$.
\end{result}

As in the vertex cover case, our algorithm bounds the VC-dimension of the family of maximal open temporal connected components, this time in terms of temporal walk number and consequently temporal path number. The structural argument, however, is different: a decomposition into a bounded number of temporal walks imposes a form of partial transitivity. Whenever two base-walks pass through a common internal vertex, the walk arriving earlier can be followed by the continuation of the other. Consequently, every vertex preceding the meeting point on the first walk can reach every vertex succeeding it on the second. Since at most $2\twn$ vertices are endpoints of base-walks, this switching property applies wherever base-walks meet outside a bounded set of exceptional vertices. These constraints on reachability allow us to bound the VC-dimension.

This property between base-walks has a natural interpretation in transportation networks, where temporal walks represent scheduled journeys. Passengers can travel along a journey and transfer to another at a shared station whenever the arrival and departure times permit. Such transfers induce reachability from earlier stops on the first journey to later stops on the second, constraining the possible reachability patterns. Bounded temporal walk number, and consequently bounded temporal path number, captures networks whose interactions are organized into a small number of such journeys, providing a concrete modeling motivation for the partial transitivity exploited by our algorithm.

We complement this algorithm with a \Wone-hardness result for \otcc, parameterized by temporal path number and consequently temporal walk number, on directed and undirected proper simple temporal graphs, thereby ruling out an \FPT algorithm unless \FPT = \Wone. To obtain this result, we reduce again from \textup{\textsc{M2SOD}} parameterized by the dimension. In contrast, \ctcc remains \NP-hard even with a constant number of temporal paths, and consequently temporal walks.
		
\begin{result}\label{res:open-tpn-wone-hard}
    \openTCC is \Wone-hard when parameterized by temporal path number, and consequently temporal walk number, on directed and undirected proper simple temporal graphs.
\end{result}

\begin{result}\label{res:closed-tpn-paranp-hard}
    \closedTCC on (un)directed, (non-)strict temporal graphs is \NP-hard even on graphs with $\tpn=6$, and consequently with $\twn\leq6$.
\end{result}

\subsection*{Related Work}\addcontentsline{toc}{subsection}{Related Work}\label{sec:related-work-short}
 The study of connected components in temporal graphs dates back at least to Bhadra and Ferreira~\cite{bhadra_ComputingMulticast_2002,bhadra_ComplexityConnected_2003}, who introduced open and closed tccs as the natural extension of static connected components and proved that computing a maximum open or closed tcc is \NP-hard.
Jarry and Lotker~\cite{jarry_ConnectivityEvolving_2004} later showed that both variants are polynomial-time solvable on undirected trees.
Subsequent empirical and metric works explored the size of tccs in human contact and social network data~\cite{tang_CharacterisingTemporal_2010,nicosia_ComponentsTimeVarying_2012,nicosia_GraphMetrics_2013}.
Casteigts~\cite{casteigts_FindingStructure_2018} established \Wone-hardness by solution size on strict temporal graphs.
Costa, Lopes, Marino, and Silva~\cite{costa_ComputingLarge_2023} studied parameterizations by lifetime, component size, and their combination, obtaining an \FPT algorithm for the combined parameter on undirected non-strict graphs.
Casteigts, Morawietz, and Wolf~\cite{casteigts_DistanceTransitivity_2024} obtained \FPT algorithms for \openTCC parameterized by distance to transitivity, while \closedTCC remains \NP-hard already for distance one.
These parameters measure the number of modifications required to make the reachability graph transitive.
Temporal path number is incomparable with these transitivity modulator parameters.
Beyond open and closed tccs, several path-based and snapshot-based variants of temporal connectivity have been studied. For a concise overview of the related literature, we refer to~\cite{doring_TemporalConnected_2025}.

A variety of structural parameters have been considered for other problems on temporal graphs \cite{enright_structuralparameters_2024,fluschnik_timegoes_2020,enright_familiestractable_2025,casteigts2026minimumtemporalspannershappy}.
Parameterized results for treewidth combined with either lifetime or temporal degree have been obtained via the MSO approach~\cite{enright_DeletingEdges_2021,haag_FeedbackEdge_2022,doring_fomso_2026}.

\section{Technical Overview}\label{sec:technical-overview}
\addtocontents{toc}{\protect\setcounter{tocdepth}{1}}
    We group our results by technique to highlight the common ideas behind the constructions and algorithms, and explain how each technique adapts to different parameters.

A \textit{temporal graph} $\gcal=(V,E,\lambda)$ consists of a static graph $G=(V,E)$, called the \textit{footprint}, and a labeling function $\lambda$ assigning a nonempty finite set of positive integer time labels to each edge.
A pair $(e,t)$ with $t\in\lambda(e)$ is a \textit{temporal edge}; we denote their set by $\ecal$ and write $n=|V|$ and $M=|\ecal|$.
A \textit{temporal path} is a sequence of temporal edges whose edges form a path in the footprint and whose time labels are non-decreasing, or strictly increasing for \textit{strict} paths.
A \textit{temporal walk} may revisit vertices.
We write $u\rightreach v$ if $u$ reaches $v$ via a temporal path.
Two vertices are \emph{mutually reachable} if each reaches the other, and a set is \emph{temporally connected} if any two of its vertices are mutually reachable.
An \emph{open temporal connected component} is an inclusion-maximal temporally connected set $C\subseteq V$; a \emph{closed temporal connected component} is an inclusion-maximal set $C$ for which $\gcal[C]$ is temporally connected.
Thus, for closed components the paths witnessing connectivity must remain within $C$, whereas for open components they may use vertices outside of $C$.
The problems \openTCC and \closedTCC ask for a component of maximum size.

\subsection{\XP algorithm for \openTCC via \vcdimension}
    A temporal graph can have exponentially many open tccs~\cite{balev_temporallyconnected_2024}. We bound their number polynomially for undirected graphs of constant vertex cover number and for directed or undirected graphs of constant temporal walk number. Since $\twn\leq\tpn$, the latter also bounds their number for constant temporal path number. We view these components as a set system and bound its \vcdimension, which measures how freely its members can intersect a subset of vertices.
    \begin{definition}[\vcdimension~\cite{vapnik_UniformConvergence_2015}]
        Let $\mathcal{F}$ be a family of subsets of a universe $U$.
        A set $A\subseteq U$ is \emph{shattered} by $\mathcal{F}$ if, for every subset $S\subseteq A$, there exists a set $F\in\mathcal{F}$ such that $F\cap A=S$.
        The \emph{\vcdimension} of $\mathcal{F}$, denoted by $\vcdim(\mathcal{F})$, is the maximum size of a set shattered by~$\mathcal{F}$.
    \end{definition}
    \begin{lemma}[Sauer--Shelah--Perles Lemma~\cite{sauer_DensityFamilies_1972,shelah_CombinatorialProblem_1972}]
        \label{lem:overview-sauer-shelah-perles}
        Let $\mathcal{F}$ be a set system over a universe of size $n$ with \vcdimension at most $d$.
        Then $|\mathcal{F}|\leq\sum_{i=0}^{d}\binom{n}{i}=\mathcal{O}(n^d)$.
    \end{lemma}
    Let $\Copen{\gcal}=\{C\subseteq V \colon C\text{ is a maximal open tcc}\}$. This is precisely the family of maximal cliques in the mutual-reachability graph, whose edges join mutually reachable vertices.
    By the Sauer--Shelah lemma, $\vcdim\bigl(\Copen{\gcal}\bigr)\leq d$ implies that $\lvert\Copen{\gcal}\rvert\leq n^{d+\mathcal{O}(1)}$. Since all maximal cliques of a graph can be enumerated with polynomial delay, we can enumerate all members of $\Copen{\gcal}$ in total time $n^{d+\mathcal{O}(1)}$. Returning a largest enumerated set solves \openTCC, so it remains to bound the \vcdimension.

    The following lemma identifies a reachability configuration forced by large VC-dimension of open tccs. 
    See \Cref{fig:overview-vcdim-structure} for an illustration of the configuration.
    \begin{lemma}[Blocking set configuration] \label{lem:overview-vcdim-classification-open}
        If $\vcdim\left(\Copen{\gcal}\right)\geq 2d+2$, then there exist a \emph{shattered set} $A := \{a_1, \dots, a_{d+1}\} \subseteq V$ and a \emph{blocking set} $B := \{b_1, \dots, b_{d+1}\} \subseteq V \setminus A$ such that the following hold for all distinct $i,j \in [d+1]$:
            (1) $a_i \rightreach a_j$ and $a_i \leftreach a_j$, (2) $b_i \rightreach a_j$ and $b_i \leftreach a_j$, and (3) $b_i \not\rightreach a_i$.
    \end{lemma}
    \begin{figure}[t]
        \centering
        \includegraphics[width=0.5\linewidth]{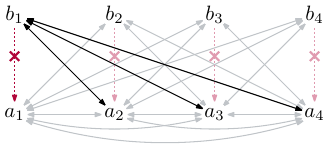}
        \caption{Sketch of a shattered set of size $4$. {Any two vertices of $\{a_1,\dots,a_4\}$ are mutually reachable. Each $b_i$ is mutually reachable with every $a_j$ for $j\neq i$, but not with $a_i$.} For visibility, the relations of $b_1$ are highlighted.}
        \label{fig:overview-vcdim-structure}
    \end{figure}
    \begin{proof}
        Let $\mathcal{G}=(V,E,\lambda)$ be a temporal graph with $\vcdim(\Copen{\gcal})\geq 2d+2$.
        Then there must exist a set $A \subseteq V$ of size $2d+2$ that is \emph{shattered} by $\Copen{\gcal}$. That is, for every subset $S \in 2^A$ there exists an open tcc $C(S) \in \Copen{\gcal}$ such that $C(S) \cap A = S$.
        Let $A = \{a_1, \dots, a_{2d+2}\}$. Since $A$ is shattered by $\Copen{\gcal}$, there exists $C(A)\in\Copen{\gcal}$ with $A\subseteq C(A)$. Consequently, $A$ is temporally connected: for all distinct $i,j \in [2d+2]$, we have $a_i \rightreach a_j$ and $a_i \leftreach a_j$.
        Again by definition of a shattered set, for each $a_i \in A$, the set $A \setminus \{a_i\}=:A_{-i}$ must be contained in some $C(A_{-i}) \in \Copen{\gcal}$. Since $a_i\notin C(A_{-i})$ but {$a_i$ is mutually reachable with every vertex of $A_{-i}$}, there must exist a \emph{blocking} vertex $b_i\in C(A_{-i})\setminus A$ which is not {mutually reachable} with $a_i$, \ie
        \[
            b_i \not\rightreach a_i\text{ or }b_i \not\leftreach a_i\qquad \text{and} \qquad b_i \rightreach a_j\text{ and }b_i \leftreach a_j\text{ for all }i \neq j. 
        \]
        It follows that $b_i\notin A$, since {any two vertices of $A$ are mutually reachable}. The blocking vertices are pairwise distinct: for $i\neq j$, $b_i$ is mutually reachable with $a_j$, whereas $b_j$ is not.

        Overall, this implies the existence of a set $B = \{b_1, \dots, b_{2d+2}\} \subseteq V \setminus A$ of blocking vertices such that each $b_i$ is {not mutually reachable} with $a_i$ but {mutually reachable} with every $a_j$ for $i \neq j$. {Since, for each $i$, at least one of $b_i\not\rightreach a_i$ and $a_i\not\rightreach b_i$ holds} and $\lvert A\rvert=2d+2$, the pigeonhole principle implies a subset $A' \subseteq A$ of size at least $d+1$ for which {the missing reachabilities have the same direction}. Without loss of generality, we can assume it is $b_i\not\curvearrowleft a_i$. This proves the claim.
    \end{proof}

    A bounded temporal walk number restricts which blocking configurations can occur. 
    Although temporal reachability is not transitive in general, the temporal order along each base-walk imposes a form of partial transitivity: Whenever two walks cross in a vertex, every vertex on the walk prefix that arrives first can switch to the other base-walk an reach the remaining walk suffix. With only a bounded number of base-walks, these forced reachabilities rule out large blocking configurations and hence bound the VC-dimension.
    Importantly, the argument does not need to compute a decomposition into base-walks. Its existence alone bounds the family being enumerated. 
    \begin{restatable}{theorem}{openWalkXP}\label{thm:full-otcc-is-XP-twn}
        \openTCC on directed and undirected temporal graphs is solvable in time $n^{2\twn+\mathcal O(1)}$, and hence in time $n^{2\tpn+\mathcal O(1)}$, under both strict and non-strict reachability.
    \end{restatable}
    \begin{proof}[Proof sketch]
        Let $k=\twn(\gcal)$ and fix a decomposition of $\gcal$ into $k$ temporal base-walks. Assume towards contradiction that $\vcdim(\Copen{\gcal})\geq 2k+2$. Then by \Cref{lem:overview-vcdim-classification-open}, there exist vertices $a_1,\dots,a_{k+1}$ and blocking vertices $b_1,\dots,b_{k+1}$ such that each $b_i$ reaches every $a_j$ with $j\neq i$, but does not reach $a_i$.

        Each base-walk containing a blocking vertex has a first such vertex. Since there are $k$ base-walks but $k+1$ blocking vertices, some $b_i$ is never first. Thus, on every base-walk containing $b_i$, some other blocking vertex occurs before the first occurrence of $b_i$.
        
        Now, consider a base-walk $W$ that has an earliest incoming edge at $b_i$, and let $b_j$ be the first blocking vertex of $W$. Since $i\neq j$, the blocking configuration requires a temporal path from $b_i$ to $a_j$. 
        By our choice of $W$, it arrives before any temporal path leaves $b_i$. Consequently, $b_j$ can follow the base-walk $W$ to $b_i$ in time to then follow the path from $b_i$ to $a_j$. This gives $b_j\rightreach a_j$, contradicting the blocking configuration.
        Consequently, the VC-dimension is at most $2\twn+1$, and the common enumeration framework yields the stated running time.
        Since $\twn\leq\tpn$, the temporal path-number bound follows as well.
    \end{proof}
    \begin{figure}[t]
        \centering
        \begin{minipage}[c]{0.54\linewidth}
            \caption{The transfer argument used for the temporal walk number. An earlier blocking vertex $b_j$ follows the orange route to $b_i$ and then follows the path by which $b_i$ must reach $a_j$. The timing permits this transfer, creating the forbidden reachability from $b_j$ to $a_j$.}
            \label{fig:overview-walk-switch}
        \end{minipage}\hfill
        \begin{minipage}[c]{0.4\linewidth}
            \centering
            \includegraphics[width=0.9\linewidth]{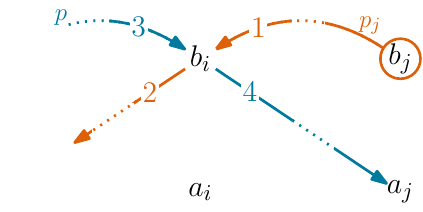}
        \end{minipage}
    \end{figure}
            
    For undirected graphs, a bounded vertex cover also rules out large blocking configurations. The key is to classify vertices outside the cover by the relative order of their latest edges to cover neighbors. Only boundedly many types arise, and repeated types force reachabilities that contradict the blocking configuration. See \Cref{fig:overview-vc-open} for an illustration of the argument.
    \begin{restatable}{theorem}{openVCXP}\label{thm:undirected-vc-xp-open}
        \openTCC on undirected temporal graphs with vertex cover number $\vc$ is solvable in time $n^{2^{\mathcal O(\vc\log\vc)}}$ under both strict and non-strict reachability.
    \end{restatable}
    \begin{proof}[Proof sketch]
        Let $k=\vc(\gcal)$ and let $Z$ be a vertex cover of size $k$. For an independent vertex $x$ and a cover neighbor $z$, let $\ell^+(x,z)$ be the latest time label of the edge $xz$. The \textit{type} of $x$ records its cover neighbors and their weak order by $\ell^+(x,z)$, including ties, but not the actual time labels. There are only $2^{\mathcal O(k\log k)}$ such types. 
        After discarding the at most $k$ blocking pairs involving a cover vertex, a sufficiently large blocking configuration therefore contains two vertices $a_1,a_2$ of the same type.

        The configuration now requires temporal paths from $b_1$ to $a_2$ and from $b_2$ to $a_1$, which enter their target through some last cover vertices $z_1$ and $z_2$, respectively. The forbidden reachabilities $b_1\not\rightreach a_1$ and $b_2\not\rightreach a_2$ force $\ell^+(a_1,z_1)<\ell^+(a_2,z_1)$ and $\ell^+(a_2,z_2)<\ell^+(a_1,z_2)$.
        Since $a_1$ and $a_2$ have the same weak order, $z_1$ and $z_2$ cannot be tied, since the two inequalities would then contradict one another. After exchanging the indices if necessary, the common order gives $\ell^+(a_2,z_1)<\ell^+(a_2,z_2)<\ell^+(a_1,z_2)$.
        Thus, $b_1$ can follow its path to $a_2$ and continue through $z_2$ to $a_1$, producing the forbidden reachability. The strict inequalities make this continuation valid under both strict and non-strict reachability. Consequently, the VC-dimension is at most $2^{\mathcal O(k\log k)}$, and the common enumeration framework yields the stated running time.
    \end{proof}
    \begin{figure}[ht]
        \centering
        \includegraphics[width=0.8\linewidth]{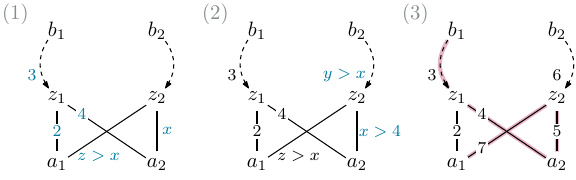}
        \caption{The path argument used for the vertex cover number. The required paths from $b_1$ to $a_2$ and from $b_2$ to $a_1$, together with the common latest-edge type of $a_1,a_2$, force a temporal walk from $b_1$ through $a_2$ to $a_1$.}
        \label{fig:overview-vc-open}
    \end{figure}

\subsection[\XP algorithm for \closedTCC via temporal tree certificates]{\XP algorithm for \closedTCC wrt \vc in undirected graphs via temporal tree certificates}
    Since closed temporal connected components are not hereditary, the enumeration framework used for computing maximum open components does not work for closed components. We instead provide an \XP algorithm that uses a combinatorial argument.
    \begin{restatable}{theorem}{closedVCXP}\label{thm:undirected-vc-xp-closed}
        \closedTCC on undirected temporal graphs with vertex cover
        number $\vc$ can be solved in time
        $M^{\mathcal O(\vc^2)}(n+M)^{\mathcal O(1)}$ and
        $M^{2^{\mathcal O(\vc\log\vc)}}(n+M)^{\mathcal O(1)}$
        in the non-strict and strict settings, respectively.
    \end{restatable}
    The strict-case algorithm is more general and also applies to non-strict graphs. We provide a separate algorithm for non-strict graphs, because it is simpler, has a better running time, and provides useful intuition for the more involved strict case.
    Both algorithms build on the out-tree-guessing approach used by Casteigts, Döring, Molter, and Zehavi for an \XP algorithm wrt \vc for \textup{\textsc{Minimum Temporal Spanner}} on simple and proper temporal graphs~\cite{casteigts2026minimumtemporalspannershappy}. 
    We adapt their representative-tree argument to certificates for closed temporal connected vertex sets and extend it to arbitrary undirected temporal graphs.

    The goal is to describe a closed temporally connected set by a collection of temporal spanning trees whose size is bounded in terms of $\vc(G)$.
    The argument has four steps:
    \begin{enumerate}
        \item \textit{Witnessing complexity.} For every vertex $v$ of a closed temporally connected set, temporal reachability to all other vertices in the set can be witnessed by a temporal spanning tree rooted at $v$, which is called its \textit{temporal out-tree}. 
        \item \textit{Bounding the size of trees.} After deleting all independent-set leaves, each such temporal out-tree contains at most $2\vc(G)$ vertices. We call such a tree \textit{pruned}.
        \item \textit{Bounding the number of trees.} 
        Most independent vertices can \textit{delegate} their out-tree to a bounded collection of representatives, and we only need to store a few incident edges through which it enters the representing out-tree and then reaches the entire set.
        \item \textit{Guessing and completing the cores.} The algorithm guesses the pruned out-trees of a guessed set of representatives and then adds every independent vertex $x$ that satisfies:
        \begin{enumerate}
            \item for every guessed tree, $x$ is either an internal vertex or can be attached as a leaf, and
            \item there exists a guessed tree such that $x$ is either the representative of that tree or can reach that representative and thus delegate to the representative.
        \end{enumerate}   
    \end{enumerate}
    The cases of strict and non-strict temporal graphs differ in the details of the third step.
    In the non-strict case, one representative tree per cover vertex suffices. We group the independent roots according to the cover vertex $a$ reached by the earliest edge in the tree, and choose a representative whose earliest edge has the latest time label. Every other root in the group reaches $a$ in time to traverse this edge to the representative and then follow its out-tree. This is possible because the edges are undirected and non-strict reachability permits consecutive edges with equal time labels. See \Cref{fig:o-vc-XP-NS-outtree-delegation}, where $y$ is chosen as the representative.

    In the strict case, that non-strict delegation rule cannot apply because the edges incident to the root could have the same time label.
    Instead, we classify each independent vertex $x$ by its \textit{weak-order type}, recording which cover vertices are children of $x$ in its out-tree and the weak order of their incident tree-edge labels. For $k=\vc(G)$, there are at most $(k+1)^k=2^{\mathcal O(k\log k)}$ such types.

    Among all vertices with the same type, we compare their label vectors in the order prescribed by the type. For two distinct independent vertices, we observe that at the first position where their label vectors differ, the root with the earlier label can delegate to the one with the later label if we store its own incident tree edges up to this point. Refer to \Cref{fig:o-vc-XP-strict-outtree-delegation} for an illustration. Consequently, we can delegate all trees apart from the ones rooted at vertices with lexicographically largest vectors. If there is more than one such root, we choose two representatives arbitrarily: replacing the first edge of a representative's path gives reachability to every vertex except possibly the representative root itself; the second tree supplies a path to this root.
    
        \begin{figure}[t]
            \centering
            \begin{minipage}[b]{0.48\linewidth}
                \centering
                \includegraphics[width=0.9\linewidth]{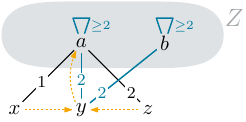}
                \caption{Non-strict out-tree delegation. Independent vertices $x$ and $z$ delegate their out-trees to a latest independent vertex $y$ (dotted arrows). The representative tree (blue) could be rerooted at the neighboring cover vertex $a$.}
                \label{fig:o-vc-XP-NS-outtree-delegation}
            \end{minipage}\hfill
            \begin{minipage}[b]{0.48\linewidth}
                \centering
                \includegraphics[width=0.9\linewidth]{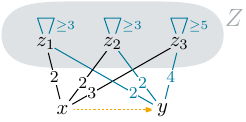}
                \caption{Strict out-tree delegation. Independent vertex $x$ can delegate to $y$ (dotted arrow). The representative tree cannot be rerooted at a neighboring cover vertex because of the two time-2 edges of $y$.}
                \label{fig:o-vc-XP-strict-outtree-delegation}
            \end{minipage}
        \end{figure}

\subsection{\paraNP-hardness for \closedTCC wrt \tpn and \twn in directed and undirected graphs from \clique using bridges}
    While bounded vertex cover yields \XP algorithms for both open and closed components, bounded temporal path number does not extend the tractability of open components to closed ones. The requirement that witnessing paths must stay inside the component is crucial: we can join the temporal edges of the classical proper and simple hardness construction into a constant number of base-paths through vertices that no nontrivial closed component can contain. These connections keep the parameter small without providing additional paths inside a closed solution.
    \begin{restatable}{theorem}{closedPathHardness}\label{thm:closed-tpn-paranp-hard}
        \closedTCC on directed and undirected temporal graphs is \NP-hard even when $\tpn=6$, and consequently when $\twn\leq6$, under both strict and non-strict reachability.
    \end{restatable}
    \begin{proof}[Proof sketch]
        We reduce from \clique, using temporal paths to encode adjacency and special bridge vertices to keep the temporal path number bounded without creating additional paths inside a closed solution.
        The classical hardness reduction for \openTCC (see \Cref{fig:overview-hardness-opentcc}) represents each edge of the source graph by temporal paths of length two, with the first edges appearing before all second edges. This encodes adjacency as mutual reachability between the represented vertices. 
        \begin{figure}[ht]
            \centering
            \includegraphics[width=0.75\linewidth]{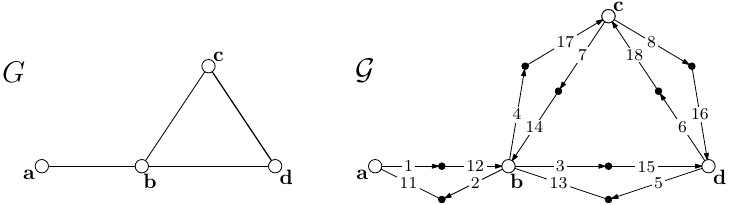}
            \caption{Illustration of the classical \Wone-hardness reduction for \openTCC with the \clique instance $G$ on the left and the hardness construction on the right.}
            \label{fig:overview-hardness-opentcc}
        \end{figure}
        We adjust this construction so that closed and open tccs coincide. Then we connect the individual temporal edges into temporal base-paths. To avoid unwanted reachabilities by chaining the edges in a temporal path, we insert special pairs of vertices that can never belong to a nontrivial \ctcc; we call them \emph{bridges}.
        Intuitively, a bridge creates a temporal ``bottleneck'': it can be crossed in one direction only, so 
        \begin{figure}[t]
            \centering
            \begin{minipage}[c]{0.44\linewidth}
                \caption{An example of a bridge $(\beta_1,\beta_2)$ containing the bridge-vertices $\beta_1$ and $\beta_2$ with $t_1<t_2<t_3$.}
                \label{fig:overview-bridge-vertices}
            \end{minipage}\hfill
            \begin{minipage}[c]{0.50\linewidth}
                \centering
                \includegraphics[width=\linewidth]{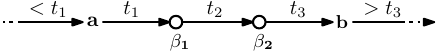}
            \end{minipage}
        \end{figure}
        no bridge vertex can belong to a nontrivial closed component. Thus, a closed solution cannot use these connecting routes as internal witnesses, and the intended adjacency tests are preserved. 
        These chained routes may revisit vertices. To obtain temporal paths, we replace each original vertex by a gadget with one vertex per incident temporal edge and add auxiliary paths to preserve the tccs. Since $\twn\leq\tpn$, the construction with temporal path number six also establishes hardness for temporal walk number at most six.
    \end{proof}

\subsection{Hardness results via rank comparisons}
    Every one of our hardness proofs exploits a basic obstruction that seems central to the difficulty of computing temporal connected components:
    \begin{quote}
        Given multiple orderings of $V$, we can encode comparisons between the ranks in each of these orderings using temporal paths, even under strong structural restrictions.
    Let~$<_1$ and $<_2$ be two orderings with $x<_1y$ and $y<_2z$ but $x\not<_1z$ and $x\not<_2z$. Since temporal reachability is non-transitive, it is possible to construct a temporal graph with temporal paths $x\rightsquigarrow y$ and $y\rightsquigarrow z$, but with $x\not\rightsquigarrow z$. That is, $x\rightsquigarrow y$ arrives at $y$ strictly later than $y\rightsquigarrow z$ departs from $y$. Such a construction is impossible in static graphs.
    \end{quote}
    %
    We therefore choose source problems whose constraints can be expressed through such rank comparisons, namely \textup{\textsc{Directed Cut}}, which is \NP-hard \cite{farley_directedmaximalcut_1982}, and \textup{\textsc{Mixed 2-Sided Orthant Depth}}, which is \Wone-hard parameterized by the dimension (see \Cref{thm:M2SOD-wone-hard}).

    \textup{\textsc{Directed Cut}} asks whether a loopless directed graph $D=(V,A)$ has a subset of vertices $S\subseteq V$ with at least $k$ arcs directed from $S$ to $V\setminus S$. Equivalently, a set of arcs can be contained in a directed cut precisely when no vertex occurs as both the head and the tail of selected arcs. For two distinct selected arcs $e,f$, this means $\operatorname{head}(e)\neq\operatorname{tail}(f)$ and $\operatorname{tail}(e)\neq\operatorname{head}(f)$; we call these two inequalities the \emph{directed-cut condition}. This can be expressed as a rank comparison over the arcs: Fix an arbitrary ordering of $V$ with rank function $r$ and represent each arc $e$ by a vertex $x_e$. Then a set of arcs $A'\subseteq A$ can be contained in a directed cut if and only if, for every ordered pair of distinct arcs $e,f\in A'$, we have 
    $r(\operatorname{head}(e))< r(\operatorname{tail}(f))$ or $r(\operatorname{head}(e))> r(\operatorname{tail}(f))$. 

    \textup{\textsc{Mixed 2-Sided Orthant Depth (M2SOD)}} has as input an integer $\ell$ and an indexed family $\mathcal O$ of \textit{mixed 2-sided orthants} in $\mathbb{R}^d$, which are regions $O=\{x\in\mathbb{R}^d \mid x_i\geq a,\ x_j\leq b\}$, where $i\neq j$. The task is to decide whether there is a point contained in at least $\ell$ orthants, counting repeated orthants separately.
    Write $\mathsf{ldim}(O)=i$, $\mathsf{udim}(O)=j$, $\mathsf{lval}(O)=a$, and $\mathsf{uval}(O)=b$. Two orthants $O,P\in\mathcal O$ intersect exactly when $\mathsf{ldim}(O)\neq\mathsf{udim}(P)$ or $\mathsf{lval}(O)\leq\mathsf{uval}(P)$, and the analogous condition holds with the roles reversed.
    A family of orthants has a common point if and only if the orthants intersect pairwise. Hence, the existence of such a common point can be determined solely by comparing lower and upper bounds in matching coordinates of any pair in the family.
    
    The two vertex cover hardness proofs use the following rank-comparison gadget. For the other footprint parameters, we use two comparison vertices with subdivided connections to obtain smaller parameter bounds while retaining a common reduction for open and closed components. We keep the constructions proper and, where possible, simple.

    \begin{definition}[rank-comparison gadget]
        Let $X,Y$ be possibly overlapping vertex sets and let $r_1\colon X\to[\mu]$ and $r_2\colon Y\to[\mu]$ be rank functions. The directed gadget $Q_{\leq}(X,Y;r_1,r_2)$ contains the terminals $X\cup Y$, two fresh vertices $a_1,a_2$, and fresh threshold vertices $p_\ell$ for $\ell\in[\mu]$. Its edges are $(x,a_1)$ for $x\in X$, $(a_1,p_\ell)$ and $(p_\ell,a_2)$ for $\ell\in[\mu]$, and $(a_2,y)$ for $y\in Y$.

        Assign globally distinct time labels in two consecutive blocks $L^-<L^+$. In $L^-$, process ranks $\ell=1,\ldots,\mu$ in increasing order, placing all edges $(x,a_1)$ with $r_1(x)=\ell$ before $(a_1,p_\ell)$. In $L^+$, again process ranks in increasing order, placing $(p_\ell,a_2)$ before all edges $(a_2,y)$ with $r_2(y)=\ell$. Fix arbitrary orders within each group.

        The gadget $Q_<$ reverses the order of the two groups within each rank in $L^-$. For $r^{\mathrm{rev}}:=\mu+1-r$, define $Q_{\geq}(X,Y;r_1,r_2):=Q_{\leq}(X,Y;r_1^{\mathrm{rev}},r_2^{\mathrm{rev}})$ and $Q_>(X,Y;r_1,r_2):=Q_<(X,Y;r_1^{\mathrm{rev}},r_2^{\mathrm{rev}})$.
    \end{definition}
    \begin{figure}[htb]
        \centering
        \includegraphics[page=1,width=0.9\linewidth]{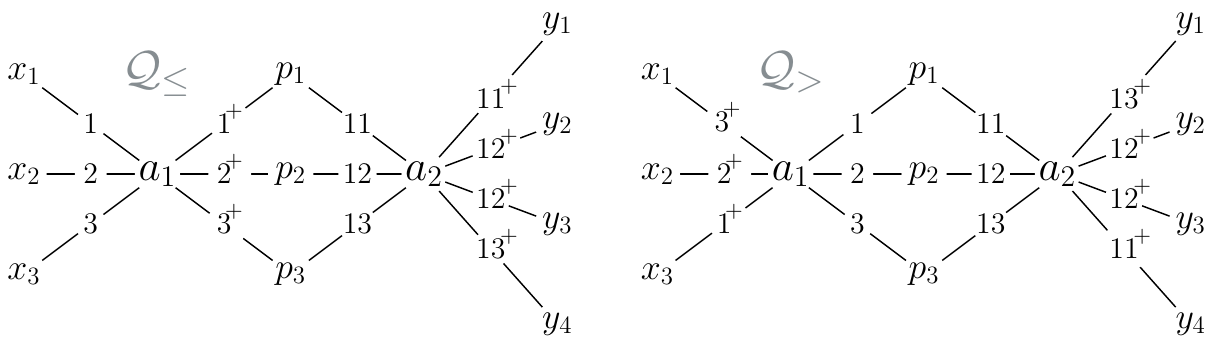}
        \caption{Rank-comparison gadgets $\mathcal Q_{\leq}$ and $\mathcal Q_{>}$ for $X=\{x_1,x_2,x_3\}$, $Y=\{y_1,y_2,y_3,y_4\}$, $\mu=3$, $r_1(x_i)=i$, and $(r_2(y_1),r_2(y_2),r_2(y_3),r_2(y_4))=(1,2,2,3)$. The right gadget reverses both ranks and uses the strict comparison. Here $t<t^+<t+1$; repeated schematic labels are separated by tie-breaking. The directed edges point from $X$ through $a_1,p_\ell,a_2$ to $Y$; orientations are omitted.}
        \label{fig:o-rankcomparison-gadget}
    \end{figure}

    \begin{lemma}[Reachability through a rank-comparison gadget]\label{lem:overview-rank-comparison}
        For distinct terminals $x\in X$ and $y\in Y$ and a comparison $\diamond\in\{\leq,<,\geq,>\}$, the directed gadget satisfies
        \[
            x\rightreach y\text{ in }\mathcal Q_{\diamond}(X,Y;r_1,r_2)
            \quad\Longleftrightarrow\quad r_1(x)\diamond r_2(y).
        \]
        Each gadget is simple and proper and has a vertex cover of size at most two. These properties hold under both strict and non-strict reachability.
    \end{lemma}
    \begin{proof}
        Every temporal path between distinct terminals must have the form $x,a_1,p_\ell,a_2,y$. 
        In $\mathcal Q_{\leq}$, the edge $x a_1$ precedes $a_1 p_\ell$ exactly when $r_1(x)\leq\ell$, and $p_\ell a_2$ precedes $a_2 y$ exactly when $\ell\leq r_2(y)$. Such a path therefore exists precisely when $\exists\ell\in[\mu]: r_1(x)\leq\ell\leq r_2(y)$, which is equivalent to $r_1(x)\leq r_2(y)$.
        For $\mathcal Q_{<}$, reversing the order within each early rank changes the first inequality to $r_1(x)<\ell$, giving $r_1(x)<r_2(y)$. Reversing the ranks gives the other two comparisons. 
        Every static edge receives one label, all labels are distinct, and $\{a_1,a_2\}$ meets every edge. Thus the gadget is simple and proper, has a vertex cover of size at most two, and its temporal paths are valid under either reachability convention.
    \end{proof}

    \subsubsection{\Wone-hardness for \openTCC and \closedTCC wrt \vc in undirected graphs from \textup{\textsc{M2SOD}}}
    \begin{restatable}{theorem}{undirectedVCHardness}\label{thm:vc-whard}
        \openTCC and \closedTCC on undirected simple proper temporal graphs are \Wone-hard when parameterized by the vertex cover number of the footprint.
    \end{restatable}
    \begin{proof}[Proof sketch]
        We represent every orthant $O$ by an \emph{orthant-vertex} $z_O$. For every dimension~$i$, we add a rank-comparison gadget connecting the vertices whose orthants are lower-bounded in dimension $i$ to those whose orthants are upper-bounded in dimension $i$ and have an upper value at least as large. The two cover vertices of each gadget are denoted by $L_i$ and $U_i$. To connect orthant-vertices whose relevant bounds lie in different dimensions $i\neq j$, an additional path from $L_i$ to $U_j$ provides reachability without comparing values.
        Lastly, we add auxiliary vertices $a_1,a_2,b_1,b_2$ so that the set of the rank-comparison vertices together with the auxiliary vertices is temporally connected and each of these vertices is mutually reachable with every orthant-vertex. As a result, every connected component will be both open and closed.

        We then show that, for distinct orthant-vertices, $z_O\rightreach z_{O'}$ if and only if $\mathsf{ldim}(O)\neq\mathsf{udim}(O')$ or $\mathsf{lval}(O)\leq\mathsf{uval}(O')$. If the dimensions differ, the corresponding portal provides a path; if they agree and the inequality holds, the comparison gadget provides one. Both paths use only non-orthant vertices internally.
    \end{proof}
        \begin{figure}[t]
            \centering
            \begin{minipage}[c]{0.49\linewidth}
                \caption{Sketch of the orthant reduction for two dimensions $i\neq j$.
                The orthant-vertex $z_O$ is attached to $L_i$ and $U_j$ according to the dimensions of its lower and upper bounds. Between $L_i$ and $U_i$, a rank-comparison gadget $Q_\leq$ implements comparisons within a dimension; portals such as $x_{i,j}$ connect different dimensions without comparing values.
                The labels $1,\ldots,5$ indicate consecutive temporal phases, and superscripts indicate local precedence; exact time labels and the rank-dependent ordering within comparison gadgets are omitted.
                }
                \label{fig:o-vc-Wh}
            \end{minipage}\hfill
            \begin{minipage}[c]{0.47\linewidth}
                \centering
                \includegraphics[width=\textwidth]{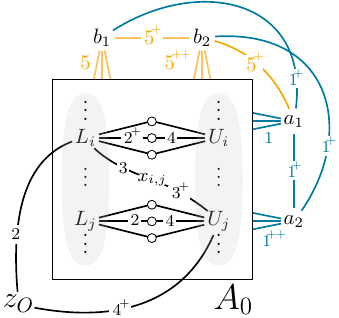}
            \end{minipage}
        \end{figure}

    For temporal path number, we encode the same rank comparisons through the order of vertices along base-paths. Two base-paths per dimension suffice: transfers between them implement the comparison of lower and upper bounds. Thus, the dimension bounds the temporal path number, and consequently the temporal walk number.
    \begin{restatable}{theorem}{openPathHardness}\label{thm:openTCC-wone-hard}
        \openTCC is \Wone-hard when parameterized by temporal path number, and consequently temporal walk number, on simple proper directed and undirected temporal graphs.
    \end{restatable}
    \begin{proof}[Proof sketch]
        Each orthant $O$ is represented by an \textit{orthant-vertex} $z_O$, and the construction will ensure that distinct orthant-vertices are mutually reachable if and only if their orthants intersect.
        For every dimension $i\in[d]$, we introduce one base-path $U_i$ for the upper values and one base-path $L_i$ for the lower values.
        For every bound value $\alpha$, we introduce a \textit{switch vertex} $s_{i,\alpha}$ on both $L_i$ and $U_i$, ordered by increasing~$\alpha$.
        We place a private \textit{wall vertex} immediately after each switch on $L_i$ and immediately before it on $U_i$. These walls prevent repeated edges between switches and will also block unwanted backward traversals in the undirected construction.
        The orthant-vertex $z_O$ is then placed on exactly two paths, namely $L_{\mathsf{ldim}(O)}$ and $U_{\mathsf{udim}(O)}$. On the lower path, it is placed before the switch corresponding to $\mathsf{lval}(O)$, while on the upper path, it is placed after the switch corresponding to $\mathsf{uval}(O)$.
        This allows a temporal path to switch from the lower path to the upper path within the same dimension exactly when the lower value is at most the upper value.
        To connect orthant-vertices whose relevant bounds lie in different dimensions $i\neq j$, we add an additional \textit{portal vertex} after the last orthant-vertex on $L_i$ and before the first orthant-vertex on $U_j$. These portal vertices provide reachability without comparing values.
        Lastly, to ensure that the orthant-vertices are preferred in a connected component, we replace each by sufficiently many consecutive copies and place them in opposite orders on their lower and upper paths, making the copies of each orthant-vertex temporally connected. Thus, every represented orthant contributes sufficiently many vertices to the connected component. The construction uses $2d$ base-paths, bounding $\tpn$ by $2d$. Since $\twn\leq\tpn$, the same reduction also establishes \Wone-hardness parameterized by temporal walk number.
    \end{proof}
    \begin{figure}[t]
        \centering
        \includegraphics[width=\linewidth]{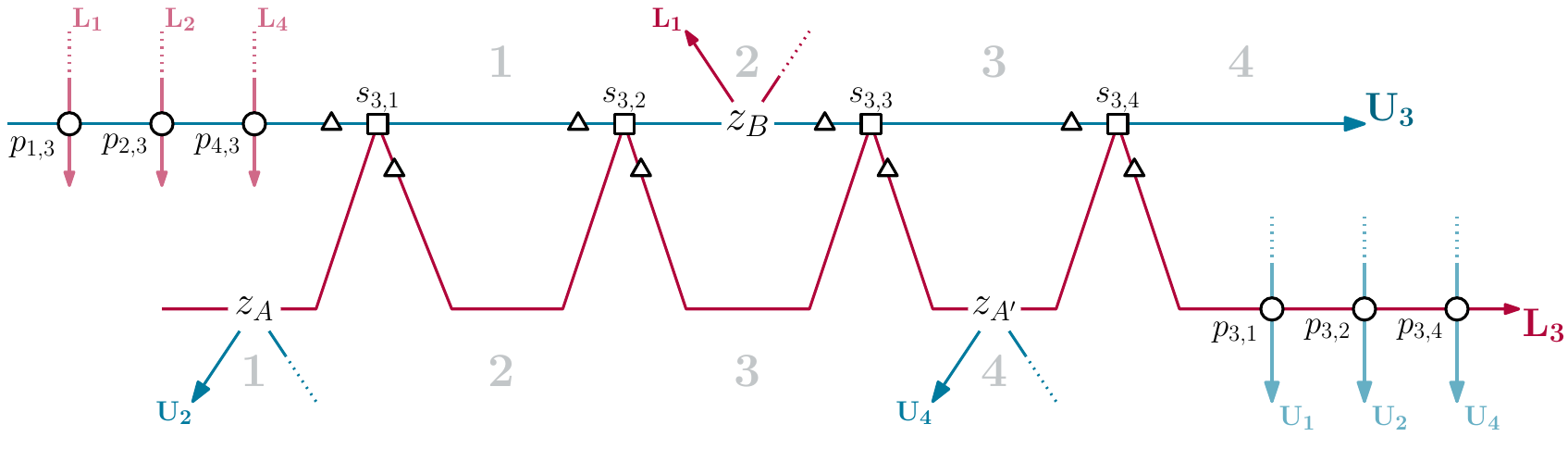}
        \caption{Encoding bound comparisons using the lower base-path $L_3$ (red) and upper base-path $U_3$ (blue).
        Squares are switches, circles are portals to other dimensions, and triangles are walls; gray numbers indicate bound values.
        All lower-path edges occur before all upper-path edges.
        Thus, $z_A$ can reach $z_B$ by switching from $L_3$ to $U_3$, since $A$'s lower bound $1$ is at most $B$'s upper bound $2$.
        The lower bound $4$ of $A'$ prevents the same reachability from $z_{A'}$.
        Each orthant is shown by one vertex, with its copies omitted.}
        \label{fig:o-orthants3}
    \end{figure}
    
    \subsubsection{\paraNP-hardness for \openTCC and \closedTCC wrt footprint parameters from \textup{\textsc{Directed Cut}}}
    We reduce from \textup{\textsc{Directed Cut}}, using its edge characterization: find the largest set of arcs $A'\subseteq A$ such that every two distinct arcs $e,f\in A'$ satisfy the directed-cut condition 
        $\operatorname{head}(e)\neq\operatorname{tail}(f)$
         and 
        $\operatorname{tail}(e)\neq\operatorname{head}(f)$.
    We represent every arc $e$ of the input graph by an arc-vertex $x_e$, and arrange that $x_e$ reaches $x_f$ exactly when $\operatorname{head}(e)\neq\operatorname{tail}(f)$. Thus, mutual reachability checks the two inequalities of the directed-cut condition.
    Fixing an ordering of the input vertices with rank function $r$, we realize these inequalities using two \emph{comparison portals}. The first will provide a temporal path when $r(\operatorname{head}(e))<r(\operatorname{tail}(f))$ and the second when $r(\operatorname{head}(e))>r(\operatorname{tail}(f))$.
    All three reductions follow the same general layout (see \Cref{fig:o-D-vc} left):
    \begin{enumerate}
        \item \emph{Two comparison portals.} We implement the portals either by the rank-comparison gadgets $Q_<$ and $Q_>$ or by two individual comparison vertices $a,b$.
        \item \emph{Connections to portals.} We use direct in/out edges or subdivided in/out edges. Their time labels ensure that comparison paths cannot be chained to create unintended reachability.
        \item \emph{Auxiliary edges.} We add suitably early and late auxiliary edges so that every valid selection of arc-vertices extends to a closed temporally connected set containing all auxiliary vertices, without changing mutual reachability between arc-vertices.
    \end{enumerate}
    Writing $Z$ for all auxiliary vertices, we set the target size to $k'=|Z|+k$. A directed cut with at least $k$ arcs therefore gives a closed, and hence also open, temporally connected set of size at least $k'$. Conversely, any open or closed tcc of this size contains at least $k$ mutually reachable arc-vertices, whose arcs satisfy the directed-cut condition.

    \begin{figure}[t]
        \centering
        \begin{minipage}[c]{0.46\linewidth}
            \centering
            \includegraphics[width=0.9\linewidth]{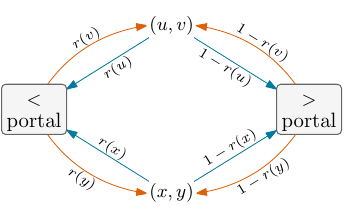}
        \end{minipage}\hfill
        \begin{minipage}[c]{0.50\linewidth}
            \centering
            \includegraphics[width=\linewidth]{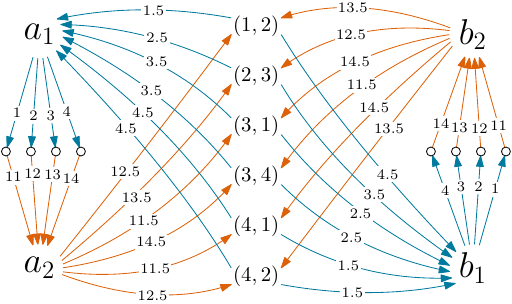}
        \end{minipage}
        \caption{Left: the general framework for an ordered pair of arcs $e=(u,v)$ and $f=(x,y)$. Blue in-connections and orange out-connections are direct or subdivided; each comparison portal is a rank-comparison gadget or a single vertex. Arrows indicate the intended temporal traversal. Right: the two rank-comparison gadgets in the directed vertex cover construction. Central vertices represent arc-vertices. Blue and orange edges belong to the earlier and later comparison blocks, respectively. The additional edges ensuring auxiliary connectivity are omitted.}
        \label{fig:o-D-vc}
    \end{figure}

    For directed vertex cover, we add the rank-comparison gadgets $Q_<$ and $Q_>$ with cover vertices $a_1,a_2$ and $b_1,b_2$, respectively, and direct in/out edges. The additional early and late edges provide auxiliary connectivity and are incident with the same four cover vertices.
    \begin{restatable}{theorem}{directedVCFour}\label{thm:directed-vc-four}
        Both \openTCC and \closedTCC are \NP-hard on simple, proper directed temporal graphs even when $\vc(\gcal)=4$.
    \end{restatable}

    To obtain planar footprints with smaller treewidth, pathwidth, and feedback vertex number, we replace each rank-comparison gadget by a single comparison vertex $a$ or $b$. Arrivals and departures at $a$ are ordered by increasing head and tail ranks, realizing $r(\operatorname{head}(e))<r(\operatorname{tail}(f))$, and decreasing ranks at $b$ realize the reverse inequality. Within each rank, departures precede arrivals as in the strict rank-comparison gadget.
    It is important that a path completing one comparison must not enter an arc-vertex early enough to start a second comparison. Thus, we subdivide the outgoing connections: the first edge performs the comparison, while the final edge reaches the arc-vertex only after both comparison blocks have ended. The auxiliary edges must also be redesigned carefully to preserve planarity.

    In the directed construction (\Cref{fig:o-D-planar}), comparisons follow $x_e\to a\to\ell_f\to x_f$ or $x_e\to r_e\to b\to r_f\to x_f$, with the final edge delayed until block~$4$. Opposite arcs let the connections at $b$ share $r_e$ without repeating a directed edge. 
    Thus, the construction remains simple and, with distinct labels, proper. 
    Deleting $a$ leaves paths meeting only at $b$, giving feedback vertex number one.
    \begin{restatable}{theorem}{directedPlanarHardness}\label{thm:directed-planar}
        Both \openTCC and \closedTCC are \NP-hard on simple, proper directed temporal graphs whose footprint is planar and satisfies $\tw(G)=2$, $\pw(G)=3$, and $\fvn(G)=1$.
    \end{restatable}
    \begin{figure}[t]
        \centering
        \includegraphics[width=0.9\linewidth]{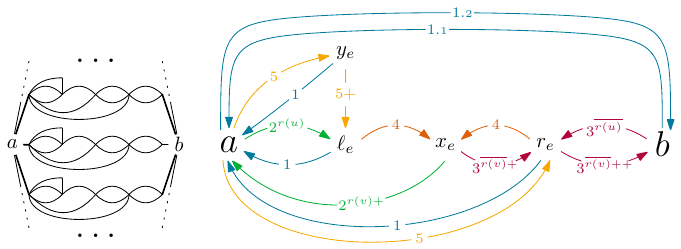}
        \caption{Directed planar construction: the schematic footprint (left) and one arc gadget (right) for $e=(u,v)$. Blocks $1,5$ provide auxiliary connectivity, $2,3$ perform comparisons, and $4$ returns to arc-vertices. Superscripts specify rank order, with $\overline{r(w)}=|V|+1-r(w)$; plus signs specify local precedence. Schematic ties are broken to make all times distinct.}
        \label{fig:o-D-planar}
    \end{figure}

    In the undirected setting, we use separate subdivided connections at each comparison vertex. The two edges of an incoming connection occur consecutively towards the portal, so no path can traverse it backwards into the arc-vertex. Outgoing connections again delay arrival at arc-vertices until after both comparisons. 
    The construction is non-simple but still proper. Deleting $a,b$ leaves disjoint stars.
    \begin{restatable}{theorem}{undirectedPlanarHardness}\label{thm:undirected-planar}
        Both \openTCC and \closedTCC are \NP-hard on proper undirected temporal graphs whose footprint is planar and satisfies $\tw(G)=2$, $\pw(G)=3$, and $\fvn(G)=2$.
    \end{restatable}

\subsection{Combined parameters}
    \begin{restatable}{theorem}{combinedTemporalFPT}\label{thm:combined-temporal-fpt}
        \openTCC and \closedTCC on directed and undirected temporal graphs are \FPT when parameterized by $\tw+\lifetime$, $\tw+\tdegree$, or $\tw+\tpn$ under both strict and non-strict reachability.
    \end{restatable}
    The algorithms follow from MSO formulations and the optimization version of Courcelle's theorem. The $\tw+\tpn$ follows from $\tdegree\leq 2\tpn$. However, polynomial kernels are unlikely.
    \begin{restatable}{theorem}{combinedKernelBounds}\label{thm:kernel-lower-bounds}
        \openTCC and \closedTCC on directed and undirected temporal graphs do not admit polynomial kernels when parameterized by $\tw+\tdegree+\lifetime$, or by $\tw+\tpn$ and consequently $\tw+\twn$, unless $\coNP\subseteq\NP/poly$.
    \end{restatable}
    The kernel lower bound for $\tw+\twn$ follows from the same cross-composition as for $\tw+\tpn$, since $\twn\leq\tpn$ also bounds the output parameter.

\addtocontents{toc}{\protect\setcounter{tocdepth}{\number\value{tocdepth}}}

\section{Conclusion and Outlook}\label{sec:conclusion-outlook}
 Our results nearly settle the parameterized complexity of maximum open and closed temporal connected components for the footprint parameters considered: only vertex cover admits \XP algorithms, and only in undirected temporal graphs. Combining treewidth with lifetime or maximum temporal total degree yields \FPT algorithms, but these temporal parameters tend to be large in practice.

Temporal path number offers a different structural perspective: it measures how many temporal routes suffice to describe all interactions. Unlike the footprint parameters considered, bounded \twn, and consequently bounded \tpn, permits \XP algorithms for \openTCC in both directed and undirected temporal graphs. However, \openTCC remains \Wone-hard parameterized by \tpn and consequently \twn, while \closedTCC is \NP-hard already for $\tpn=6$ and consequently $\twn\leq6$.

Beyond closing the remaining gaps at small constant parameter values, a promising direction is therefore to restrict how the base-paths interact. Our hardness constructions use paths that visit shared vertices in conflicting orders. Requiring pairwise monotone paths---whose orders agree on their shared vertices---may yield stronger tractability results. Does this restriction make \openTCC \FPT parameterized by \tpn, or even permit \XP algorithms for \closedTCC?

\clearpage
\tableofcontents
\clearpage

\section{Extended Related Work}\label{sec:related work}

    The study of connected components in temporal graphs dates back at least to Bhadra and Ferreira  2002 \cite{bhadra_ComputingMulticast_2002,bhadra_ComplexityConnected_2003}, who introduced open and closed tccs as the natural extension of static connected components. They proved that on directed, non-strict temporal graphs computing a maximum open or closed tcc is \NP-complete via a reduction from \clique.
    Jarry and Lotker \cite{jarry_ConnectivityEvolving_2004} later showed that both variants are polynomial-time solvable on trees (all undirected (non-)strict).
    {They also claim that both variants are \NP-hard on unit-2 disk graphs, however their proof seems to be wrong as they reduce from computing static connected components which is poly time solvable.}
    Subsequent empirical and metric works—apparently unaware of these earlier papers—reintroduced the tcc notions on (un)directed strict graphs, reproved computational hardness, and explored the size of tccs in human contact and social network data \cite{tang_CharacterisingTemporal_2010,nicosia_ComponentsTimeVarying_2012,nicosia_GraphMetrics_2013}.
    
    \smallparagraph{Parameterized complexity.} Casteigts \cite{casteigts_FindingStructure_2018} established \Wone-hardness for \openTCC and \closedTCC by solution size on strict temporal graphs and later, together with Corsini and Sarkar \cite{casteigts_SimpleStrict_2024}, refined the Bhadra–Ferreira reduction with the semaphore construction to show \NP-hardness even in simple (non-)strict graphs.
    Costa, Lopes, Marino, and Silva~\cite{costa_ComputingLarge_2023} studied \openTCC and \closedTCC under parameterization by lifetime, component size, and their combination. They proved all cases \Wone-hard except size+lifetime  on undirected non-strict graphs, which is \FPT.
    Casteigts, Komusiewicz, and Morawietz~\cite{casteigts_hardnessfinding_2026} considered the more permissive decision problem of whether a temporal graph contains any nontrivial closed tcc at all, meaning one of size at least two in directed graphs and at least three in undirected graphs. They showed that even this relaxed variant is \NP-hard for directed proper simple temporal graphs of constant lifetime, for undirected proper simple temporal graphs in the non-strict setting, and for undirected simple temporal graphs of constant lifetime in the strict setting.
    The only parameter known to yield more general \FPT results is \textit{distance to transitivity} by Casteigts, Morawietz, and Wolf~\cite{casteigts_DistanceTransitivity_2024}, measuring the number of vertex or edge deletions required to make the reachability graph transitive. This yields \FPT algorithms for \openTCC on (un)directed (non-)strict graphs, while \closedTCC remains \NP-hard already for constant distance one.
    
    
    \smallparagraph{Path-based variants of tccs.}
        Beyond open and closed tccs, $\Delta$-components require a temporal path within every window of length $\Delta$ (closed in strict graphs in \cite{huyghues-despointes_Forteconnexite_2016}, open in non-strict graphs in \cite{calzado_ConnectivityModel_2015}).
        Costa, Lopes, Marino, and Silva~\cite{costa_ComputingLarge_2023} introduced \textit{unilateral} variants of open and closed connected components, where for each pair of vertices only one is required to reach the other, and studied their parameterized complexity. 
        Balev, Pigné, Sanlaville, and Schoeters~\cite{balev_temporallyconnected_2024} studied connected components variants from a source- and sink-based perspective, where single, multiple, or all vertices must reach single, multiple, or all other vertices. They provide structural results, such as bounds on the number of tccs, and a detailed analysis of exponential-time algorithms.
    
    \smallparagraph{Snapshot-based variants of tccs.}
    There exist several extensions of connectivity of a vertex set~$X$ which do not take the path-based approach.
        \textit{$T$-interval connectivity} requires a common connected spanning subgraph on the vertices in $X$ across every length-$T$ window; it admits an optimal $\mathcal{O}(\lifetime)$ online algorithm \cite{casteigts_EfficientlyTesting_2015}.
        \textit{Persistent components} must also be connected in every snapshot of a time interval, but the spanning subgraphs can differ \cite{vernet_studyconnectivity_2023}.
        \textit{Window-CC}'s form a connected component in the static graph formed by taking the union of the snapshots over a time window \cite{xie_QueryingConnected_2023} and can be computed efficiently as they correspond to connected components of a static graph.
        Another notion of static-temporal components contains all temporal vertices that form a connected component in the static expansion of the graph; this variant can also be computed in polynomial time \cite{rannou_StronglyConnected_2021}.
    
    For a concise overview of the different notions of temporal connected components and the related literature, we refer to \cite{doring_TemporalConnected_2025}.

    \smallparagraph{Checking (maximal) temporal connectivity.}
        Confirming whether a set of vertices is temporally connected is simple. One can use a temporal variant of Dijkstra's algorithm \cite{halpern_ShortestPath_1974,berman_Vulnerabilityscheduled_1996,xuan_computingshortest_2003}, or stream the edges in chronological order while recording reachability from a fixed source \cite{wu_Pathproblems_2014}.
        Variants with practical restrictions, such as forbidding or bounding waiting time, have also been studied \cite{halpern_ShortestPath_1974,bentert_Efficientcomputation_2020}.
    %
        Checking a temporally connected set $X$ for maximality (if $X$ is a tcc) depends on the considered connectivity notion. For open tccs, it suffices to test for each vertex outside~$X$ whether adding it makes~$X$ disconnected which is polynomial-time solvable, while for closed tccs, deciding whether~$X$ is maximal is \NP-complete \cite{costa_ComputingLarge_2023}. 

\smallparagraph{Parameterized complexity on temporal graphs.}
A variety of structural parameters have been considered for different problems on temporal graphs, including the \textit{lifetime}, the number of \textit{edges per time step}, the \textit{temporal/static degree}, the size of a \textit{timed feedback edge set} \cite{haag_FeedbackEdge_2022}, the \textit{temporal core }\cite{zschoche_ComplexityFinding_2020}, \textit{vertex-} and \textit{tree-interval-membership width} \cite{enright_familiestractable_2025,hand_MakingLife_2022}, and the \textit{treewidth} of the underlying static graph.
Parameterized results for treewidth combined with either lifetime or temporal degree have been obtained via the MSO approach for several problems \cite{enright_DeletingEdges_2021,haag_FeedbackEdge_2022,deligkas_Beinginfluencer_2024}. A~survey of temporal treewidth variants and their use for parameterized complexity can be found in \cite{fluschnik_timegoes_2020}.
Most recently, the exact edge cover was introduced as a natural mathematical construct for temporal graphs motivated by train systems \cite{deligkas_HowMany_2025}. This gives rise to the \textit{temporal path number} parameter\,--\,the minimum size of an exact edge cover.


\section{Preliminaries}\label{sec:full-preliminaries}
        
        A \textit{temporal graph} $\gcal=(V,E,\lambda)$ consists of a static graph $G=(V,E)$, called the \textit{footprint}, along with a labeling function $\lambda$.
        The temporal graph is called \emph{(un)directed} if the footprint is (un)directed.
        A pair $(e, t)$, where $e \in E$ and $t \in \lambda(e)$, is a \textit{temporal edge} with \textit{label} $t$. We denote the set of all temporal edges by $\ecal$.
        The \textit{temporal total degree} of a vertex $v\in V$ is defined as $\delta^t(v) = \lvert \{(e,t) \colon v\in e,\; t\in\lambda(e)\}\rvert$ and $\tdegree=\max_{v\in V} \delta^t(v)$ denotes the \textit{maximum temporal total degree} of~\gcal. The \textit{static degree} of $v$ is defined as $\delta(v) = \lvert \{e\in E \colon v\in e\}\rvert$. 
        The largest time label is referred to as the \textit{lifetime} \lifetime. The static graph $G_t = (V, E_t)$, where $E_t = \{e \in E \colon t \in \lambda(e)\}$, is called the \textit{snapshot} at time $t$.
        
        A \textit{temporal path} is a sequence of temporal edges $\tuple{(e_i,t_i)}$ where $\tuple{e_i}$ forms a path in the footprint and the time labels $\tuple{t_i}$ are non-decreasing. If the time labels are strictly increasing, the path is called \textit{strict}; otherwise, it is called \textit{non-strict}.
        A \emph{temporal walk} is defined analogously except that vertices may be revisited.
        If there exists a temporal path from $u$ to $v$, we say \emph{$u$ reaches $v$} denoted $u\rightreach v$. A graph where all reachability is considered exclusively using (non-)strict paths is called a \emph{(non-)strict temporal graph}. A temporal graph is \emph{simple} if every static edge has exactly one time label, that is, $|\lambda(e)|=1$ for every $e\in E$. A temporal labeling $\lambda$ is called \emph{proper} if incident edges have distinct labels. In that case, there is no distinction between strict and non-strict.
        The \textit{reachability graph} $\rcal(\gcal) = (V, \{(u,v) \colon u \rightreach v\text{ in }\gcal\})$ is a static directed graph whose arcs represent temporal reachability between vertices in~$\gcal$.


        A set of vertices $C \subseteq V$ is \emph{temporally connected} if {for every ordered pair $u,v\in C$, there is a temporal path from $u$ to $v$ in $\gcal$}.
        Following~\cite{bhadra_ComplexityConnected_2003}, we distinguish between two notions of temporal connected components: \textit{closed} connected components, where the paths witnessing connectivity must remain within $C$, and \textit{open} connected components, where these paths may also use vertices outside of $C$. 
        \begin{definition}[open/closed temporal connected component]
            A subset of vertices $C\subseteq V$ is an \emph{open temporal connected component (\otcc)} if $C$ is inclusion-maximal subject to being temporally connected, and a \emph{closed temporal connected component (\ctcc)} if it is inclusion-maximal subject to the induced temporal subgraph $\gcal[C]$ being temporally connected.
        \end{definition}
        
    \subsection{Parameterized Complexity}
    We refer to the standard books for a basic overview of parameterized complexity theory~\cite{cygan_ParameterizedAlgorithms_2015,downey_FundamentalsParameterized_2013,fomin_kernelizationtheory_2019}. At a high level, parameterized complexity studies the complexity of a problem with respect to its input size $n$  and the size of a parameter $k$.
    A problem is \emph{fixed-parameter tractable} (\FPT) by $k$ if it can be solved in time $f(k) \cdot \poly(n)$, where $f$ is a computable function.
    Showing that a problem is $\Wone$-hard parameterized by $k$ rules out the existence of such an \FPT algorithm under the assumption $\Wone \neq \FPT$.
    A less favorable, but still positive, outcome is an algorithm with an exponential running time $\bigoh(n^{f(k)})$ for some computable function~$f$; problems admitting such algorithms belong to the class~$\XP$.
    A problem is $\paraNP$-hard if it remains \NP-hard even when the parameter $k$ is constant. Thus, $\paraNP$-hardness excludes both \FPT and \XP algorithms under standard complexity assumptions.
    Another notion central to the parameterized algorithms and complexity is that of a \emph{kernelization} algorithm. 
    \begin{definition}[kernelization~\cite{fomin_kernelizationtheory_2019}]
    Let $L$ be a parameterized problem over a finite alphabet~$\Sigma$. A \emph{kernelization algorithm}, or in short, a \emph{kernelization}, for $L$ is an algorithm with the following property. For any given $(x, k) \in \Sigma^* \times \mathbb{N}$, it outputs in time polynomial in $|(x, k)|$ a string $x_0 \in \Sigma^*$ and an integer $k_0 \in \mathbb{N}$ such that $$((x, k) \in L \Leftrightarrow (x_0 , k_0) \in L)\text{ and }|x_0|, k_0 \le h(k),$$ where $h$ is an arbitrary computable function. If $K$ is a kernelization for $L$, then for every instance $(x, k)$ of $L$, the result of running $K$ on the input $(x, k)$ is called the \emph{kernel} of $(x, k)$ (under $K$). The function $h$ is referred to as the \emph{size} of the kernel. If $h$ is a polynomial function, then we say that the kernel is \emph{polynomial}.
    \end{definition}

    We consider four static parameters: treewidth~\tw, {pathwidth}~\pw, {feedback vertex number}~\fvn, and {vertex cover number}~\vc.
    All four refer to the footprint; for directed footprints, we use the underlying undirected graph.
    Treewidth measures how close a graph is to being a tree; graphs of treewidth at most~1 are precisely the forests.
    \iflong
    \begin{definition}[tree decomposition, treewidth]
        Let $G=(V,E)$ be an undirected static graph.  
        A \emph{tree decomposition} of $G$ is a pair 
        $T = (T,\{B_u \colon u \in V(T)\})$ consisting of a tree $T$ and a family of bags $B_u \subseteq V$ such that
        \begin{enumerate}
            \item $\bigcup_{u \in V(T)} B_u = V$,
            \item for every $e \in E$ there exists $u \in V(T)$ with $e \subseteq B_u$, and
            \item for every $v \in V$, the set $\{u \in V(T) \colon v \in B_u\}$ induces a connected subtree of~$T$.
        \end{enumerate}
        The \emph{width} of $T$ is defined as $width(T) := \max_{u \in V(T)} \lvert B_u\rvert - 1$.
        The \emph{treewidth} of $G$ is 
        $
            \tw(G) := \min\{width(T) \colon T \text{ is a tree decomposition of } G\}.
        $
    \end{definition}
    \else
    \begin{definition}
        Let $G=(V,E)$ be an undirected static graph.  
        A \emph{tree decomposition} of $G$ is a pair 
        $T = (T,\{B_u \colon u \in V(T)\})$ consisting of a tree $T$ and a family of bags $B_u \subseteq V$ such: that (i) $\bigcup_{u \in V(T)} B_u = V$, (ii) for every $e \in E$ there exists $u \in V(T)$ with $e \subseteq B_u$, and (iii) for every $v \in V$, the set $\{u \in V(T) \colon v \in B_u\}$ induces a connected subtree of~$T$.
        The \emph{width} of $T$ is defined as $width(T) := \max_{u \in V(T)} \lvert B_u\rvert - 1$.
        The \emph{treewidth} of $G$ is $\tw(G) := \min\{width(T) \colon T \text{ is a tree decomposition of } G\}$.
    \end{definition}
    \fi
    A \emph{path decomposition} is a tree decomposition whose decomposition tree is a path. The \emph{pathwidth} $\pw(G)$ is the minimum width of a path decomposition of~$G$.
    A \emph{feedback vertex set} is a set $S\subseteq V(G)$ such that $G-S$ is a forest, and the \emph{feedback vertex number} $\fvn(G)$ is the minimum size of such a set.
    A \emph{vertex cover} is a set $S\subseteq V(G)$ containing at least one endpoint of every edge of~$G$, and the \emph{vertex cover number} $\vc(G)$ is the minimum size of such a set.
    These parameters satisfy $\vc(G)\geq\pw(G)\geq\tw(G)$, $\vc(G)\geq\fvn(G)$, and $\fvn(G)+1\geq \tw(G)$.

    \subsection{Temporal Path and Temporal Walk Number}\label{subsec:tpn-tm-incomparable}

        We consider temporal graphs which are given as the union of a small number of temporal paths or walks. This is motivated by the fact that many real-world temporal networks are assembled from a small number of temporal routes, such as scheduled journeys in a transportation network. This notion was introduced by Deligkas et al.~\cite{deligkas_HowMany_2025} in the form of an exact edge-cover.
        \begin{definition}[\kpathgraphs{k}, temporal path number]
            A temporal graph $\gcal=(V,E, \lambda)$ is a \emph{\kpathgraph{k}}, if there exists a collection $\pcal=\{P_1,\dots,P_k\}$ of $k$ {\em temporal base-paths} such that the multiset union of their temporal edges is exactly \ecal.
            We may write $\gcal=\biguplus\pcal=\biguplus_{i\in[\paraPaths]}P_i$.
            The \emph{temporal path number} \tpn of a temporal graph $\hcal$ is the minimum number of paths needed to define \hcal as a path graph.
        \end{definition}
        The notions \kwalkgraphs{k} and temporal walk number $\twn(\gcal)$ are defined analogously.
        Since every temporal path is a temporal walk, $\twn(\gcal)\leq\tpn(\gcal)$ for every temporal graph~$\gcal$.

        To relate these parameters to earlier work on tccs, we compare temporal path number with the two transitivity modulator parameters studied by Casteigts et al.~\cite{casteigts_DistanceTransitivity_2024} which measure how far temporal reachability is from being transitive. The two constructions below show that temporal path number and either transitivity modulator parameter are incomparable. Thus, bounded temporal path number describes a structural restriction distinct from bounded transitivity modulator, and neither hardness nor tractability results for one parameter carry over to the other. 
    
        Both transitivity modulator parameters are defined in terms of the reachability graph~$\rcal(\gcal)$. Observe that reachability in~$\gcal$ is transitive if and only if the arc relation of $\rcal(\gcal)$ is transitive.
        The \emph{vertex-deletion distance to transitivity}, denoted $\delta_{\mathrm{vd}}$, is the minimum number of vertices whose deletion from $\rcal(\gcal)$ yields a transitive graph.
        The \emph{arc-modification distance to transitivity}, denoted $\delta_{\mathrm{am}}$, is the minimum number of arcs whose addition or deletion from $\rcal(\gcal)$ yields a transitive graph.
        \begin{lemma}\label{lem:tpn2-large-tm}
            There exists an infinite class of (un)directed, (non-)strict temporal graphs on $n$ vertices with $\tpn=2$ and $n/3\;\le\; \delta_{\mathrm{vd}} \;\le\;2\cdot \delta_{\mathrm{am}}$.
        \end{lemma}
        \begin{proof}
            For $\ell\ge 2$ define a temporal graph $\gcal_\ell$ with vertex set $V_\ell=\{v_1,\dots,v_\ell\} \cup \{b_1, \dots, b_{\ell}\} \cup\{o_1,\dots,o_{\ell}\}$. Hence, $n_\ell=\lvert V_\ell\rvert = 3\ell$.
            The edges of $\gcal_\ell$ are given by the union of two temporal base-paths $P_b$ and $P_o$ on the common vertices $v_1,\dots,v_\ell$.
            The \emph{orange} temporal path~$P_o$ visits $o_{1},v_{1},o_2,v_2, \dots,o_{\ell},v_\ell$, whereas the \emph{blue} temporal path~$P_b$ goes in the opposite direction $v_\ell,b_\ell,\dots,v_2,b_2,v_1,b_1$.
            The time labels are pairwise distinct and strictly increase along the order of each path so that every label of $P_o$ is strictly larger than every label of $P_b$. For the directed version, both paths are directed according to the given traversal directions, while the undirected version is obtained by simply forgetting these orientations. Since the graphs is proper, strict and non-strict temporal reachability coincide.
            See \Cref{fig:tm-tpn-incomp-1} (left) for an example with $\ell=4$.

            The two base-paths show that $\tpn(\gcal_\ell)\leq2$. Towards $\ell\geq2$ note that the footprint of $\gcal_\ell$ is not a path and therefore its temporal edges cannot form a single temporal base-path. Hence, $\tpn(\gcal_\ell)=2$ in all four settings.
            \begin{figure}[t]
                \centering
                \includegraphics[width=0.8\linewidth]{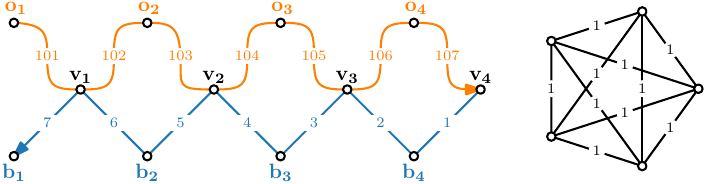}
                \caption{Illustration of $\gcal_4$ for \Cref{lem:tpn2-large-tm} (left) and $\gcal_5$ for \Cref{lem:tm-small-tpn-large} (right). In $\gcal_4$, the orange base-path $P_o$ goes from left to right at the top; the blue base-path $P_b$ goes from right to left at the bottom.}
                \label{fig:tm-tpn-incomp-1}
            \end{figure}
            
            Consider the reachability graph $\rcal(\gcal_\ell)$.
            For every $i\in[\ell]$, the triple $T_i=\{o_i, v_{i}, b_i\}$ witnesses a violation of transitivity: we have $o_i \rightreach v_{i}$ and $v_{i} \rightreach b_i$, but $o_i \not\rightreach b_i$.
            Indeed, $o_i$ is incident only with orange edges, while $b_i$ is incident with only blue edges. Since every orange label is larger than every blue label, no temporal path can use these edges in this order. This is regardless of edge orientation.
            The triples $T_1,\dots,T_{\ell}$ are pairwise disjoint. Hence, any vertex set whose deletion makes $\rcal(\gcal_\ell)$ transitive must remove at least one vertex from every triple $T_i$, which implies
            $\delta_{\mathrm{vd}}(\gcal_\ell)\geq \ell = n/3$.             
            Finally, it holds that $\delta_{\mathrm{vd}}\le 2\cdot\delta_{\mathrm{am}}$, as observed in \cite[Page~5]{casteigts_DistanceTransitivity_2024}, which completes the proof.
        \end{proof}

        \begin{lemma}\label{lem:tm-small-tpn-large}
            There exists an infinite class of (un)directed, (non-)strict temporal graphs on $n$ vertices with $\delta_{\mathrm{vd}}=0=\delta_{\mathrm{am}}$, and $n/2 \leq \tpn$.
        \end{lemma}
        \begin{proof}
            For the undirected setting and $n\ge 2$, define a temporal graph $\gcal_n^{\mathrm U}=(K_n,\lambda)$ with the complete graph on $n$ vertices as footprint and $\lambda(e)=1$ for every edge $e\in E(K_n)$. Its reachability graph is complete bidirected and therefore transitive.
            See \Cref{fig:tm-tpn-incomp-1} (right) for an example with $n=5$.

            For the directed setting, let $\gcal_n^{\mathrm D}$ have the transitive orientation of $K_n$ as its footprint: the vertex set is $\{v_1,\dots,v_n\}$ and the arc $(v_i,v_j)$ exists if and only if $i<j$. Again every arc has label $1$.
            Under both strict and non-strict reachability, the reachability graph of $\gcal_n^{\mathrm D}$ is this transitive tournament.
            Consequently, for $\gcal_n\in\{\gcal_n^{\mathrm U},\gcal_n^{\mathrm D}\}$, we have $\delta_{\mathrm{vd}}(\gcal_n)=0=\delta_{\mathrm{am}}(\gcal_n)$.
            
            It remains to lower-bound $\tpn(\gcal_n)$. Both footprints contain $\binom{n}{2}$ edges or arcs, while every temporal base-path contains at most $n-1$ temporal edges. Covering all temporal edges therefore requires at least $\frac{\binom{n}{2}}{n-1}=\frac{n}{2}$ temporal base-paths. Hence, $\tpn(\gcal_n)\ge n/2$.
            In the strict setting, no temporal base-path can contain two consecutive edges because all labels are equal, and therefore $\tpn(\gcal_n)=\binom{n}{2}$.
            \end{proof}

\section{Hardness for Constant Footprint Parameters}\label{sec:pNPhardness-vc-tw-pw-fvn-planar}
    In this section we present the hardness results for constant footprint parameters. 
    See \Cref{fig:footprint-navigation} for a navigation of the resulting parameter bounds to the subsections containing their proofs. 

    \begin{center}
        \begin{minipage}{\linewidth}
            \centering
            \begin{minipage}{0.9\linewidth}
\begingroup
\definecolor{resultDiagramHard}{rgb}{0.753,0,0.235}
\definecolor{resultDiagramEasy}{rgb}{0.251,0.427,0.212}
\hypersetup{hidelinks}
\newcommand{\diagramSec}[1]{\hyperref[#1]{Sec.~\ref*{#1}}}
\pgfmathsetlengthmacro{\diagramUnit}{\linewidth/402}
\pgfmathsetlengthmacro{\diagramFontSize}{7*\diagramUnit}
\pgfmathsetlengthmacro{\diagramLineSkip}{8.4*\diagramUnit}
\begin{tikzpicture}[
    inner sep=0pt, outer sep=0pt,
    every node/.style={align=center,
        font=\fontsize{\diagramFontSize}{\diagramLineSkip}\selectfont\bfseries},
    hardness/.style={text=resultDiagramHard},
    algorithm/.style={text=resultDiagramEasy}
]
    \node[anchor=north west] (background) at (0,0)
        {\includegraphics[width=\linewidth]{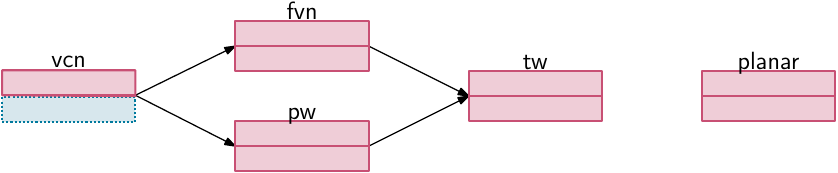}};
    \begin{scope}[x=\diagramUnit,y=-\diagramUnit,overlay]
        \node[hardness] at (33,40.84)
            {\diagramSec{subsec:directed-vc-hardness}};

        \node[hardness] at (145,16.84)
            {\diagramSec{subsec:directed-planar-hardness}};
        \node[hardness] at (145,28.84)
            {\diagramSec{subsec:undirected-planar-hardness}};
        \node[hardness] at (145,64.84)
            {\diagramSec{subsec:directed-planar-hardness}};
        \node[hardness] at (145,76.84)
            {\diagramSec{subsec:undirected-planar-hardness}};
        \node[hardness] at (257,40.84)
            {\diagramSec{subsec:directed-planar-hardness}};
        \node[hardness] at (257,52.84)
            {\diagramSec{subsec:undirected-planar-hardness}};
        \node[hardness] at (369,40.84)
            {\diagramSec{subsec:directed-planar-hardness}};
        \node[hardness] at (369,52.84)
            {\diagramSec{subsec:undirected-planar-hardness}};
    \end{scope}
\end{tikzpicture}%
\endgroup

            \end{minipage}
            \captionof{figure}{Navigation map for the footprint-parameter results in \Cref{sec:pNPhardness-vc-tw-pw-fvn-planar}.
            Each entry links to the subsection containing the corresponding result. }
            \label{fig:footprint-navigation}
        \end{minipage}
    \end{center}

        All results are obtained by a polynomial-time reduction from \textup{\textsc{Directed Cut}} where we are given a loopless directed graph $D=(V,A)$ and integer $k\in\mathbb{N}_{\geq 1}$ and the goal is to find a set of vertices $S\subseteq V$ such that there are at least $k$ arcs from $S$ to $V\setminus S$.
        This problem is known to be \NP-hard \cite{farley_directedmaximalcut_1982}.
        For the remainder of this section, we fix an ordering $\leq$ on $V$ and, for every $v\in V$, denote by $r(v)$ the rank of $v$ in this ordering. We additionally assume that $|V|\geq 2$ and $|A|\geq 3$, since instances not satisfying these assumptions can be solved in polynomial time.

        We use the following equivalent characterization of the \textup{\textsc{Directed Cut}} problem: A set $F\subseteq A$ is contained in a directed cut if and only if, for every two distinct arcs $e,f\in F$, we have $\operatorname{head}(e)\neq\operatorname{tail}(f)$ and $\operatorname{head}(f)\neq\operatorname{tail}(e)$. We call these two inequalities the \emph{directed-cut condition}. In each construction, we encode this condition by mutual reachability between the corresponding arc-vertices.

        \bigparagraph{General construction.}
        We represent every arc $e$ of the input graph by an arc-vertex $x_e$, and arrange that, for distinct arcs $e,f$, the vertex $x_e$ reaches $x_f$ exactly when $\operatorname{head}(e)\neq\operatorname{tail}(f)$. Thus, mutual reachability checks the two inequalities of the directed-cut condition.
        With the fixed rank function $r$, we realize these inequalities using two \emph{comparison portals}. The first will provide a temporal path when $r(\operatorname{head}(e))<r(\operatorname{tail}(f))$ and the second when $r(\operatorname{head}(e))>r(\operatorname{tail}(f))$.
        All three reductions follow the same general layout (see \Cref{fig:directed-cut-framework}):
        \begin{enumerate}
            \item \emph{Two comparison portals.} We implement the portals either by the rank-comparison gadgets $\mathcal Q_<$ and $\mathcal Q_>$, defined in \Cref{def:rank-comparison}, or by two individual comparison vertices $a,b$.
            \item \emph{Connections to portals.} We use either direct in/out edges or subdivided in/out edges. The time labels must ensure that comparison paths cannot be chained to create unintended reachability.
            \item \emph{Auxiliary edges.} We add suitably early and late auxiliary edges so that every valid selection of arc-vertices extends to a closed temporally connected set containing all auxiliary vertices, without changing mutual reachability between arc-vertices.
        \end{enumerate}
        %
        \begin{figure}[t]
            \centering
            \includegraphics[width=0.50\linewidth]{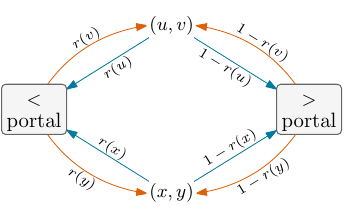}
            \caption{The general comparison framework for an ordered pair of arcs $e=(u,v)$ and $f=(x,y)$. The arc-vertices $x_e$ and $x_f$ are depicted as $(u,v)$ and $(x,y)$ respectively. Each comparison portal is a rank-comparison gadget or a single vertex. Blue in-connections and orange out-connections are direct or subdivided. Arrows indicate the intended temporal traversal. The additional edges that ensure auxiliary connectivity are omitted.}
            \label{fig:directed-cut-framework}
        \end{figure}
        To avoid repeating the same correctness argument for every construction, we formalize it in the following general observation.
        \begin{lemma}[Directed Cut $\leftrightarrow$ TCC] \label{lem:directed-cut-encoding}
            Given a loopless directed graph $D=(V,A)$ and an integer $k\geq 1$, let $\mathcal{G}$ be a temporal graph whose vertex set is the disjoint union of a set $Z$ and $\{x_e:e\in A\}$, where the vertices $x_e$ are pairwise distinct. Suppose that the following properties hold.
            \begin{enumerate}
                \item\label{prop:cut-necessity} If $e,f\in A$ are distinct and $x_e$ and $x_f$ are mutually reachable in \gcal,
                then
                \begin{align}
                    \operatorname{head}(e)\neq\operatorname{tail}(f) \quad \text{and} \quad \operatorname{head}(f)\neq\operatorname{tail}(e). \tag{$\star$}
                \end{align}
                \item For every nonempty set $F\subseteq A$ satisfying $(\star)$ for every pair of distinct $e,f\in F$, the temporal graph induced by $Z\cup\{x_e:e\in F\}$ is temporally connected. \label{prop:cut-sufficiency}
            \end{enumerate}
            Then, $D$ has a directed cut containing at least $k$ arcs if and only if $\mathcal{G}$ has an \otcc of size at least $|Z|+k$, and equivalently if and only if $\mathcal{G}$ has a \ctcc of size at least $|Z|+k$.
        \end{lemma}
        \begin{proof}
            $(\Rightarrow)$\quad
            Assume there exists a directed cut containing a set $F$ of at least $k$ arcs. By definition,
            $(\star)$ is satisfied for every pair of distinct $e,f\in F$. By Property~\ref{prop:cut-sufficiency}, $Z\cup\{x_e:e\in F\}$ is temporally connected in its induced temporal graph and has size at least $|Z|+k$. It is therefore contained in a \ctcc and, independently, an \otcc of at least size $|Z|+k$.
            
            $(\Leftarrow)$\quad 
            Let $S$ be an open or closed tcc of size at least $|Z|+k$, and set $F:=\{e\in A:x_e\in S\}$.
            Since at most $|Z|$ vertices of $S$ can come from $Z$, we have $|F|\geq k$.
            Since any two vertices of $S$ are mutually reachable in $\mathcal{G}$, Property~\ref{prop:cut-necessity} implies that $(\star)$ holds for every pair of distinct $e,f\in F$.
            Thus by the characterization of directed cuts, $F$ is contained in a directed cut of size at least $k$.
        \end{proof}

    \subsection{Directed Graphs with Constant Vertex Cover}\label{subsec:directed-vc-hardness}
    In this section we prove \NP-hardness for \otcc and \ctcc on simple directed graphs with constant vertex cover.
    
    Before we present our main reduction, we introduce rank-comparison gadgets to encode numerical comparisons by temporal reachability. We use them in the construction below and return to them for the vertex-cover hardness result in \Cref{subsec:undirected-vc-hardness}.
            
    \begin{figure}[htb]
        \centering
        \includegraphics[page=1,width=\linewidth]{Figures/rn.pdf}
        \caption{Rank-comparison gadgets $\mathcal Q_{\leq}$ and $\mathcal Q_{>}$ for $X=\{x_1,x_2,x_3\}$, $Y=\{y_1,y_2,y_3,y_4\}$, $\mu=3$, $r_1(x_i)=i$, and $(r_2(y_1),r_2(y_2),r_2(y_3),r_2(y_4))=(1,2,2,3)$. The right gadget reverses both ranks and uses the strict comparison. Here $t<t^+<t+1$; repeated schematic labels are separated by tie-breaking. In the directed version, edges point from $X$ through $a_1,p_\ell,a_2$ to $Y$; orientations are omitted.}
        \label{fig:rankcomparison-gadget}
    \end{figure}
    \begin{definition}[Rank-comparison gadget]\label{def:rank-comparison}
        Let $X$ and $Y$ be two (not necessarily distinct) vertex sets, let $\mu\in\mathbb N$, and let ${r_1}\colon X\to[\mu]$ and ${r_2}\colon Y\to[\mu]$ be functions, which we call {\em rank functions}.
        Let $\prec_X$ and $\prec_Y$ be fixed orderings of $X$ and $Y$, respectively. 

        The \emph{rank-comparison gadget} $\mathcal Q_{\leq}(X,Y;{r_1},{r_2})$ is defined as follows.
        Its vertex set is
        \[
            V=(X\cup Y)\mathbin{\dot\cup}\{a_1,a_2\}
                \mathbin{\dot\cup}\{p_\ell:\ell\in[\mu]\},
        \]
        where $a_1,a_2$ are the \emph{cover vertices} and the $p_\ell$ are the \emph{threshold vertices}. These new vertices are pairwise distinct and fresh for each copy of the gadget.
        The static edge set of the directed version is
        \[
                E=\{(x,a_1):x\in X\}
                    \cup\{(a_1,p_\ell):\ell\in[\mu]\}\cup\{(p_\ell,a_2):\ell\in[\mu]\}
                    \cup\{(a_2,y):y\in Y\}.
        \]
        The undirected version is obtained by forgetting the edge orientations.

        Every edge receives a distinct time label. The labels are arranged in two blocks $L^-$ and $L^+$, with every label in $L^-$ smaller than every label in $L^+$. Within each block, edges appear in the order specified below. 
        \begin{itemize}
            \item In $L^-$, process the ranks $\ell=1,\ldots,\mu$ in increasing order. For each $\ell$, first place all edges $x a_1$ with ${r_1}(x)=\ell$, ordered by $\prec_X$, and then place the edge $a_1 p_\ell$.
            \item In $L^+$, again process the ranks $\ell=1,\ldots,\mu$ in increasing order. For each $\ell$, first place the edge $p_\ell a_2$, and then place all edges $a_2 y$ with ${r_2}(y)=\ell$, ordered by $\prec_Y$.
        \end{itemize} 
        The strict variant $\mathcal Q_{<}(X,Y;{r_1},{r_2})$ uses the same vertices and static edges. Within each rank $\ell$ in $L^-$, reverse the order of the two groups: first place $a_1 p_\ell$, and then all edges $x a_1$ with ${r_1}(x)=\ell$, ordered by $\prec_X$. The ranks are still processed in increasing order, and $L^+$ is unchanged.
    \end{definition}
    Refer to \Cref{fig:rankcomparison-gadget} for an illustration. 
    For a rank function $r\colon W\to[\mu]$, let $r^{\mathrm{rev}}\colon W\to[\mu]$ be defined by $r^{\mathrm{rev}}(w):=\mu+1-r(w)$.
    We write $\mathcal Q_{>}(X,Y;{r_1},{r_2}):=\mathcal Q_{<}(X,Y;{r_1^{\mathrm{rev}}},{r_2^{\mathrm{rev}}})$ and $\mathcal Q_{\geq}(X,Y;{r_1},{r_2}):=\mathcal Q_{\leq}(X,Y;{r_1^{\mathrm{rev}}},{r_2^{\mathrm{rev}}})$.

    \begin{lemma}[Reachability through a rank-comparison gadget]\label{lem:rank-comparison}
        For distinct $x\in X$ and $y\in Y$ and a comparison $\diamond\in\{\leq,<,\geq,>\}$, the directed gadget satisfies
        \[
            x\rightreach y\text{ in }\mathcal Q_{\diamond}(X,Y;r_1,r_2)
            \quad\Longleftrightarrow\quad r_1(x)\diamond r_2(y).
        \]
        Each gadget is simple and proper and has a vertex cover of size at most two. These properties hold under both strict and non-strict reachability.
    \end{lemma}
    \begin{proof}
        Every temporal path between distinct terminals must have the form $x,a_1,p_\ell,a_2,y$. 
        In $\mathcal Q_{\leq}(X,Y;{r_1},{r_2})$, the edge $x a_1$ precedes $a_1 p_\ell$ exactly when $r_1(x)\leq\ell$, and $p_\ell a_2$ precedes $a_2 y$ exactly when $\ell\leq r_2(y)$. Such a path therefore exists precisely when $\exists\ell\in[\mu]: r_1(x)\leq\ell\leq r_2(y)$, which is equivalent to $r_1(x)\leq r_2(y)$.

        In $\mathcal Q_{<}(X,Y;{r_1},{r_2})$ the edge $a_1 p_\ell$ appears {\em before} the edges $\{x a_1\colon x\in X, {r_1}(x)=\ell\}$. This makes $x$ reach $p_{\ell'}$ only when ${r_1}(x)<\ell'$, and a temporal $x$-$y$-path through the gadget exists exactly when ${r_1}(x)<{r_2}(y)$.
        
        In $\mathcal{Q}_>(X,Y,r_1,r_2)=\mathcal Q_{<}(X,Y;{r_1^{\mathrm{rev}}},{r_2^{\mathrm{rev}}})$, the rank functions are reversed. As a result, a temporal $x$-$y$-path exists exactly when there is an $\ell\in[\mu]$ with ${r_1^{\mathrm{rev}}}(x)<\ell\leq{r_2^{\mathrm{rev}}}(y)$. Equivalently, ${r_2}(y)\leq\mu+1-\ell<{r_1}(x)$, which holds exactly when ${r_1}(x)>{r_2}(y)$. The same argument applies to $\mathcal{Q}_{\geq}(X,Y,r_1,r_2)$.

        Every static edge receives one label, all labels are distinct, and $\{a_1,a_2\}$ meets every edge. Thus the gadget is simple and proper, has a vertex cover of size at most two, and its temporal paths are valid under either reachability convention.
    \end{proof}

    The lemma suggests the following two ways to encode equality and inequality by combining gadgets. In each application, the construction must ensure that every temporal path between the vertices under consideration traverses one of the indicated gadgets from its input set to its output set along a path of the form $x,a_1,p_\ell,a_2,y$. In particular, paths using additional edges or combining parts of different gadgets must not create further reachability. This condition is verified separately in the reductions.
    \begin{itemize}
        \item For disjoint $X$ and $Y$, we combine $\mathcal Q_{\leq}(X,Y;{r_1},{r_2})$ and $\mathcal Q_{\leq}(Y,X;{r_2},{r_1})$.
        Under the stated condition, $x\in X$ reaches $y\in Y$ exactly when ${r_1}(x)\leq{r_2}(y)$, and $y$ reaches $x$ exactly when ${r_2}(y)\leq{r_1}(x)$.
        Thus, $x$ and $y$ are mutually reachable exactly when ${r_1}(x)={r_2}(y)$.
        \item For rank functions ${r_1},{r_2}\colon X\to[\mu]$, combine $\mathcal Q_{<}(X,X;{r_1},{r_2})$ and $\mathcal Q_{>}(X,X;{r_1},{r_2})$.
        Under the stated condition, for distinct $x,x'\in X$, the vertex $x$ reaches $x'$ exactly when ${r_1}(x)<{r_2}(x')$ or ${r_1}(x)>{r_2}(x')$, that is, when ${r_1}(x)\neq{r_2}(x')$.
        Consequently, $x$ and $x'$ are mutually reachable exactly when ${r_1}(x)\neq{r_2}(x')$ and ${r_1}(x')\neq{r_2}(x)$.
    \end{itemize}

                Having established the auxiliary gadgets, the remaining of the section is devoted to prove the following theorem.
        \directedVCFour*
        Given an instance $(D,k)$ of \textup{\textsc{Directed Cut}}, we construct an instance $(\mathcal G, k')$ of \textup{\textsc{oTCC}} and \textup{\textsc{cTCC}} such that $D$ has a directed cut of size at least $k$ if and only if the directed temporal graph $\mathcal G$ has an open (resp. closed) tcc of size at least $k'$.

        \bigparagraph{Intuition.}
        We implement the two comparison portals of \Cref{fig:directed-cut-framework} by the rank-comparison gadgets $\mathcal Q_<$ and $\mathcal Q_>$, with cover vertices $a_1,a_2$ and $b_1,b_2$, respectively, and direct in/out edges (see \Cref{fig:D-vc}). The first gadget provides a temporal path from $x_e$ to $x_f$ when $r(\operatorname{head}(e))<r(\operatorname{tail}(f))$, and the second when $r(\operatorname{head}(e))>r(\operatorname{tail}(f))$.
        All input blocks precede all output blocks, so a comparison path reaching an arc-vertex cannot leave it for another comparison.

        For auxiliary connectivity, early edges give auxiliary vertices access to the comparison paths, while late edges lead from the arc-vertices to $a_2$ and from $a_2$ to the other auxiliary vertices. No arc-vertex is incident with an early auxiliary edge, and the late edges cannot lead back to an arc-vertex after the comparison exits have passed. Thus, these edges do not create additional reachability between arc-vertices. All additional edges are incident with the same four cover vertices, so auxiliary connectivity does not increase the required vertex cover.

        \bigparagraph{Construction.} (See \Cref{fig:D-vc} for an overview.)
        For every arc $e\in A$, introduce an \textit{arc-vertex} $x_e$, and let $X:=\{x_e:e\in A\}$. Set ${r_1}(x_e):=r(\operatorname{head}(e))$ and ${r_2}(x_e):=r(\operatorname{tail}(e))$.
        Order $X$ lexicographically by the ranks of the head and then the tail of the corresponding arcs, and use this ordering on both sides of the rank-comparison gadgets. We add $\mathcal Q_{<}(X,X;{r_1},{r_2})$ and $\mathcal Q_{>}(X,X;{r_1},{r_2})$, and denote their cover vertices by $a_1,a_2$ and $b_1,b_2$ and their threshold vertices by $p^1_\ell$ and $p^2_\ell$, respectively. Let $Z=\{a_1,a_2,b_1,b_2\}\cup\{p^1_\ell,p^2_\ell:\ell\in[|V|]\}$ be the set of all auxiliary vertices.
        The lifetime is partitioned into four consecutive blocks $T_1,T_2,T_3,T_4$, where every label in $T_r$ precedes every label in $T_{r+1}$. Every static edge will receive exactly one unique time label, so it suffices to specify its
        block and its position within that block. 

        In $T_1$, for every $z\in Z\setminus\{a_2\}$ and $c\in\{a_1,b_1\}$, add the edge $(z,c)$ whenever $z\neq c$. These edges appear in an arbitrary order.

        In $T_2$, place the edges of the $L^-$ blocks of $\mathcal Q_{<}(X,X;{r_1},{r_2})$ and $\mathcal Q_{>}(X,X;{r_1},{r_2})$ as specified in Definition~\ref{def:rank-comparison}. In $T_3$, place the edges of their $L^+$ blocks. Thus, every edge in an $L^-$ block precedes every edge in an $L^+$ block, as required by the definition.

        In $T_4$, place first all edges $(x_e,a_2)$ and then all edges $(a_2,z)$ for $z\in Z\setminus\{a_2\}$. The order inside each of these two families is arbitrary.
        \begin{figure}[t]
            \centering
            \includegraphics[width=0.7\textwidth]{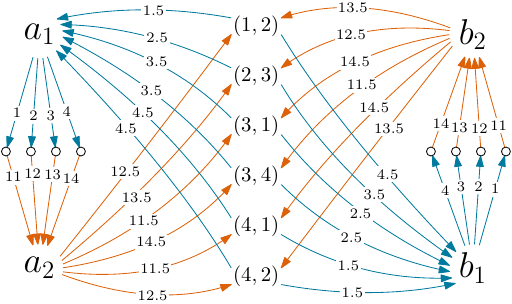}
            \caption{The two rank-comparison gadgets in the directed vertex cover construction: $\mathcal Q_<$ has cover vertices $a_1,a_2$, and $\mathcal Q_>$ has cover vertices $b_1,b_2$. Central vertices represent arc-vertices. The additional edges in $T_1$ and $T_4$ are omitted.}
            \label{fig:D-vc}
        \end{figure}

        Finally, we set the target size to $k'=k+|Z|=k+2|V|+4$. This completes the construction of the instance $(\mathcal{G},k')$ which is clearly computable in polynomial time, and $\mathcal G$ is obviously simple and proper.
        For the application of Lemma~\ref{lem:directed-cut-encoding}, we will use the set $Z$ defined above.
        
        We first prove that the vertex cover number is exactly four.
        \begin{lemma} \label{lem:directed-vc-four}
            The footprint $G$ of $\mathcal{G}$ has vertex cover number four.
        \end{lemma}
        \begin{proof}
            Let $C=\{a_1,a_2,b_1,b_2\}$. Every edge introduced in $T_1$ has an endpoint in $\{a_1,b_1\}$, every edge of a comparison gadget has an endpoint in $C$, and every edge added in $T_4$ has endpoint $a_2$. Consequently, $C$ is a vertex cover and $\vc(G)\leq 4$.
            For the converse, consider the four footprint edges $(a_1,p^1_1)$, $(a_2,p^1_2)$, $(b_1,p^2_1)$, and $(b_2,p^2_2)$. These edges form a matching of size four. Every vertex cover contains at least one endpoint of each edge of a matching, and therefore $\vc(G)\geq 4$. We conclude that $\vc(G)=4$.
        \end{proof}

        \begin{lemma} \label{lem:edge-vertex-reachability}
            For every two distinct arcs $e,f\in A$, the arc-vertex $x_e$ reaches the arc-vertex $x_f$ in $\mathcal{G}$ if and only if $\operatorname{head}(e)\neq\operatorname{tail}(f)$.
        \end{lemma}
        \begin{proof}
            By Lemma~\ref{lem:rank-comparison}, $\mathcal Q_{<}(X,X;{r_1},{r_2})$ contains a temporal path from $x_e$ to $x_f$ if and only if $r(\operatorname{head}(e))<r(\operatorname{tail}(f))$.
            Similarly, $\mathcal Q_{>}(X,X;{r_1},{r_2})$ contains a temporal path from $x_e$ to $x_f$ if and only if $r(\operatorname{head}(e))>r(\operatorname{tail}(f))$.
            It remains to check that the edges in $T_1$ and $T_4$ do not create additional reachability between edge-vertices.
            No arc-vertex is incident with an edge in $T_1$. If a temporal path uses any edge $(y,a_2)$ in $T_4$, then all remaining edges on the path are directed from $a_2$ to auxiliary vertices, and the path cannot reach an arc-vertex. Hence every temporal path from one arc-vertex to another is
            provided by one of the two comparison gadgets.
            Combining the two cases, $x_e$ reaches $x_f$ if and only if $r(\operatorname{head}(e))\neq r(\operatorname{tail}(f))$. Since $r$ is injective, this holds if and only if $\operatorname{head}(e)\neq\operatorname{tail}(f)$.
        \end{proof}
        Observe that if we apply \Cref{lem:edge-vertex-reachability} to {\em both} ordered pairs $(e,f)$ and $(f,e)$, we obtain that for every two distinct arcs $e,f\in A$, the edge-vertices $x_e$ and $x_f$ are {mutually reachable} if and only if $\operatorname{head}(e)\neq\operatorname{tail}(f)$ and $\operatorname{head}(f)\neq\operatorname{tail}(e)$. 
        {Next, we show that every auxiliary vertex is mutually reachable with every vertex of $\mathcal G$, with witnessing paths contained in the sets specified below.}

        \begin{lemma} \label{lem:auxiliary-connectivity}
            For every arc-vertex $x_e$ and every auxiliary vertex $z\in Z$, the vertices $x_e$ and $z$ are {mutually reachable} via temporal paths whose vertices belong to $Z\cup\{x_e\}$. Moreover, {$\mathcal{G}[Z]$ is temporally connected}.
        \end{lemma}
        \begin{proof}
            First, every vertex $z\in Z$ reaches $a_2$ within $Z$ by the time the edge $(p^1_1,a_2)$ appears in $T_3$. For $z=a_2$, use the trivial path; for $z=p^1_1$, use $(p^1_1,a_2)$. For every other $z$, use $(z,a_1)$ in $T_1$ if $z\neq a_1$, followed by $(a_1,p^1_1)$ in $T_2$ and $(p^1_1,a_2)$ in $T_3$.
            Since $(a_2,x_e)$ appears after $(p^1_1,a_2)$ in $T_3$, every $z\in Z$ reaches $x_e$ within $Z\cup\{x_e\}$.
            Conversely, $x_e$ uses $(x_e,a_2)$ in $T_4$ and, if $z\neq a_2$, the later edge $(a_2,z)$ to reach $z$. All vertices on these paths belong to $Z\cup\{x_e\}$.

            Finally, let $z,z'\in Z$. If $z=z'$, use the trivial path. Otherwise, if $z'\in\{a_1,p^1_1,a_2\}$ lies on the path from $z$ to $a_2$ described above, its prefix reaches $z'$ within $Z$. In every other case, append the edge $(a_2,z')$ in $T_4$. Hence, $\mathcal G[Z]$ is temporally connected.
        \end{proof}

        \begin{proof}[Proof of \Cref{thm:directed-vc-four}]
            By \Cref{lem:edge-vertex-reachability}, two edge-vertices $x_e$ and $x_f$ can be mutually reachable only if the corresponding arcs satisfy the {directed-cut condition}.
            \Cref{lem:auxiliary-connectivity} verifies the second property of \Cref{lem:directed-cut-encoding}. Applying that lemma proves the correctness of the reduction for both \textup{\textsc{oTCC}} and \textup{\textsc{cTCC}}.
            Additionally, the construction is polynomial, and \Cref{lem:directed-vc-four} shows that its footprint has vertex cover number four. 
            Thus, we have shown \Cref{thm:directed-vc-four}.
        \end{proof}

    \subsection[Directed Planar Graphs with Constant Footprint Parameters]{Directed Planar Graphs with Constant Treewidth, Pathwidth, and Feedback Vertex Number}\label{subsec:directed-planar-hardness}
                The footprint of the previous construction satisfies $\tw(G)=4$, $\pw(G)=4$, $\fvn(G)=3$, and it is not planar (because of the edges in $T_1$ and $T_4$ which are needed to obtain the result for \ctcc while keeping the \vc low).
        We next present a simple and proper construction with a planar footprint and stronger bounds on the treewidth, pathwidth, and feedback vertex number. For $\fvn$ and $\tw$ the bounds will be tight, since both problems are in \PP on forests which have treewidth 1 and feedback vertex number 0.
        The source problem and the general correctness argument remain unchanged; we simply describe the modified temporal graph, show its structural properties, and verify the two properties of \Cref{lem:directed-cut-encoding}.
        \directedPlanarHardness*
        Given an instance $(D,k)$ of \textup{\textsc{Directed Cut}}, we construct an instance $(\mathcal G, k')$ of \textup{\textsc{oTCC}} and \textup{\textsc{cTCC}} such that $D$ has a directed cut of size at least $k$ if and only if the directed temporal graph $\mathcal G$ has an open (resp. closed) tcc of size at least $k'$.

        \bigparagraph{Intuition.}
        To obtain planar footprints with small treewidth, pathwidth, and feedback vertex number, we implement the two comparison portals by individual vertices $a,b$ and use subdivided in/out connections (see \Cref{fig:D-planar-tw-pw-fvn}). At $a$, arrivals from $x_e$ and departures towards $x_f$ are ordered by increasing head and tail ranks, realizing $r(\operatorname{head}(e))<r(\operatorname{tail}(f))$. At $b$, decreasing ranks realize the reverse inequality. Within each rank, departures precede arrivals, which excludes equality.

        The connections use an asymmetric implementation of the subdivided option to keep the graph simple while minimizing the footprint parameters. At $a$, $x_e$ enters directly and is reached via a subdivision vertex $\ell_e$; at $b$, the in- and out-connections share a vertex $r_e$, using $x_e\to r_e\to b$ and $b\to r_e\to x_e$. The opposite arcs can receive different labels without repeating a directed edge. Crucially, the final edges $(\ell_e,x_e)$ and $(r_e,x_e)$ appear only after both comparison blocks. Thus, a path reaching an arc-vertex cannot leave it for another comparison.

        For auxiliary connectivity, we add early and late edges and one further vertex $y_e$ per arc. The early edges allow all auxiliary vertices to reach $a$ before the comparisons. The late edges lead from $a$ to $r_e$ and $y_e$, and from $y_e$ to $\ell_e$. In particular, the late path $a\to y_e\to\ell_e$ avoids repeating the comparison edge $a\to\ell_e$, preserving simplicity. These additions also preserve the desired footprint structure: deleting $a$ leaves the paths $y_e-\ell_e-x_e-r_e-b$, sharing only $b$, and hence a forest.

        \bigparagraph{Construction (see \Cref{fig:D-planar-tw-pw-fvn}).}
        For every arc $e\in A$, introduce an \textit{arc-vertex} $x_e$ and three auxiliary vertices $y_e,\ell_e,r_e$. Additionally, introduce two global auxiliary vertices $a$ and $b$.
        Let $Z=\{a,b\}\cup\{y_e,\ell_e,r_e:e\in A\}$.
        The lifetime is partitioned into five consecutive blocks $T_1,T_2,T_3,T_4,T_5$, where every label in $T_r$ precedes every label in $T_{r+1}$.
        Every directed edge will receive exactly one time label.

        In $T_1$, add the edges $(\ell_e,a)$, $(r_e,a)$, and $(y_e,a)$ for every $e\in A$, in an arbitrary order. Then add $(b,a)$, followed by $(a,b)$.

        In $T_2$, we process the ranks in increasing order. For every $\ell=1,\ldots,|V|$, first place all edges $(a,\ell_f)$ for which $r(\operatorname{tail}(f))=\ell$, and then place all edges $(x_e,a)$ for which $r(\operatorname{head}(e))=\ell$. The edges within each of these two families may be ordered arbitrarily.

        In $T_3$, we process the ranks in decreasing order. For every $\ell=|V|,\ldots,1$, first place all edges $(b,r_f)$ for which $r(\operatorname{tail}(f))=\ell$. Then, for every $e$ with $r(\operatorname{head}(e))=\ell$, place the consecutive path $(x_e,r_e),(r_e,b)$. These paths may be ordered arbitrarily for arcs with the same rank.

        In $T_4$, add the edges $(\ell_e,x_e)$ and $(r_e,x_e)$ for every $e\in A$, in an arbitrary order.

        In $T_5$, add the edges $(a,r_e)$ and $(a,y_e)$ for every $e\in A$, in an arbitrary order, followed by all edges $(y_e,\ell_e)$ in an arbitrary order.

        \begin{figure}[t]
            \centering
            \includegraphics[width=\textwidth]{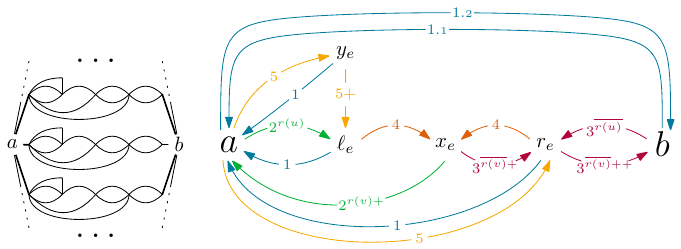}
            \caption{The directed planar construction. \\
            Right: the gadget for an arc $e=(u,v)$.
            The labels indicate temporal blocks $T_1,\ldots,T_5$. Superscripts as for $2^{r(u)}$ encode rank order, using $r$ in $T_2$ and reversed ranks $\overline{r(w)}=|V|+1-r(w)$ in $T_3$. Plus signs indicate local precedence, \ie $5<5+$. 
            Labels are schematic and ties are broken to obtain distinct times while preserving the temporal paths.
            The edges into $x_{e}$ occur in $T_4$, after both comparison blocks, preventing comparison paths from being chained. The edges in $T_1$ and $T_5$ provide auxiliary connectivity.\\
            Left: a schematic planar layout of the arc gadgets, sharing only $a,b$.}
            \label{fig:D-planar-tw-pw-fvn}
        \end{figure}

        Finally, we set the target size to $k'=k+|Z|=k+3|A|+2$. This completes the construction of the instance $(\mathcal{G},k')$, which is clearly computable in polynomial time. Since every directed edge appears exactly once, and all time labels are distinct, $\mathcal G$ is simple and proper.
        For the application of \Cref{lem:directed-cut-encoding}, we use the set $Z$ defined above.
        We first prove the structural properties of $G$.

        \begin{lemma} \label{lem:planar-footprint}
            The footprint $G$ is planar and satisfies $\mathsf{tw}(G)=2$, $\mathsf{pw}(G)=3$, and $\mathsf{fvn}(G)=1$.
        \end{lemma}
        \begin{proof}
            We build a tree decomposition of $G$. For every $e\in A$, let $B_e^1=\{a,b,r_e\}$, $B_e^2=\{a,r_e,x_e\}$, $B_e^3=\{a,x_e,\ell_e\}$, and $B_e^4=\{a,\ell_e,y_e\}$. Arrange the bags $B_e^1$ in an arbitrary order along a path and, for each $e$, attach the path $B_e^2,B_e^3,B_e^4$ to the corresponding bag $B_e^1$.
            Then every edge of $G$ is contained in a bag.
            Moreover, the bags containing any fixed vertex form a connected subtree: the vertex $a$ occurs in every bag, the vertex $b$ occurs in all bags $B_e^1$, and each of $y_e,\ell_e,x_e,r_e$ occurs only in the adjacent bags associated with $e$. So, this forms a tree decomposition of width two and $\tw(G)\leq2$.
            Since $G$ contains the triangle $a-y_e-\ell_e-a$, $\mathsf{tw}(G)\geq2$. Consequently, $\mathsf{tw}(G)=2$.  Since graphs with treewidth 2 are planar, so is our footprint $G$.

            Deleting $a$ leaves the tree formed by the paths $y_e-\ell_e-x_e-r_e-b$ over all $e\in A$, sharing only the vertex $b$. Hence, $\mathsf{fvn}(G)\leq1$. The triangle $a-y_e-\ell_e-a$ shows that $G$ is not a forest, so $\mathsf{fvn}(G)\geq1$. Consequently, $\mathsf{fvn}(G)=1$.

            A path decomposition of width three is obtained by taking, for every $e\in A$, the three consecutive bags $\{a,b,r_e,x_e\}$, $\{a,b,x_e,\ell_e\}$, and $\{a,b,\ell_e,y_e\}$, and concatenating these triples in an arbitrary order. Hence, $\mathsf{pw}(G)\leq3$. For the matching lower bound, choose three distinct arcs $e_1,e_2,e_3\in A$, which exist since $|A|\geq3$. Retaining only the edges of the three paths $a-x_{e_i}-r_{e_i}-b$ for $i\in\{1,2,3\}$ gives a subgraph isomorphic to $E_2$ in \cite[Theorem~3.1 and Figure~3]{barat_structuregraphs_2009}. This graph is an excluded minor for pathwidth at most two. Thus, $\pw(G)\geq3$, and hence $\pw(G)=3$.
        \end{proof}

        Now the reachability property.
        \begin{lemma}\label{lem:planar-reachability}
            For every two distinct arcs $e,f\in A$, the arc-vertex $x_e$ reaches the arc-vertex $x_f$ in $\mathcal G$ if and only if $\operatorname{head}(e)\neq\operatorname{tail}(f)$.
        \end{lemma}
        \begin{proof}
            Consider first the temporal paths through at $a$. A path starting in $x_e$ can enter $a$ when rank $r(\operatorname{head}(e))$ is processed and leave $a$ towards $\ell_f$ when rank $r(\operatorname{tail}(f))$ is processed. Since, at equal ranks, the edges leaving $a$ occur before the edges entering $a$, the temporal path $x_e-a-\ell_f-x_f$ exists exactly when $r(\operatorname{head}(e))<r(\operatorname{tail}(f))$.
            At $b$, the ranks are processed in decreasing order and the edges leaving $b$ again occur before the paths entering $b$ at equal ranks. Hence, the temporal path $x_e-r_e-b-r_f-x_f$ exists exactly when $r(\operatorname{head}(e))>r(\operatorname{tail}(f))$.
            In total, $x_e$ reaches $x_f$ in $\mathcal G$ if $\operatorname{head}(e)\neq\operatorname{tail}(f)$.

            It remains to check that no other temporal path creates additional reachability between edge-vertices. The only edges leaving $x_e$ are its two comparison entries, while the edges $(\ell_f,x_f)$ and $(r_f,x_f)$ entering $x_f$ occur in $T_4$, after both comparison blocks. Consequently, once a path reaches $x_f$, every comparison edge leaving $x_f$ has already passed, and the path cannot continue. The edges in $T_5$ only lead from $a$ to $r_f$ and $y_f$, and from $y_f$ to $\ell_f$, after the edges from $\ell_f$ and $r_f$ into $x_f$ have passed. Hence, every temporal path between two distinct edge-vertices is provided by one of the two comparisons described above.
        \end{proof}

        By applying this lemma to the ordered pairs $(e,f)$ and $(f,e)$, we obtain that for every two distinct arcs $e,f\in A$, the edge-vertices $x_e$ and $x_f$ are {mutually reachable} if and only if $\operatorname{head}(e)\neq\operatorname{tail}(f)$ and $\operatorname{head}(f)\neq\operatorname{tail}(e)$.
        {Next, we show that every auxiliary vertex is mutually reachable with every vertex of $\mathcal G$.}
        \begin{lemma} \label{lem:planar-closed-connectivity}
            For every arc-vertex $x_e$ and every auxiliary vertex $z\in Z$, the vertices $x_e$ and $z$ are {mutually reachable} via temporal paths whose vertices belong to $Z\cup\{x_e\}$. Moreover, {$\mathcal{G}[Z]$ is temporally connected}.
        \end{lemma}
        \begin{proof}
            First consider any vertex $z\in Z=\{a,b\}\cup\{y_e,\ell_e,r_e:e\in A\}$. If $z\neq a$, then $z$ reaches $a$ in $T_1$, while for $z=a$ the path starts at $a$. Then $z$ can use $(a,\ell_e)$ in $T_2$, followed by $(\ell_e,x_e)$ in $T_4$. Thus, $z$ reaches $x_e$.
            Conversely, any $x_e$ reaches $a$ directly in $T_2$ and reaches $b$ through the path $x_e-r_e-b$ in $T_3$. After reaching $a$, the vertex $x_e$ reaches every $r_f$ and $y_f$ directly in $T_5$, and every $\ell_f$ through the path $a-y_f-\ell_f$ in $T_5$. Thus, $x_e$ reaches every vertex of~$Z$. All vertices on these paths belong to $Z\cup\{x_e\}$.
            Finally, consider $z,z'\in Z$. As before, $z$ reaches $a$ in $T_1$. From $a$, the vertex $b$ is reached in $T_1$, every $\ell_f$ is reached in $T_2$, and every $r_f$ and $y_f$ is reached in $T_5$. Hence, {$\mathcal{G}[Z]$ is temporally connected}.
        \end{proof}
        \begin{proof}[Proof of \Cref{thm:directed-planar}]
            By \Cref{lem:planar-reachability}, two edge-vertices $x_e$ and $x_f$ are {mutually reachable} only if the corresponding arcs satisfy the {directed-cut condition}. \Cref{lem:planar-closed-connectivity} verifies the second property of \Cref{lem:directed-cut-encoding}. Applying that lemma proves the correctness of the reduction for both \textup{\textsc{oTCC}} and \textup{\textsc{cTCC}}.
            Additionally, the construction is polynomial and \Cref{lem:planar-footprint} gives all claimed structural bounds.
            Thus, we have shown \Cref{thm:directed-planar}.
        \end{proof}

    \subsection[Undirected Planar Graphs with Constant Footprint Parameters]{Undirected Planar Graphs with Constant Treewidth, Pathwidth, and Feedback Vertex Number}\label{subsec:undirected-planar-hardness}
                We extend the construction from the previous section to obtain similar hardness bounds for undirected temporal graphs. It is proper, planar, non-simple, and tight for tw but neither fvn nor pw.
        \undirectedPlanarHardness*
        Given an instance $(D,k)$ of \textup{\textsc{Directed Cut}}, we construct an instance $(\mathcal G, k')$ of \textup{\textsc{oTCC}} and \textup{\textsc{cTCC}} such that $D$ has a directed cut of size at least $k$ if and only if the undirected temporal graph $\mathcal G$ has an open (resp. closed) tcc of size at least $k'$.

        \bigparagraph{Intuition.}
        We again use two individual comparison vertices $a,b$, processing ranks in increasing order at $a$ and decreasing order at $b$. Simply forgetting the edge orientations in the directed construction creates unintended temporal paths: at $a$, a path can traverse another arc-vertex's entry edge backwards, bypassing the tail-rank comparison. At $b$, the shared subdivision vertex similarly allows a path to traverse an entry edge backwards and then use the late edge to the arc-vertex, motivating separate subdivided in- and out-connections in the undirected construction. We therefore use separate subdivided in/out connections (see \Cref{fig:UD-planar-tw-pw-fvn}) with multiple labels per edge. For each comparison vertex $c\in\{a,b\}$, the vertex $p_e^c$ belongs to the in-connection $x_e-p_e^c-c$, while $q_e^c$ belongs to the out-connection $c-q_e^c-x_e$. These separate vertices allow us to control the timing of arrivals at the portals and returns to arc-vertices independently.

        Within each comparison block, the two edges of each in-connection occur consecutively, first $x_e p_e^c$ and then $p_e^c c$. A path traversing the connection backwards from $c$ therefore misses the earlier edge into $x_e$. As in the directed construction, departures from a comparison vertex precede arrivals at equal ranks, excluding equality. The final edges $q_e^c x_e$ of the out-connections appear only after both comparison blocks, so comparison paths cannot be chained.

        For auxiliary connectivity, we add early and late occurrences of the edges between the comparison vertices and their subdivision vertices. The early occurrences, followed by the edge $ab$, allow every auxiliary vertex to reach both comparison vertices before the comparisons. The late occurrences allow arc-vertices to reach all subdivision vertices after reaching $a$ or $b$. All labels are distinct, so the construction is proper, while these repeated static edges make it non-simple. Deleting $a,b$ leaves a disjoint union of stars, preserving the small footprint parameters.

        \begin{figure}[t]
            \centering
            \includegraphics[width=0.8\textwidth]{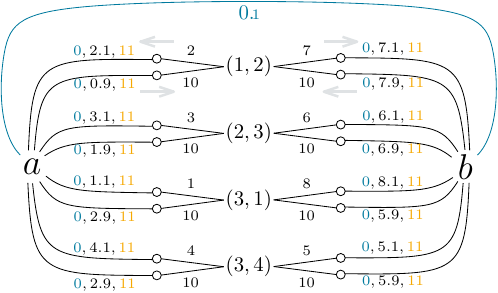}
            \caption{The undirected planar construction for arcs $(1,2)$, $(2,3)$, $(3,1)$, and $(3,4)$, with $r(i)=i$. Each central vertex represents an arc $e=(u,v)$; the hollow vertices subdivide its in- and out-connections to $a,b$. Grey arrows indicate intended temporal traversals; all edges are undirected. Blue labels $0,0.1$ belong to $T_1$, black labels below $4.5$ to $T_2$, black labels between $4.5$ and $10$ to $T_3$, labels $10$ to $T_4$, and yellow labels $11$ to $T_5$. At $a$, the in-connection has labels $r(v),r(v)+0.1$, whereas the out-connection starts at $r(u)-0.1$. At $b$, the corresponding labels are $4+\overline{r(v)},4+\overline{r(v)}+0.1$ and $4+\overline{r(u)}-0.1$, where $\overline{r(w)}=5-r(w)$. All out-connections end at time $10$, preventing comparisons from being chained. Labels are schematic; ties are broken as specified in the construction to make all times distinct.}
            \label{fig:UD-planar-tw-pw-fvn}
        \end{figure}
        \bigparagraph{Construction (see \Cref{fig:UD-planar-tw-pw-fvn}).}
        For every arc $e\in A$, introduce an \textit{arc-vertex} $x_e$ and four auxiliary vertices $p_e^a,p_e^b,q_e^a,q_e^b$. Additionally, introduce two \textit{comparison vertices} $a$ and $b$.
        Let $Z=\{a,b\}\cup\{p_e^a,p_e^b,q_e^a,q_e^b \colon e\in A\}$.
        The lifetime is partitioned into five consecutive blocks $T_1,T_2,T_3,T_4,T_5$, where every label in $T_r$ precedes every label in $T_{r+1}$.

        In $T_1$, add the edges $(a,p_e^a)$, $(a,q_e^a)$, $(b,p_e^b)$, and $(b,q_e^b)$ for every $e\in A$, in an arbitrary order. Then add the edge $(a,b)$.

        In $T_2$, we process the ranks in increasing order. For every $\ell=1,\ldots,|V|$, first place all edges $(a,q_f^a)$ for which $r(\operatorname{tail}(f))=\ell$. Then, for every $e$ with $r(\operatorname{head}(e))=\ell$, place the consecutive edges $(x_e,p_e^a)$ and $(p_e^a,a)$. The edges and consecutive pairs within each of these two families may be ordered arbitrarily.

        In $T_3$, we process the ranks in decreasing order. For every $\ell=|V|,\ldots,1$, first place all edges $(b,q_f^b)$ for which $r(\operatorname{tail}(f))=\ell$. Then, for every $e$ with $r(\operatorname{head}(e))=\ell$, place the consecutive edges $(x_e,p_e^b)$ and $(p_e^b,b)$. Again, the edges and consecutive pairs within each family may be ordered arbitrarily.

        In $T_4$, add the edges $(q_e^a,x_e)$ and $(q_e^b,x_e)$ for every $e\in A$, in an arbitrary order.

        Finally, in $T_5$, add another occurrence of each edge $(a,p_e^a)$, $(a,q_e^a)$, $(b,p_e^b)$, and $(b,q_e^b)$ for every $e\in A$, in an arbitrary order.

        We set the target size to $k'=k+|Z|=k+4|A|+2$. This completes the construction of the instance $(\mathcal G,k')$, which is clearly computable in polynomial time. All time labels are distinct, and hence $\mathcal G$ is proper. Since the edges between the comparison vertices and the auxiliary vertices receive several time labels, $\mathcal G$ is not simple.
        For the application of \Cref{lem:directed-cut-encoding}, we use the set $Z$ defined above.
        We first prove the structural properties of the footprint.

        \begin{lemma}\label{lem:undirected-planar-footprint}
            The footprint $G$ is planar and satisfies $\tw(G)=2$, $\pw(G)=3$, and $\fvn(G)=2$.
        \end{lemma}
        \begin{proof}
            We construct a tree decomposition of $G$. For every $e\in A$, let $B_e=\{a,b,x_e\}$ and let $B_e^{p,a}=\{a,x_e,p_e^a\}$, $B_e^{q,a}=\{a,x_e,q_e^a\}$, $B_e^{p,b}=\{b,x_e,p_e^b\}$, and $B_e^{q,b}=\{b,x_e,q_e^b\}$.
            Arrange the bags $B_e$ in an arbitrary order along a path and attach the four bags $B_e^{p,a},B_e^{q,a},B_e^{p,b},B_e^{q,b}$ as leaves to the corresponding bag $B_e$. Every edge of $G$ is contained in a bag, and the bags containing any fixed vertex form a connected subtree. This is a tree decomposition of width two, and hence $\tw(G)\leq2$. Since $G$ contains a cycle, $\tw(G)\geq2$, and consequently $\tw(G)=2$. Since graphs with treewidth 2 are planar, so is our footprint $G$.

            Deleting $a$ and $b$ leaves, for every $e\in A$, a star with centre $x_e$ and leaves $p_e^a,p_e^b,q_e^a,q_e^b$. Thus, $\fvn(G)\leq2$. Conversely, deleting any single vertex leaves a cycle in the footprint: Deleting $a$ leaves every cycle through $b$, deleting $b$ leaves every cycle through $a$, and, if a vertex of an arc gadget is deleted, a cycle belonging to another arc gadget remains. Hence, $\fvn(G)\geq2$, and therefore $\fvn(G)=2$.

            A path decomposition of width three is obtained by taking, for every $e\in A$, the four consecutive bags $\{a,b,x_e,p_e^a\}$, $\{a,b,x_e,q_e^a\}$, $\{a,b,x_e,p_e^b\}$, and $\{a,b,x_e,q_e^b\}$, and concatenating these groups in an arbitrary order. Hence, $\pw(G)\leq3$. 
            For the matching lower bound, choose three distinct arcs $e_1,e_2,e_3\in A$, which exist since $|A|\geq3$. Retain only the edges of the three paths $a-p_{e_i}^a-x_{e_i}-p_{e_i}^b-b$ for $i\in\{1,2,3\}$, and contract each edge $x_{e_i}p_{e_i}^b$. The resulting minor is isomorphic to $E_2$ in \cite[Theorem~3.1 and Figure~3]{barat_structuregraphs_2009}, an excluded minor for pathwidth at most two. Thus, $\pw(G)\geq3$, and hence $\pw(G)=3$.
        \end{proof}

        Next, we establish the reachability property.
        \begin{lemma}\label{lem:undirected-planar-reachability}
            For every two distinct arcs $e,f\in A$, the arc-vertex $x_e$ reaches the arc-vertex $x_f$ in $\mathcal G$ if and only if $\operatorname{head}(e)\neq\operatorname{tail}(f)$.
        \end{lemma}
        \begin{proof}
            Consider first the comparison at $a$. A temporal path starting at $x_e$ can use the consecutive edges $(x_e,p_e^a)$ and $(p_e^a,a)$ when rank $r(\operatorname{head}(e))$ is processed. It can then leave $a$ towards $q_f^a$ when rank $r(\operatorname{tail}(f))$ is processed and reaches $x_f$ using $(q_f^a,x_f)$ in $T_4$. Since, at equal ranks, the edges leaving $a$ occur before the consecutive paths entering $a$, the temporal path $x_e-p_e^a-a-q_f^a-x_f$ exists exactly when $r(\operatorname{head}(e))<r(\operatorname{tail}(f))$.
            At $b$, the ranks are processed in decreasing order, and the edges leaving $b$ again occur before the consecutive paths entering $b$ at equal ranks. Hence, the temporal path $x_e-p_e^b-b-q_f^b-x_f$ exists exactly when $r(\operatorname{head}(e))>r(\operatorname{tail}(f))$.

            It remains to check that the undirected edges do not create additional reachability. The edges in $T_1$ occur before an arc-vertex can reach any of its auxiliary vertices. Before $T_4$, the only edges incident with $x_e$ are $(x_e,p_e^a)$ in $T_2$ and $(x_e,p_e^b)$ in $T_3$.
            For $c\in\{a,b\}$, the edge $(x_f,p_f^c)$ appears immediately before the comparison occurrence of $(p_f^c,c)$. A path reaching $p_f^c$ from $c$ during the comparison block therefore misses the edge $(p_f^c,x_f)$. Since $p_f^c$ has no other neighbors, a path starting at an arc-vertex cannot reach another arc-vertex by traversing an in-connection backwards.
            The edges $(q_f^a,x_f)$ and $(q_f^b,x_f)$ occur only in $T_4$, after both comparison blocks. Consequently, once a temporal path reaches $x_f$, every comparison edge leaving $x_f$ has already passed. Finally, the edges in $T_5$ occur after all edges entering an arc-vertex. Thus, every temporal path between two distinct edge-vertices is provided by one of the two comparisons described above.
        \end{proof}
        By applying this lemma to the ordered pairs $(e,f)$ and $(f,e)$, we obtain that for every two distinct arcs $e,f\in A$, the edge-vertices $x_e$ and $x_f$ are {mutually reachable} if and only if $\operatorname{head}(e)\neq\operatorname{tail}(f)$ and $\operatorname{head}(f)\neq\operatorname{tail}(e)$.
        {Next, we show that every auxiliary vertex is mutually reachable with every vertex of $\mathcal G$, with witnessing paths contained in the sets specified below.}
        \begin{lemma}\label{lem:undirected-planar-connectivity}
            For every arc-vertex $x_e$ and every auxiliary vertex $z\in Z$, the vertices $x_e$ and $z$ are {mutually reachable} via temporal paths whose vertices belong to $Z\cup\{x_e\}$. Moreover, {$\mathcal G[Z]$ is temporally connected}.
        \end{lemma}
        \begin{proof}
            First consider a vertex $z\in Z=\{a,b\}\cup\{p_e^a,p_e^b,q_e^a,q_e^b \colon e\in A\}$. Every vertex $p_f^a$ and $q_f^a$ reaches $a$ in $T_1$, and every vertex $p_f^b$ and $q_f^b$ reaches $b$ in $T_1$. Since all these edges precede $(a,b)$, every vertex of $Z$ reaches both $a$ and $b$ in $T_1$. The path can then use $(a,q_e^a)$ in $T_2$, followed by $(q_e^a,x_e)$ in $T_4$ to reach $x_e$.
            Conversely, $x_e$ reaches $a$ through the path $x_e-p_e^a-a$ in $T_2$ and reaches $b$ through the path $x_e-p_e^b-b$ in $T_3$. After reaching the appropriate comparison vertex, it reaches every vertex $p_f^a,q_f^a,p_f^b,q_f^b$ directly in $T_5$. Thus, $x_e$ reaches every vertex of $Z$. All vertices on these paths belong to $Z\cup\{x_e\}$.
            Finally, let $z,z'\in Z$. In $T_1$, the vertex $z$ reaches its incident comparison vertex and, using $(a,b)$ if necessary, reaches either $a$ or $b$. If $z'$ is one of the comparison vertices, it has already been reached; otherwise, the path reaches $z'$ using the appropriate edge in $T_5$. Hence, {$\mathcal G[Z]$ is temporally connected}.
        \end{proof}
        \begin{proof}[Proof of \Cref{thm:undirected-planar}]
            By \Cref{lem:undirected-planar-reachability}, two edge-vertices $x_e$ and $x_f$ are {mutually reachable} only if the corresponding arcs satisfy the {directed-cut condition}. \Cref{lem:undirected-planar-connectivity} verifies the second property of \Cref{lem:directed-cut-encoding}. Applying that lemma proves the correctness of the reduction for both \textup{\textsc{oTCC}} and \textup{\textsc{cTCC}}.
            Additionally, the construction is polynomial and proper, and \Cref{lem:undirected-planar-footprint} gives all claimed structural bounds.
            \Cref{thm:undirected-planar} therefore follows.
        \end{proof}

        Consequently, \openTCC and \closedTCC remain \paraNP-hard for these footprint parameters in undirected temporal graphs. For undirected vertex cover, however, we obtain \XP algorithms in the next section.

\section[Vertex Cover and Temporal Path/Walk Number]{Vertex Cover, Temporal Path Number, and Temporal Walk Number}\label{sec:vc-tpn}
    In this section we resolve the parameterized complexity with respect to the temporal path number, the temporal walk number, and the vertex cover number in undirected graphs. We prove the \XP algorithm for \twn and the hardness results for \tpn. Since $\twn(\gcal)\leq\tpn(\gcal)$, these yield the corresponding algorithm for \tpn and hardness results for \twn.
    The results are organized in four subsections, based on the technique used. See \Cref{fig:vc-tpn-navigation} for a navigation of the resulting parameter bounds to the sub- or subsubsections containing the corresponding result.

    \begin{center}
        \begin{minipage}{\linewidth}
            \centering
            \begin{minipage}{0.85\linewidth}
\begingroup
\definecolor{resultDiagramHard}{rgb}{0.753,0,0.235}
\definecolor{resultDiagramEasy}{rgb}{0.251,0.427,0.212}
\hypersetup{hidelinks}
\newcommand{\diagramSec}[1]{\hyperref[#1]{Sec.~\ref*{#1}}}
\pgfmathsetlengthmacro{\diagramUnit}{\linewidth/326}
\pgfmathsetlengthmacro{\diagramFontSize}{7*\diagramUnit}
\pgfmathsetlengthmacro{\diagramLineSkip}{8.4*\diagramUnit}
\pgfmathsetlengthmacro{\diagramSmallFontSize}{6*\diagramUnit}
\pgfmathsetlengthmacro{\diagramSmallLineSkip}{7.2*\diagramUnit}
\begin{tikzpicture}[
    inner sep=0pt, outer sep=0pt,
    every node/.style={align=center,
        font=\fontsize{\diagramFontSize}{\diagramLineSkip}\selectfont\bfseries},
    hardness/.style={text=resultDiagramHard},
    algorithm/.style={text=resultDiagramEasy}
]
    \node[anchor=north west] (background) at (0,0)
        {\includegraphics[width=\linewidth]{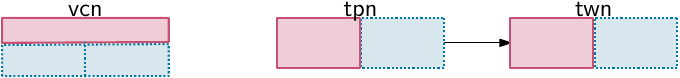}};
    \begin{scope}[x=\diagramUnit,y=-\diagramUnit,overlay]
        \node[hardness] at (153,20.57)
            {\diagramSec{sec:closed-tpn-pNPh}};
        \node[hardness,
            font=\fontsize{\diagramSmallFontSize}{\diagramSmallLineSkip}\selectfont\bfseries]
            at (193,14.57)
            {\diagramSec{sec:open-tpn-Wh}};
        \node[algorithm,
            font=\fontsize{\diagramSmallFontSize}{\diagramSmallLineSkip}\selectfont\bfseries]
            at (193,26.57)
            {\diagramSec{subsubsec:temporal-walk-path-xp}};

        \node[hardness] at (265,20.57)
            {\diagramSec{sec:closed-tpn-pNPh}};
        \node[hardness,
            font=\fontsize{\diagramSmallFontSize}{\diagramSmallLineSkip}\selectfont\bfseries]
            at (305,14.57)
            {\diagramSec{sec:open-tpn-Wh}};
        \node[algorithm,
            font=\fontsize{\diagramSmallFontSize}{\diagramSmallLineSkip}\selectfont\bfseries]
            at (305,26.57)
            {\diagramSec{subsubsec:temporal-walk-path-xp}};

        \node[algorithm,
            font=\fontsize{\diagramSmallFontSize}{\diagramSmallLineSkip}\selectfont\bfseries]
            at (21,25.00)
            {\diagramSec{subsec:vc-closed-xp}};
        \node[hardness,
            font=\fontsize{\diagramSmallFontSize}{\diagramSmallLineSkip}\selectfont\bfseries]
            at (21,32.70)
            {\diagramSec{subsec:undirected-vc-hardness}};
        \node[algorithm,
            font=\fontsize{\diagramSmallFontSize}{\diagramSmallLineSkip}\selectfont\bfseries]
            at (61,25.00)
            {\diagramSec{subsubsec:undirected-vc-open-xp}};
        \node[hardness,
            font=\fontsize{\diagramSmallFontSize}{\diagramSmallLineSkip}\selectfont\bfseries]
            at (61,32.70)
            {\diagramSec{subsec:undirected-vc-hardness}};
    \end{scope}
\end{tikzpicture}%
\endgroup

            \end{minipage}
            \captionof{figure}{Navigation map for the results for vertex cover in undirected graphs, temporal path number, and temporal walk number in \Cref{sec:vc-tpn}. Each entry links to the sub- or subsubsection containing the corresponding result. 
            Red entries denote \Wone hardness results and green entries denote \XP algorithms.}
            \label{fig:vc-tpn-navigation}
        \end{minipage}
    \end{center}

    \subsection{\XP Algorithms for \openTCC via VC-Dimension}\label{sec:open}
        Both algorithms bound the VC-dimension of the family of maximal open temporal connected components. We introduce the common enumeration framework before proving the two structural bounds.
    \begin{definition}[\vcdimension~\cite{vapnik_UniformConvergence_2015}]
        Let $\mathcal{F}$ be a family of subsets of a universe $U$.
        A set $A\subseteq U$ is \emph{shattered} by $\mathcal{F}$ if, for every subset $S\subseteq A$, there exists a set $F\in\mathcal{F}$ such that $F\cap A=S$.
        The \emph{\vcdimension} of $\mathcal{F}$, denoted by $\vcdim(\mathcal{F})$, is the maximum size of a set shattered by~$\mathcal{F}$.
    \end{definition}

    \begin{lemma}[Sauer--Shelah--Perles Lemma~\cite{sauer_DensityFamilies_1972,shelah_CombinatorialProblem_1972}]
        \label{sauer shelah perles lemma}
        Let $\mathcal{F}$ be a set system over a universe of size $n$ with \vcdimension at most $d$.
        Then $|\mathcal{F}|\leq\sum_{i=0}^{d}\binom{n}{i}=\mathcal{O}(n^d)$.
    \end{lemma}
    We denote by $\Copen{\gcal}$ and $\Cclosed{\gcal}$ the families of maximal open and closed temporally connected sets, respectively, interpreted as set families. That is, $\Copen{\gcal}=\{C\subseteq V \colon C\text{ is a maximal open tcc}\}$
    and $\Cclosed{\gcal} =\{C\subseteq V \colon C\text{ is a maximal closed tcc}\}$.

    The following observation provides the algorithmic framework for our application of VC-dimension to open temporal connected components.
    \begin{proposition}\label{obs:vcdim-enumeration}
        Let $\gcal$ be a temporal graph on $n$ vertices. If $\vcdim\bigl(\Copen{\gcal}\bigr)\leq d$, then \openTCC can be solved in time $n^{d+\mathcal{O}(1)}$.
    \end{proposition}
    \begin{proof}
        {Define the \emph{mutual-reachability graph} $H_{\gcal}$ as the undirected graph on vertex set $V$ with an edge $\{u,v\}$ for distinct vertices $u,v$ if and only if $u\rightreach v$ and $v\rightreach u$ in $\gcal$.} This graph $H_{\gcal}$ can be constructed in polynomial time.
        A set $C\subseteq V$ is an open temporally connected set if and only if $C$ is a clique in $H_{\gcal}$. Consequently, $\Copen{\gcal}$ is precisely the family of maximal cliques of $H_{\gcal}$.

        By the Sauer--Shelah--Perles lemma, the assumption $\vcdim(\Copen{\gcal})\leq d$ implies
            $\lvert\Copen{\gcal}\rvert
            \leq \sum_{i=0}^{d}\binom{n}{i}
            \leq n^{d+\mathcal{O}(1)}$.
        Since all maximal cliques of a graph can be enumerated with polynomial delay, we can enumerate all members of $\Copen{\gcal}$ in total time $n^{d+\mathcal{O}(1)}$. Returning a largest enumerated set solves \openTCC.
    \end{proof}
    {This argument does not extend to closed tccs: whether two vertices are mutually reachable in $\gcal[C]$ depends on the chosen set $C$, since the witnessing temporal paths must stay within $C$.} A bound on the \vcdimension $\vcdim\bigl(\Cclosed{\gcal}\bigr)\leq d$ still implies, by the Sauer--Shelah lemma, that $\lvert\Cclosed{\gcal}\rvert\leq n^{d+\mathcal{O}(1)}$, but this only gives us a combinatorial bound on the number of maximal closed tccs. It does not (by itself) give us an algorithm for enumerating them/computing the max. Thus no XP algorithm just from this.

    The following lemma identifies a reachability configuration forced by large VC-dimension of open tccs. Both the vertex cover and temporal walk number arguments will show that such a configuration is impossible under constant bounds. See \Cref{fig:vcdim-structure} for an illustration of the configuration.
    \begin{lemma}[Blocking sets] \label{lem:vcdim-classification-open}
        If $\vcdim\left(\Copen{\gcal}\right)\geq 2d+2$ then there exist a \emph{shattered set} $A := \{a_1, \dots, a_{d+1}\} \subseteq V$ and a \emph{blocking set} $B := \{b_1, \dots, b_{d+1}\} \subseteq V \setminus A$ such that for every $i,j \in [d+1], i\neq j$, holds
            \begin{enumerate}
                \item $a_i \rightreach a_j$ and $a_i \leftreach a_j$,
                \item $b_i \rightreach a_j$ and $b_i \leftreach a_j$, and
                \item $b_i \not\rightreach a_i$.
            \end{enumerate}
    \end{lemma}
    \begin{proof}
        Let $\mathcal{G}=(V,E,\lambda)$ be a temporal graph with $\vcdim(\Copen{\gcal})\geq 2d+2$.
        Then there must exist a set $A \subseteq V$ of size $2d+2$ that is \emph{shattered} by $\Copen{\gcal}$. That is, for every subset $S \in 2^A$ there exists an open tcc $C(S) \in \Copen{\gcal}$ such that $C(S) \cap A = S$.
        Let $A = \{a_1, \dots, a_{2d+2}\}$. Since $A$ is shattered by $\Copen{\gcal}$, there exists $C(A)\in\Copen{\gcal}$ with $A\subseteq C(A)$. Consequently, $A$ is temporally connected and for every $i,j \in [2d+2], i\neq j$, holds
        \begin{itemize}
            \item $a_i \rightreach a_j$ and $a_i \leftreach a_j$.
        \end{itemize}
        Again by definition of a shattered set, for each $a_i \in A$, the set $A \setminus \{a_i\}=:A_{-i}$ must be contained in some $C(A_{-i}) \in \Copen{\gcal}$. Since $a_i\notin C(A_{-i})$ but {$a_i$ is mutually reachable with every vertex of $A_{-i}$}, there must exist a \emph{blocking} vertex $b_i\in C(A_{-i})\setminus A$ which is not {mutually reachable} with $a_i$, \ie
        \begin{itemize}
            \item $b_i \not\rightreach a_i$ or $b_i \not\leftreach a_i$, and
            \item $b_i \rightreach a_j$ and $b_i \leftreach a_j$ for all $i \neq j$.
        \end{itemize}
        It follows that $b_i\notin A$, since {any two vertices of $A$ are mutually reachable}. Consequently, the blocking vertices are pairwise distinct: $b_i\neq b_j$ for $i\neq j$ since $b_i\rightreach a_j$ and $b_j\not\rightreach a_j$.

        Overall, this implies the existence of a set $B = \{b_1, \dots, b_{2d+2}\} \subseteq V \setminus A$ of blocking vertices such that each $b_i$ is {not mutually reachable} with $a_i$ but {mutually reachable} with every $a_j$ for $i \neq j$. {Since, for each $i$, at least one of $b_i\not\rightreach a_i$ and $a_i\not\rightreach b_i$ holds} and $\lvert A\rvert=2d+2$, the pigeonhole principle implies a subset $A' \subseteq A$ of size at least $d+1$ for which {the missing reachabilities have the same direction}. Without loss of generality, we can assume it is $b_i\not\curvearrowleft a_i$. This proves the claim.
        \begin{figure}[t]
            \centering
            \includegraphics[width=0.5\linewidth]{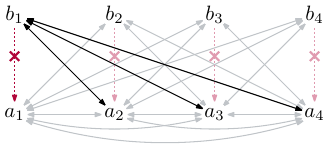}
            \caption{Sketch of a shattered set of size $4$. {Any two vertices of $\{a_1,\dots,a_4\}$ are mutually reachable. Each $b_i$ is mutually reachable with every $a_j$ for $j\neq i$, but not with $a_i$.} For visibility, the relations of $b_1$ are highlighted.}
            \label{fig:vcdim-structure}
        \end{figure}
    \end{proof}

        \subsubsection{Undirected vertex cover}\label{subsubsec:undirected-vc-open-xp}
            For the vertex-cover argument, we classify vertices of the independent set by the relative order of incident edge times. The same abstract types will also be used for the closed-component algorithm in \Cref{subsec:vc-closed-xp}.
                    \begin{definition}[weak-order type]\label{def:weak-order-type}
            Let $Z$ be a finite set. A \emph{weak-order type over $Z$} is a pair $(Y,\preceq)$, where $Y\subseteq Z$ and $\preceq$ is a total preorder on $Y$.
        \end{definition}
        \begin{remark}\label{obs:number-weak-order-types}
            A set $Z$ of size $k$ admits at most $(k+1)^k=2^{\mathcal O(k\log k)}$ weak-order types.
        \end{remark}
        \begin{proof}
            Encode a weak order on a subset of $Z$ by mapping every absent element to $0$ and every present element to the rank of its equivalence class in $[k]$. This injectively encodes every type by a map from $Z$ to $\{0,\ldots,k\}$ and gives the claimed upper bound.
        \end{proof}

                    We consider the families of maximal \otcc{}s as set systems over the vertex set, and prove that in undirected temporal graphs the vertex cover number of the footprint bounds the \vcdimension of these families.
        The Sauer--Shelah--Perles Lemma and \Cref{obs:vcdim-enumeration} then yield an \XP algorithm.
        \openVCXP*
        To prove the theorem, we use \Cref{lem:vcdim-classification-open} to bound the VC-dimension of \Copen{\gcal} by a function of $\vc(\gcal)$ whenever $\gcal$ is undirected.
        For that we associate with each independent vertex a weak-order type induced by the latest times at which it can be entered from the vertex cover.
        \begin{definition}[latest-edge type]\label{def:latest-edge-type}
            Let $Z$ be a vertex cover of the footprint and let $I:=V\setminus Z$. For $x\in I$ and $z\in N_G(x)\cap Z$, define $\ell^+(x,z):=\max\{t:(xz,t)\in\mathcal E\}$.
            The \emph{latest-edge type} of $x$ with respect to $Z$ is the weak-order type $\tau_Z^+(x):=(N_G(x)\cap Z,\preceq_x^+)$, where $z\preceq_x^+z'$ if and only if $\ell^+(x,z)\leq\ell^+(x,z')$.
        \end{definition}
        Thus, the latest-edge type records the cover neighbors of $x$ and their relative order by latest available edge time, but not the actual times.
        \begin{lemma}\label{lem:undirected-vc-vcdim}
            Let \gcal be an undirected temporal graph with vertex cover number $\vc$.
            Then $\vcdim\big( \Copen{\gcal} \big) \leq 2\big( 2^{\mathcal O(\vc\log\vc)}+\vc\big)+1$.
        \end{lemma}
        \begin{proof}
            Let $\gcal=(V,E,\lambda)$ be an undirected temporal graph with vertex cover number $\vc(\gcal)=k$. Let $Z$ denote the set of cover vertices and let $I:=V\setminus Z$.
            Towards contradiction assume that $\vcdim\big( \Copen{\gcal} \big) \geq 2\big( 2^{\mathcal O(k\log k)}+k \big) +2$.
            By \Cref{lem:vcdim-classification-open} with $d=2^{\mathcal O(k\log k)}+k$, there exists a shattered set $A:=\{a_1,\dots,a_{2^{\mathcal O(k\log k)}+k+1}\}\subseteq V$ and a blocking set $B:=\{b_1,\dots,b_{2^{\mathcal O(k\log k)}+k+1}\}\subseteq V\setminus A$ such that for every $i,j \in [2^{\mathcal O(k\log k)}+k+1], i\neq j$, holds
            \begin{enumerate}
                \item $a_i \rightreach a_j$ and $a_i \leftreach a_j$,
                \item $b_i \rightreach a_j$ and $b_i \leftreach a_j$, and
                \item $b_i \not\rightreach a_i$.
            \end{enumerate}
            At most $k$ vertices of $A\cup B$ can be in $Z$, which leaves $2^{\mathcal O(k\log k)}+1$ for each set. After renaming $A:=\{a_1,\dots,a_{2^{\mathcal O(k\log k)}+1}\}$ and $B:=\{b_1,\dots,b_{2^{\mathcal O(k\log k)}+1}\}$, the same three properties 1.-3. must still hold.
            We show that this is not realizable with only $k$ cover vertices.

            Since $A\cup B\subseteq I$ form an independent set, any temporal path between them must visit at least one cover vertex from $Z$. By \Cref{def:latest-edge-type,obs:number-weak-order-types}, the vertices of $I$ have at most $2^{\mathcal O(k\log k)}$ latest-edge types with respect to $Z$.

            Since $A$ contains $2^{\mathcal O(k\log k)}+1$ elements, the pigeonhole principle yields two distinct elements of $A$ having the same latest-edge type, which we denote by $a_1$ and $a_2$.
            Let $Z':=N_G(a_1)\cap Z=N_G(a_2)\cap Z$ denote their common set of cover neighbors. Since the weak orders in their common latest-edge type coincide, we denote this order by $\preceq$.

            \begin{figure}[t]
                \centering
                \includegraphics[width=0.9\textwidth]{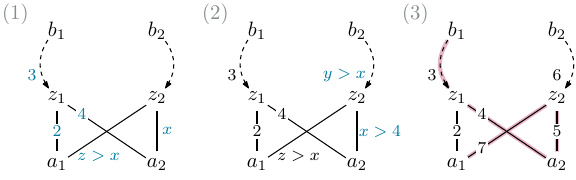}
                \caption{Illustration of the argument of \Cref{lem:undirected-vc-vcdim} in three steps. (1) The vertex $b_1$ reaches $a_2$ through some cover vertex $z_1$ but does not reach $a_1$. (2) After ordering $z_1$ and the analogous cover vertex $z_2$ for $b_2$, their label relations imply $\ell^+(a_2,z_1)<\ell^+(a_2,z_2)<\ell^+(a_1,z_2)$. (3) Consequently, $b_1$ reaches $a_1$ through $z_1,a_2,z_2$, a contradiction.}
                \label{fig:vcXP-vcdim}
            \end{figure}
            Refer to \Cref{fig:vcXP-vcdim} for an illustration of the following argument.
            Since $b_1\curvearrowright a_2$ and $b_2\curvearrowright a_1$,
                there exist $z_1,z_2\in Z'$ such that there are temporal paths $(b_1,\dots,z_1,a_2)$ and $(b_2,\dots,z_2,a_1)$.
            On the other hand, $b_1\not\curvearrowright a_1$ and $b_2\not\curvearrowright a_2$, so it must hold that
            \[\ell^+(a_1,z_1)<\ell^+(a_2,z_1)\quad\text{and}\quad\ell^+(a_2,z_2)<\ell^+(a_1,z_2).\]
            Without loss of generality, assume that $z_1\preceq z_2$ in the common weak order. The two vertices cannot be tied in the order, as otherwise $\ell^+(a_i,z_1)=\ell^+(a_i,z_2)$ for both $i\in[2]$, and hence
            \[
                \ell^+(a_1,z_1)<\ell^+(a_2,z_1)=\ell^+(a_2,z_2)
                <\ell^+(a_1,z_2)=\ell^+(a_1,z_1),
            \]
            which is a contradiction. Consequently, $z_1\prec z_2$, and the common weak order gives
            \[
                \ell^+(a_2,z_1)<\ell^+(a_2,z_2)<\ell^+(a_1,z_2).
            \]
            But now $(b_1,\dots,z_1,a_2,z_2,a_1)$ is a temporal walk from $b_1$ to $a_1$, which contradicts $b_1\not\curvearrowright a_1$.
            As a result, we obtain a contradiction to the defining properties of $A$ and $B$ for both strict and non-strict reachability. Hence, no shattered set of the assumed size exists, and therefore $\vcdim\bigl(\Copen{\gcal}\bigr)\leq 2\bigl(2^{\mathcal O(k\log k)}+k\bigr)+1$.
        \end{proof}
        With this we can conclude the theorem.
        \begin{proof}[Proof of \Cref{thm:undirected-vc-xp-open}]
            By \Cref{lem:undirected-vc-vcdim}, the VC-dimension of $\Copen{\gcal}$ is bounded by a function of $\vc(\gcal)$. By the Sauer--Shelah--Perles Lemma, this implies that the number of maximal open tccs is polynomially bounded in $n$ with an exponent depending on $\vc(\gcal)$.
            Finally, by \Cref{obs:vcdim-enumeration}, we can enumerate all maximal open tccs in time $n^{\mathcal O(\vcdim(\Copen{\gcal}))} = n^{\mathcal O(2^{\mathcal O(k\log k)}+k)}$, which is polynomial in $n$ for fixed $\vc(\gcal)$.
        \end{proof}

        \subsubsection{Temporal walk and path number}\label{subsubsec:temporal-walk-path-xp}
                    We now apply the same VC-dimension framework to temporal walk number, using the common blocking-set lemma (\Cref{lem:vcdim-classification-open}).
        \openWalkXP*
        \begin{lemma}
        \label{lem:otcc bounded VCdimension}
            The maximal \otcc{}s in an (un)directed, (non-)strict temporal graph $\gcal$ with temporal walk number \twn form a set system over a universe of size $\lvert V\rvert=n$ with \vcdimension at most $2\twn+1$, and consequently at most $2\tpn+1$.
        \end{lemma}
        \begin{proof}
            Let $\mathcal{G}=\bigcup_{i\in[k]}W_i$ be a temporal graph with $\twn(\mathcal{G})=k$.
            Towards contradiction assume that $\vcdim\bigl(\Copen{\mathcal{G}}\bigr)\geq 2k+2$.
            By \Cref{lem:vcdim-classification-open} with $d=k$, there exist a shattered set $A=\{a_1,\dots,a_{k+1}\}$ and a blocking set $B=\{b_1,\dots,b_{k+1}\}$ such that for all distinct $i,j\in[k+1]$,
            \begin{enumerate}
                \item $a_i\rightreach a_j$ and $a_i\leftreach a_j$,
                \item $b_i\rightreach a_j$ and $b_i\leftreach a_j$, and
                \item $b_i\not\rightreach a_i$. \label{enum: three}
            \end{enumerate}
            We show that this configuration cannot be realized by $k$ temporal base-walks.
            
            \begin{figure}[ht]
                \centering
                \includegraphics[width=0.4\linewidth]{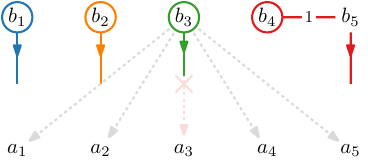}
                \caption{Sketch of a temporal graph \gcal with $\twn=4$, five shattered vertices $a_1,\dots,a_5$, five blocking vertices $b_1,\dots,b_5$, and four base-walks $W_1,\dots,W_4$ whose first blocking vertices are $b_1,\dots,b_4$, respectively. The construction demands that $b_3$ reaches $a_1,a_2,a_4,a_5$ but does not reach $a_3$.}
                \label{fig:XP-first_block}
            \end{figure}
            For each base-walk that contains a vertex of $B$, consider the first occurrence of a blocking vertex on that walk and call it the walk's \emph{first blocking vertex}.
            Since there are $k$ base-walks but $k+1$ blocking vertices, there exists a blocking vertex $b_i$ that is not the first blocking vertex on any base-walk.
            Consequently, on every base-walk containing $b_i$, some other blocking vertex occurs before the first occurrence of~$b_i$.
            See for example \Cref{fig:XP-first_block}, vertex $b_5$.

            Let $W_j$ be the base-path with the earliest incoming edge at $b_i$, and let $b_j, i\neq j$ denote the first blocking vertex on $W_j$. 
            Now recall that $b_i$ has to reach $a_j$ via some temporal path. If this $b_i$-$a_j$-path were to use $W_j$ then $b_j$ could also reach $a_j$, contradicting $b_j \not\rightreach a_j$ (\Cref{enum: three}).
            If the $b_i$-$a_j$-path does not use $W_j$, then it takes some other base-walk $W$ which arrives at $b_i$ after $W_j$ (because $W_j$ was chosen to be the earliest base-walk arriving at $b_i$), and consequently leaves $b_i$ after $W_j$ arrived. Thus, every vertex on $W_j$ before $b_i$ can reach the vertices on $W$ after $b_i$.
            As a result, $b_j$ can reach $a_j$ using first $W_j$ and then the $b_i$-$a_j$-path using $W$, which contradicts $b_j\not\rightreach a_j$. See \Cref{fig:vcdim} for an illustration of the argument.
            \begin{figure}[ht]
                \centering
                \includegraphics[width=0.35\linewidth]{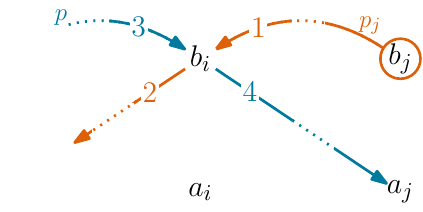}
                \caption{Illustration of the key argument for bounded \vcdimension.
                The orange base-path $P_j$ is the earliest to reach $b_i$ (here at time 1); its first blocking vertex $b_j$ is circled in orange.
                By assumption, $b_i$ is not the first blocking vertex on any base-path. The blue base-path $p$ indicates the path by which $b_i$ eventually reaches $a_j$; note that $a_j$ does not need to lie directly on $p$, it suffices that $p$ starts the path.
                }
                \label{fig:vcdim}
            \end{figure}
            
            
            We conclude that a union of $\twn=k$ temporal walks cannot realize such a reachability configuration on $k+1$ vertices, and the \vcdimension of $\Copen{\gcal}$ must therefore be less than $2k+2$.
            Note that these arguments are independent of edge directionality and of the (non-)strictness of temporal reachability. Since $\twn(\gcal)\leq\tpn(\gcal)$, the temporal path-number bound follows.
        \end{proof}
        \begin{corollary}\label{cor:full-otcc-tpn}
            \openTCC on (un)directed, (non-)strict temporal graphs can be solved in time $n^{2\twn+\mathcal O(1)}$, and consequently in time $n^{2\tpn+\mathcal O(1)}$.
        \end{corollary}
        \begin{proof}
            Apply the common enumeration framework to the VC-dimension bound in \Cref{lem:otcc bounded VCdimension}. This gives running time $n^{2\twn+\mathcal O(1)}$. Since $\twn(\gcal)\leq\tpn(\gcal)$, the temporal path-number running time follows.
        \end{proof}

        The temporal path number bounds both the temporal and the static degree of a temporal graph.
        So it is possible that the above algorithm only relies on a bounded degree and does not actually require bounded tpn.
        However, even when both the static and temporal total degree are bounded by~$6$, a temporal graph may admit exponentially many maximal \otcc{}s.
        The temporal graph depicted in \Cref{fig:exponential_otcc-bounded_deg} is a witness to this. 
        \begin{figure}[ht]
            \centering
            \includegraphics[width=0.65\linewidth]{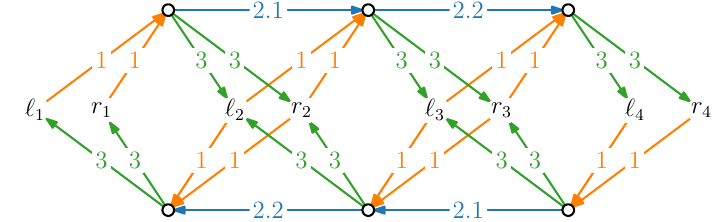}
            \caption{A temporal graph with maximum static and temporal total degree~$6$ and exponentially many maximal \otcc{}s.
            Each maximal \otcc{} contains exactly one vertex from every pair $(\ell_i,r_i)$, yielding $2^{\Theta(n)}$ distinct maximal \otcc{}s.
            The graph is directed; an undirected version can be obtained by subdividing each directed temporal edge $(u,v,t)$ into $(u,uv,t-\varepsilon)$ and $(uv,v,t+\varepsilon)$ without increasing the degree.}
            \label{fig:exponential_otcc-bounded_deg}
        \end{figure}

    \subsection{\XP Algorithms for \closedTCC Parameterized by Undirected Vertex Cover}\label{subsec:vc-closed-xp}
        Since closed temporal connected components are not hereditary, the enumeration framework that we used for open components does not apply to closed components. We instead provide an \XP algorithm that uses a combinatorial argument.
                \newcommand{\earlyE}[1]{\ensuremath{e^-(#1)}\xspace}
        \newcommand{\earlyT}[1]{\ensuremath{t^-(#1)}\xspace}
        \newcommand{\earlyV}[1]{\ensuremath{z^-(#1)}\xspace}
        \newcommand{\outtree}[1]{\ensuremath{T_{#1}^{\mathrm{out}}}\xspace}
        \newcommand{\prunedouttree}[1]{\ensuremath{\overline{T}_{#1}^{\mathrm{out}}}\xspace}
        \newcommand{\arr}[2]{\ensuremath{\operatorname{arr}_{#1}(#2)}}
        
        \closedVCXP*

        The two bounds in \Cref{thm:undirected-vc-xp-closed} for non-strict and strict temporal graphs are proven separately in \Cref{thm:undirected-vc-xp-closed-nonstrict,thm:undirected-vc-xp-closed-strict}, respectively.
        While the strict-case algorithm is more general and also applies to non-strict graphs, the non-strict-case algorithm is  simpler, has a better running time and provides useful intuition for the more involved strict case. 
        Both algorithms build on the out-tree-guessing approach used by Casteigts, Döring, Molter, and Zehavi for an \XP algorithm for \textup{\textsc{Minimum Temporal Spanner}} on happy\footnote{A temporal graph is \textit{happy} if it is both simple and proper. Recall that \textit{simple} means every static edge is assigned exactly one time label, and \textit{proper} means any two incident edges have distinct time labels.} temporal graphs~\cite{casteigts2026minimumtemporalspannershappy}. A \emph{temporal spanner} is a temporally connected subgraph spanning the whole graph, and a \emph{minimum temporal spanner} is one with a minimum number of temporal edges.
        We adapt their representative-tree argument to certificates for closed temporal connected vertex sets and extend it from happy graphs to arbitrary undirected temporal graphs.
 
        The common idea is to describe a closed temporally connected set by a parameter-bound collection of temporal spanning trees.
        The argument has four steps:
        \begin{enumerate}
            \item \textit{Witnessing complexity.} For every vertex of a closed temporally connected set, temporal reachability to all other vertices can be witnessed by a temporal spanning tree rooted at that vertex. Even further, the union of these trees is a minimum temporal spanner of the closed temporally connected set.
            \item \textit{Bounding the size of trees.} After deleting all independent-set leaves, each such  temporal spanning tree contains at most twice as many vertices as the vertex cover. We call such a tree \textit{pruned}.
            \item \textit{Bounding the number of trees.} Initially there is one out-tree for every vertex of the closed temporally connected set. However, most independent vertices can \textit{delegate} their out-tree to a bounded collection of representatives. Informally such a vertex keeps only a few incident edges through which it enters the out-tree of a representative and thereby reaches the entire set.
            \item \textit{Guessing and completing the cores.} The algorithm guesses the pruned out-trees of a guessed set of representatives and then adds every indepdent vertex $x$ that satisfies two local conditions: \begin{enumerate}
                \item for every guessed tree, $x$ is either an internal vertex of the tree or can be attached as a leaf, and
                \item there exists a guessed tree such that $x$ is either the representative of that tree or can reach that representative (and can therefore delegate).
            \end{enumerate}   
            The first condition ensures that all representatives reach $x$, while the second ensures that $x$ reaches every representative. Intuitively, the union of the representatives form a pivot-like structure that connects the entire closed temporally connected set.
        \end{enumerate}
        The two cases of strict and non-strict temporal graphs differ in the details of the third step. In the non-strict case, the delegation is straightforward: one representative tree per cover vertex is enough. In the strict case, the delegation is more restricted and we must keep individual trees for each cover vertex as well as up to two representatives for each weak-order type\footnote{This refers to the weak-order types of \Cref{def:weak-order-type}.} of cover vertices.

        We first establish the common framework of out-tree certificates and then describe the two algorithms separately.
        
        \subsubsection{The common certificate framework}
        The connectivity of a closed temporally connected set $C\subseteq V$ can be witnessed by $\lvert C\rvert$ tree subgraphs, one rooted at each vertex of $C$. The tree rooted at $v$ certifies that $v$ reaches every vertex of $C$ using only vertices of $C$. Such witnesses can be obtained, for example, by an earliest-arrival search from each vertex.
        \begin{definition}[temporal out-tree]\label{def:temporal-out-tree}
            Let $C\subseteq V$ and $v\in C$. A \emph{temporal out-tree rooted at $v$ spanning $C$} is a temporal spanning subgraph \outtree{v} of $\gcal[C]$ with exactly $|C|-1$ temporal edges whose footprint is a tree and in which $v$ has a temporal path to every vertex of $C$.
        \end{definition}
        \begin{proposition}\label{prop:temporal-out-trees}
            For every $C\subseteq V$, the following are equivalent.
            \begin{enumerate}
                \item $C$ is closed temporally connected;
                \item $\gcal[C]$ contains a temporal spanner spanning $C$;
                \item there exists a collection $(\outtree{v})_{v\in C}$ of temporal out-trees such that, for every $v\in C$, the out-tree \outtree{v} is rooted at $v$ and spans $C$, and their union $\bigcup_{v\in C}\outtree{v}$ is a minimum temporal spanner of $\gcal[C]$.
            \end{enumerate}
        \end{proposition}
        \begin{proof}
            The first two statements are equivalent by definition, and the third clearly implies the second. 
            Conversely, let $S$ be a minimum temporal spanner of $\gcal[C]$. For every $v\in C$, an earliest-arrival search from $v$ within $S$ reaches all of $C$. Retaining, for every vertex other than $v$, the temporal edge by which it is first reached yields exactly $|C|-1$ temporal edges, a tree footprint, and a temporal path from $v$ to every vertex. Hence, the retained subgraph is a temporal out-tree \outtree{v} rooted at $v$ and contained in $S$. The union $H:=\bigcup_{v\in C}\outtree{v}$ is temporally connected and satisfies $H\subseteq S$. Since $S$ is minimum, it follows that $H=S$.
        \end{proof}
        Note that the out-tree collection certifies that $C$ is closed temporally connected, but does not certify whether $C$ is inclusion-maximal and hence a closed tcc. The maximality will be handled separately by the algorithms.

        The size of each such out-tree can be bounded in terms of a vertex cover after removing independent-set leaves. 
        \begin{definition}[pruned temporal out-tree] \label{def:pruned-outtree}
            Let $Z$ be a vertex cover of the footprint of $\gcal$, let $I:=V\setminus Z$, and let $C\subseteq V$ be closed temporally connected. Set $Z_C:=C\cap Z$ and $I_C:=C\cap I$.
            For a temporal out-tree \outtree{v} rooted at $v\in C$ and spanning $C$, its \emph{pruned temporal out-tree} \prunedouttree{v} is obtained by deleting every leaf in $I_C$.
        \end{definition}
        Pruned temporal out-trees can be viewed as a temporal generalization of the templates from~\cite[Definition~1]{casteigts_minimumtemporal_2026}. While a template records only the static tree, a pruned temporal out-tree contains the actual vertices and temporal edges. 
        \begin{proposition}\label{prop:pruned-temporal-out-trees}
            Under the notation of \Cref{def:pruned-outtree}, \prunedouttree{v} is a temporal out-tree rooted at $v$ on its remaining vertex set and $\bigl|V(\prunedouttree{v})\bigr|\leq \max\{1,2|Z_C|\}$.
            In particular, if $Z_C\neq\emptyset$ and $|Z|=k$, then \prunedouttree{v} has at most $2|Z_C|\leq 2k$ vertices.
        \end{proposition}
        \begin{proof}
            Pruning the leaves preserves the rooted-tree footprint and every temporal path from the root to the remaining vertices. 
            Since $I$ is independent, every child of an internal vertex in $I_C$ and its unique parent (if one exists) must belong to $Z_C$.
            Thus, if $Z_C\neq\emptyset$, at most $|Z_C|$ vertices of $I_C$ remain after pruning, in addition to the $|Z_C|$ vertices of $Z_C$. If $Z_C=\emptyset$, then $C\subseteq I$ and hence $|C|\leq 1$.
        \end{proof}

        Although every pruned temporal out-tree has size bounded in terms of $k$, the na\"ive certificate from \Cref{prop:temporal-out-trees} contains one tree for every vertex of $C$, whose number is unbounded. We therefore need to show that a bounded collection of representative trees suffices.

        \subsubsection{Non-strict temporal graphs}
        In a simple and proper temporal graph, every minimum temporal spanner is the union of at most one temporal out-tree per cover vertex and at most one additional temporal edge per independent vertex~\cite[Lemma~8]{casteigts_minimumtemporal_2026}. 
        We extend this argument to arbitrary undirected non-strict temporal graphs and adapt it to closed temporally connected sets.
        The key idea is that, among a set of independent vertices whose first tree-edge points to the same cover vertex $z$, everyone including $z$ can use the out-tree with the latest first edge. We call this \textit{delegating the out-tree}. See \Cref{fig:vc-XP-NS-outtree-delegation} for an illustration.
        \begin{figure}[t]
            \centering
            \includegraphics[width=0.45\textwidth]{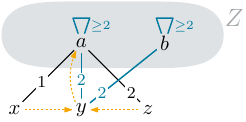}
            \caption{Argument for the bounded number of representative trees: two independent vertices $x$ and $z$ delegate their out-trees to a latest independent vertex $y$ (dotted arrows). The representative tree (blue) can be rerooted at the neighboring cover vertex $a$. For each cover vertex, it suffices to choose one such representative out-tree.}
            \label{fig:vc-XP-NS-outtree-delegation}
        \end{figure}
        \begin{lemma}[Non-strict VC-Tree Certificate Lemma]
            \label{lem:nonstrict-vc-tree-certificate}
            Let $\gcal$ be an undirected non-strict temporal graph, let $Z$ be a vertex cover of its footprint with $|Z|=k$, and let $I:=V\setminus Z$. Let $C$ be closed temporally connected and set $Z_C:=C\cap Z$ and $I_C:=C\cap I$.
            For every minimum temporal spanner $S$ of $\gcal[C]$, there is a collection $\{\outtree{v}\}_{v\in C}$ of temporal out-trees contained in $S$ such that $S$ is the union of
            \begin{itemize}
                \item at most $|Z_C|$ cover-rooted representative temporal out-trees spanning $C$ and
                \item at most one additional temporal edge incident with each remaining vertex of $I_C$.
            \end{itemize} 
            Consequently, this certificate contains at most $|Z_C|$ temporal out-trees, each of whose pruned temporal out-trees has at most $\max\{1,2|Z_C|\}$ vertices.
        \end{lemma}
        \begin{proof}
            The statement is immediate if $|C|\leq 1$, so assume that $|C|\geq2$ and let $S$ be a minimum temporal spanner of $\gcal[C]$. Then according to \Cref{prop:temporal-out-trees} there exists a collection of temporal out-trees $\{\outtree{v}\}_{v\in C}$ contained in $S$.
            For every $x\in I_C$, choose a minimum-label edge $\earlyE{x}=(x\earlyV{x},\earlyT{x})$ of \outtree{x} that is incident with $x$. Such an edge exists: on a temporal path from the root $x$ containing a minimum-label edge of the tree, the first edge has the same minimum label and is incident with $x$. Since $I_C$ is independent, $\earlyV{x}\in Z_C$. For $z\in Z_C$, set $I_z:=\{x\in I_C:\earlyV{x}=z\}$.

            For each $z\in Z_C$, we will now choose one representative $r_z$ and its out-tree.
            If $I_z\neq\emptyset$, we choose $r_z\in\arg\max_{x\in I_z}\earlyT{x}$ and its out-tree \outtree{r_z}. Since the edge $r_zz$ at time \earlyT{r_z} is a minimum-label edge of this tree and the graph is undirected and non-strict, \outtree{r_z} can also be rooted at $z$. Moreover, every $x\in I_z\setminus\{r_z\}$ can first use \earlyE{x} and then follow this rerooted tree to reach all of $C$, because $\earlyT{x}\leq\earlyT{r_z}$. This is illustrated in \Cref{fig:vc-XP-NS-outtree-delegation}. If $I_z=\emptyset$, we choose $r_z=z$ and its out-tree \outtree{z} instead.

            Let $H$ be the union of the chosen out-trees and the edge \earlyE{x} for every non-representative $x\in I_C$. Every vertex of $C$ reaches all of $C$ in $H$, so $H$ is a temporal spanner contained in $S$. Since $S$ is minimum, $H=S$. Finally, \Cref{prop:pruned-temporal-out-trees} gives the size bound for every chosen out-tree.
        \end{proof}
        We now have all the tools to present and prove the algorithm.
        \begin{theorem}\label{thm:undirected-vc-xp-closed-nonstrict}
            \closedTCC on undirected non-strict temporal graphs with vertex cover number $\vc$ can be solved in time
            $M^{\mathcal{O}(\vc^2)}(n+M)^{\mathcal{O}(1)}$.
        \end{theorem}
        \begin{proof}[Proof of \Cref{thm:undirected-vc-xp-closed-nonstrict}]
            Let $\gcal=(V,E,\lambda)$ be an undirected non-strict temporal graph with vertex cover number $\vc(\gcal)=k$. A minimum vertex cover $Z$ of the footprint can be computed in $2^kn^{\mathcal{O}(1)}$ time using an exhaustive search algorithm. Let $I:=V\setminus Z$. We call the vertices from $Z$ \textit{cover-vertices} and the ones from $I$ \textit{independent-vertices}.

            For a closed temporally connected set $C$, we use the notation $Z_C,I_C,I_z,r_z,\earlyE{x},\earlyT{x}$, and \earlyV{x} introduced in the proof of \Cref{lem:nonstrict-vc-tree-certificate}.
            
            \bigparagraph{Bounded guesses.} We enumerate all options for the following guesses.
            \begin{enumerate}
                \item Guess a set $Z_C\subseteq Z$. 
                \item For every $z\in Z_C$, guess a temporal subgraph $K_z$ of $\mathcal G$ satisfying the following properties: \begin{itemize}
                    \item $K_z$ is a temporal rooted tree, \ie the root can reach every vertex $V(K_z)$ via $K_z$ and the footprint is a tree;
                    \item $\lvert V(K_z)\rvert\leq 2\lvert Z_C\rvert$;
                    \item $V(K_z)\cap Z=Z_C$;
                    \item every leaf of $K_z$ is a cover vertex;
                    \item $K_z$ is rooted at a representative $r_z\in\{z\}\cup I$. If $r_z\in I$, we additionally require that $r_zz$ is an edge of $K_z$ with minimum time label in $K_z$. We denote the label of $r_zz$ by $t_z$.
                \end{itemize}
                We discard any guess that does not satisfy these requirements.
            \end{enumerate}
            Let $I_K\subseteq I$ be the union of all independent vertices occurring in one of the pruned out-trees~$K_z$. The goal is now to compute the largest closed temporally connected set that is consistent with this guess, i.e., to determine which vertices of $I\setminus I_K$ can be added.

            \bigparagraph{Completing a guess.}
            To that end, let $K_z$ be a guessed pruned out-tree with root $r_z$. For every $z'\in Z_C$, we denote by $\arr{K_z}{z'}$ the time label of the last edge on the unique $r_z$-$z'$ path in $K_z$, and set $\arr{K_z}{z}=-\infty$ if $r_z=z$.
            A vertex $x\in I\setminus V(K_z)$ is \emph{attachable} to $K_z$ if there exists a cover vertex $z'\in Z_C$ and a time $t\in\lambda(z'x)$ such that $\arr{K_z}{z'}\leq t$. In this case, $x$ can be added as a leaf with parent $z'$ to $K_z$ while preserving that $K_z$ is a temporal out-tree of~$r_z$. We say that an independent vertex $x\in I$ satisfies the \emph{leaf condition} if, for every $z\in Z_C$, either $x\in V(K_z)$ or $x$ is attachable to $K_z$.
            Furthermore, we say that an independent vertex $x\in I$ satisfies the \emph{source condition} if either $x$ is one of the guessed representatives (and thereby a root of an out-tree), or there exists an independent representative $r_z\in I$ and a time $t\in\lambda(xz)$ such that $t\leq t_z$.

            The algorithm now checks for every vertex in $I_K$ whether it satisfies the leaf condition and the source condition. If this is not the case, the guess is discarded. Otherwise, we set $I_{\mathrm{add}}:=\{x\in I\setminus I_K:x\text{ satisfies the leaf and source conditions}\}$, and record $C':=Z_C\cup I_K\cup I_{\mathrm{add}}$ as a candidate solution.
             
            \bigparagraph{Soundness (the candidate $C'$ is closed temporally connected).}
            Fix $z\in Z_C$. Take $K_z$ and attach every independent vertex $x\in C'\setminus V(K_z)$ as a leaf using an edge that witnesses the leaf condition of $x$ for $K_z$. Denote the resulting temporal tree by $\widehat T_z$. By definition of the leaf condition, $\widehat T_z$ is a temporal out-tree rooted at $r_z$ spanning $C'$.
            If $r_z=z$, then $\widehat T_z$ directly witnesses that $z$ reaches every vertex of $C'$. Otherwise $r_z\in I$ is the representative of $z$. Recall that $r_zz$ has time label $t_z$ and is a minimum-label edge of $K_z$. Every attached edge has label at least $t_z$, so $r_zz$ remains a minimum-label edge of $\widehat T_z$. Since the graph is undirected and non-strict, $\widehat T_z$ can therefore also be rooted at $z$. Consequently, every cover vertex in $Z_C$ reaches every vertex of $C'$.

            Let $x=r_z\in I$ be an independent representative. Then $\widehat T_z$ is rooted at $x$ and directly witnesses that $x$ reaches every vertex of $C'$.

            Let $x\in (C'\cap I)\setminus\{r_z:z\in Z_C\}$. Since $x$ satisfies the source condition, there exists an independent representative $r_z$ and an edge $xz$ at some time $t\leq t_z$. The vertex $x$ uses this edge and then follows the $z$-rooted version of $\widehat T_z$. Thus every vertex of $C'$ reaches every other vertex using only vertices of $C'$, and $C'$ is closed temporally connected.

            The candidate $C'$ does not have to be inclusion-maximal. However, it is contained in an inclusion-maximal closed temporally connected set, \ie in a ctcc. 

            \bigparagraph{Completeness (a maximum closed tcc is found).}
            Let $C$ be a maximum closed tcc, let $Z_C:=C\cap Z$, $I_C:=C\cap I$, and consider the pruned out-trees from \Cref{prop:pruned-temporal-out-trees}. For every $z\in Z_C$ with $I_z\neq\emptyset$, choose $r_z\in\argmax_{x\in I_z}\earlyT{x}$ and take \prunedouttree{r_z}; if $I_z=\emptyset$, take \prunedouttree{z}. The algorithm eventually guesses $Z_C$ together with exactly these pruned out-trees.

            We show that every independent vertex of $C$ passes the two conditions for this guess. Fix $x\in I_C$ and $z\in Z_C$, and consider the pruned tree $K_z$. If $x$ occurs in $K_z$, the leaf condition for this tree holds directly. Otherwise, $x$ was removed when pruning the corresponding full out-tree and was therefore a leaf in that tree. Since $I$ is independent, the parent of $x$ is some cover vertex $z'\in Z_C$. If the tree edge $z'x$ has time $t$, then the temporal path from the root to $z'$ arrives no later than $t$, and the same path is contained in $K_z$. Thus $\arr{K_z}{z'}\leq t$ and $x$ is attachable to $K_z$. Since $z\in Z_C$ was arbitrary, $x$ satisfies the leaf condition. 
            Now for the source condition, set $z:=z^-(x)$. Then $x\in I_z$, and in particular $I_z\neq\emptyset$, so the guess contains a representative $r_z$. By the choice of $r_z$, we have $\earlyT{x}\leq\earlyT{r_z}=t_z$. Since the edge $xz$ is present at time $\earlyT{x}$, this edge witnesses the source condition of $x$. Hence every independent vertex of $C$, including the ones contained in $I_K$, satisfies the required conditions.

            Consequently, $C\subseteq C'$ for this guess. We have already shown that $C'$ is closed temporally connected and is contained in an inclusion-maximal closed temporally connected set $D$, which is a closed tcc. The maximality of the size of $C$ therefore gives $|C'|\leq |D|\leq |C|$. Together with $C\subseteq C'$, this implies $C'=C$. Thus at least one of our guesses recovers a maximum closed tcc.

            \bigparagraph{Running time.}
            There are at most $2^k$ possibilities for $Z_C$.
            For every $z\in Z_C$, we guess one pruned temporal out-tree on at most $2k$ vertices. Hence, over all at most $k$ trees, we guess only $\mathcal{O}(k^2)$ temporal edges. These edges, together with the choice of a root, determine all vertices of each nontrivial tree.
            Thus, all guesses can be enumerated in $M^{\mathcal{O}(k^2)}$ time. 
            For a fixed guess, we can check in polynomial time that every $K_z$ has a tree footprint, satisfies the size and vertex-cover requirements, has only cover vertices as leaves, and is a non-strict temporal out-tree from its proposed root. If the root is an independent representative, we also check the required minimum-label edge. The arrival times in the valid guessed trees and the leaf and source conditions for every independent vertex can likewise be computed in $(n+M)^{\mathcal{O}(1)}$ time. Therefore, \closedTCC on undirected non-strict temporal graphs can be solved in $M^{\mathcal{O}(k^2)}(n+M)^{\mathcal{O}(1)}$ time.
        \end{proof}

        \subsubsection{Strict temporal graphs}
        Two parts of the representative-tree argument break for strict temporal graphs. They are illustrated in the figure on the right:
        \par\medskip
        \noindent\begin{minipage}[c]{0.56\textwidth}
            \begin{enumerate}
                \item 
                    A cover vertex $z$ cannot always delegate its out-tree to an independent vertex $y$. When $y$ has several minimum-label edges, rerooting at $z$ leads to two consecutive edges in the tree with the same label, which does not form a strict temporal path.
                \item One representative tree per cover vertex no longer suffices: independent vertices $x,y\in I$ having their first out-tree-edge to the same cover vertex $z$ at the same latest time cannot always delegate to one another by the same argument.
            \end{enumerate}
        \end{minipage}%
        \hfill
        \begin{minipage}[c]{0.42\textwidth}
            \centering
            \includegraphics[width=\linewidth]
                {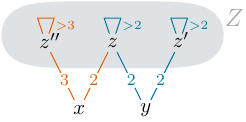}
        \end{minipage}
        \par\medskip
        \noindent We resolve the first issue by guessing additional out-trees rooted at each cover-vertex and keeping the independent roots of the other out-trees.
        Towards resolving the second issue, we define a weak-order type -- a closed-analogue of the latest-edge type used in the open proof (see \Cref{def:latest-edge-type}) -- and keep up to two representatives for each of these types.
        \begin{definition}[earliest-edge type]\label{def:earliest-edge-type}
            Let $Z$ be a vertex cover of the footprint, let $I:=V\setminus Z$, and let $T$ be a temporal out-tree rooted at $x\in I$. Set $Y:=\{z\in Z:(xz,t)\in\mathcal E(T)\text{ for some }t\}$. For every $z\in Y$, let $\ell_{x,T}^-(z)$ be the label of the unique temporal edge between $x$ and $z$ in $T$.
            The \emph{earliest-edge type} of $x$ in $T$ with respect to $Z$ is the weak-order type $\tau_Z^-(x;T):=(Y,\preceq_{x,T}^-)$, where $z\preceq_{x,T}^-z'$ if and only if $\ell_{x,T}^-(z)\leq\ell_{x,T}^-(z')$.
        \end{definition}
        For a fixed out-tree, every independent vertex has a unique earliest-edge type. By \Cref{obs:number-weak-order-types}, at most $(k+1)^k=2^{\mathcal O(k\log k)}$ such types exist for a vertex cover of size $k$.
        
        \begin{lemma}[Strict VC-Tree Certificate Lemma]\label{prop:strict-vc-tree-certificate}
            Let $\gcal$ be an undirected strict temporal graph, let $Z$ be a vertex cover of its footprint with $|Z|=k$, and let $I:=V\setminus Z$. Let $C$ be closed temporally connected and set $Z_C:=C\cap Z$ and $I_C:=C\cap I$.
            For every minimum temporal spanner $S$ of $\gcal[C]$, there is a collection $\{\outtree{v}\}_{v\in C}$ of temporal out-trees contained in $S$ such that $S$ is the union of
            \begin{itemize}
                \item at most $|Z_C|$ cover-rooted temporal out-trees spanning $C$,
                \item at most two representative temporal out-trees spanning $C$ for every earliest-edge type $\tau_Z^-(x;\outtree{x})$ realized by a vertex $x\in I_C$, and
                \item at most $|Z_C|$ additional temporal edges incident with each remaining vertex of $I_C$.
            \end{itemize} 
            Consequently, this certificate contains at most $|Z_C|+2(k+1)^k=2^{\mathcal O(k\log k)}$ temporal out-trees, each of whose pruned temporal out-trees has at most $\max\{1,2|Z_C|\}$ vertices.
        \end{lemma}
        \begin{proof}
            The statement is immediate if $\lvert C\rvert\leq 1$, so assume that $\lvert C\rvert \geq 2$ and let $S$ be a minimum temporal spanner of $\gcal[C]$.
            Then according to \Cref{prop:temporal-out-trees}, there exists a collection of temporal out-trees $\{\outtree{v}\}_{v\in C}$ spanning $C$ whose union is $S$.

            We keep the out-tree \outtree{z} for every $z\in Z_C$. It remains to choose a bounded number of representative out-trees for the vertices of $I_C$.
            For every $x\in I_C$, let $F_x:=\{(xz,t)\in\mathcal E(\outtree{x}):z\in Z_C\}$ be the set of temporal edges of \outtree{x} incident with its root.
            Since $I_C$ is independent, every neighbor of $x$ in \outtree{x} belongs to $Z_C$. Consequently, $|F_x|\leq |Z_C|$.

            Fix an earliest-edge type $\sigma=(Y,\preceq)$ realized by some vertices of $I_C$, and let $I_\sigma$ be the set of all vertices $x\in I_C$ for which $\tau_Z^-(x;\outtree{x})=\sigma$.
            Let $\sigma_1,\ldots,\sigma_q$ be the equivalence classes of $\preceq$, ordered from earliest to latest. For every $x\in I_\sigma$ and $i\in[q]$, denote by $t_i(x)$ the common label of the edges $xz$ with $z\in\sigma_i$ in \outtree{x}. Thus, $t_1(x)<t_2(x)<\cdots<t_q(x)$.

            \bigparagraph{Delegating an out-tree.}
            Let $x,y\in I_\sigma$, and assume that there is $h\in[q]$ such that $t_i(x)=t_i(y)$ for every $i<h$ and $t_h(x)<t_h(y)$.
            We claim that $x$ reaches every vertex of $C$ in $F_x\cup\outtree{y}$ and can thus delegate its out-tree to $y$.
            First, observe that $x$ reaches $y$ using the path $(xz,zy)$, for an arbitrary $z\in\sigma_h$, since $t_h(x)<t_h(y)$.
            Now fix $v\in C\setminus\{x,y\}$. The unique $y$-$v$ path in the footprint of \outtree{y} starts with an edge $yz'$ for some $z'\in\sigma_i$.
            If $i\leq h$, then $x$ first uses the edge $xz'\in F_x$ and afterwards follows the suffix of this $y$-$v$ path starting at $z'$. 
            If $i>h$, then $x$ first uses $xz$ at time $t_h(x)$, traverses $zy$ at time $t_h(y)$, and subsequently follows the $y$-$v$ path in \outtree{y}. This is a strict temporal path because $t_h(x)<t_h(y)<t_i(y)$.
            Hence, $x$ can delegate its out-tree to $y$ while keeping only the edges in $F_x$, of which there are at most $|Z_C|$. See \Cref{fig:vc-XP-strict-outtree-delegation} for an illustration.
            \begin{figure}[t]
                \centering
                \includegraphics[width=0.4\textwidth]{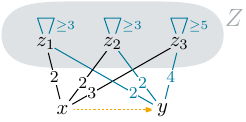}
                \caption{Delegation of an out-tree from $x$ to $y$ (dotted yellow arrow) when $t_h(x)<t_h(y)$ for the earliest index $h$ with $t_h(x)\neq t_h(y)$. Here, $h=3$.}
                \label{fig:vc-XP-strict-outtree-delegation}
            \end{figure}

            \bigparagraph{Choosing the representatives.} We will now successively delegate the independent vertices with the same earliest-edge type $\sigma$.
            Set 
            \[ A_0:=I_\sigma \qquad\text{and}\qquad A_i:=\argmax_{x\in A_{i-1}} t_i(x) \quad\text{for }i=1,\ldots,q.\]
            All vertices in $A_{i-1}$ have the same time label values in the first $i-1$ coordinates. Thus, every vertex in $A_{i-1}\setminus A_i$ can delegate to every vertex of $A_i$ according to the delegation claim above. 
            Additionally, all vertices of $A_q$ have exactly the same time label vector $(t_1(x),\ldots,t_q(x))$.

            If $|A_q|=1$, we choose its unique vertex $x\in A_q$ as the representative of type $\sigma$. Every other vertex in $I_\sigma$ can delegate to $x$ by the claim above.

            If $|A_q|\geq2$, we choose two arbitrary distinct vertices $r_\sigma,s_\sigma\in A_q$ as representatives.
            Let $x\in A_q\setminus\{r_\sigma,s_\sigma\}$. We show that $x$ reaches all of $C$ in $F_x\cup\outtree{r_\sigma}\cup\outtree{s_\sigma}$.
            Consider first some $v\in C\setminus\{r_\sigma\}$. The unique $r_\sigma$-$v$ path in \outtree{r_\sigma} starts with an edge $r_\sigma z$ for some $z\in Y$. Since $x$ and $r_\sigma$ have identical time vectors, the edge $xz\in F_x$ has the same label as $r_\sigma z$, while the remaining suffix of the $r_\sigma$-$v$ path has strictly larger labels. Hence, $x$ reaches $v$ by using $xz$ and then this suffix.
            Next, $x$ must also reach $r_\sigma$. For this, it takes the $s_\sigma$-$r_\sigma$ path in \outtree{s_\sigma}. This path contains at least two edges because $r_\sigma,s_\sigma\in I_C$ and $I_C$ is independent. Replacing its first edge $s_\sigma z$ by the equally labeled edge $xz\in F_x$ gives a strict temporal path from $x$ to $r_\sigma$.
            Consequently, every vertex of $A_q\setminus\{r_\sigma,s_\sigma\}$ can delegate to the two out-trees of $r_\sigma$ and $s_\sigma$. Every vertex removed in an earlier step can delegate to either representative by the first claim.

            We apply this selection independently to every realized earliest-edge type.
            Let $H$ be the union of the cover-rooted out-trees, the at most two chosen representative out-trees for every type, and the set $F_x$ for every remaining vertex $x\in I_C$.
            Every cover vertex and every representative reaches all of $C$ in its own out-tree. By the two delegation arguments, every remaining vertex of $I_C$ also reaches all of $C$ in $H$. Hence, $H$ is a temporal spanner of $\gcal[C]$.
            Moreover, $H\subseteq S$. Since $S$ is a minimum temporal spanner, it follows that $H=S$.

            Finally, there are at most $(k+1)^k$ earliest-edge types by \Cref{obs:number-weak-order-types}. We choose at most two representative out-trees for each type and one cover-rooted out-tree for each vertex of $Z_C$. The claimed number of out-trees follows, and \Cref{prop:pruned-temporal-out-trees} gives the size bound for every pruned out-tree.
        \end{proof}
        We now have all the tools to present and prove the algorithm.
        \begin{theorem}\label{thm:undirected-vc-xp-closed-strict}
            \closedTCC on undirected strict temporal graphs with vertex cover number $\vc$ can be solved in time $M^{2^{\mathcal O(\vc\log\vc)}}(n+M)^{\mathcal O(1)}$.
        \end{theorem}
        \begin{proof}[Proof of \Cref{thm:undirected-vc-xp-closed-strict}]
            Let $\gcal=(V,E,\lambda)$ be an undirected strict temporal graph with vertex cover number $\vc(\gcal)=k$. A minimum vertex cover $Z$ of the footprint can be computed in $2^kn^{\mathcal O(1)}$ time using an exhaustive search algorithm. Let $I:=V\setminus Z$. We call the vertices from $Z$ \textit{cover-vertices} and the ones from $I$ \textit{independent-vertices}.

            For an earliest-edge type $\sigma=(Y,\preceq)$, let $\sigma_1,\ldots,\sigma_q$ denote the equivalence classes of $\preceq$, ordered from earliest to latest. If a temporal out-tree $K$ rooted at $x\in I$ has type $\sigma$, we denote its time label vector by $\mathbf t_K(x):=(t_1(x),\ldots,t_q(x))$, where $t_i(x)$ is the common label of the edges from $x$ to the vertices in $\sigma_i$.

            \bigparagraph{Bounded guesses.} We enumerate all options for the following guesses.
            \begin{enumerate}
                \item Guess a set $Z_C\subseteq Z$.
                \item For every $z\in Z_C$, guess a temporal subgraph $K_z$ of $\gcal$ satisfying the following properties:
                \begin{itemize}
                    \item $K_z$ is a strict temporal rooted tree with root $z$;
                    \item $|V(K_z)|\leq 2|Z_C|$;
                    \item $V(K_z)\cap Z=Z_C$; and
                    \item every leaf of $K_z$ is a cover vertex.
                \end{itemize}
                \item For every earliest-edge type $\sigma$ with respect to $Z$, guess at most two temporal subgraphs $K_\sigma^j$ of $\gcal$ satisfying the following properties:
                \begin{itemize}
                    \item $K_\sigma^j$ is a strict temporal rooted tree with an independent root $r_\sigma^j\in I$;
                    \item $|V(K_\sigma^j)|\leq 2|Z_C|$;
                    \item $V(K_\sigma^j)\cap Z=Z_C$;
                    \item every leaf of $K_\sigma^j$ is a cover vertex; and
                    \item $\tau_Z^-(r_\sigma^j;K_\sigma^j)=\sigma$.
                \end{itemize}
                If we guess two trees, we additionally require that their roots are distinct and that $\mathbf t_{K_\sigma^1}(r_\sigma^1)=\mathbf t_{K_\sigma^2}(r_\sigma^2)$.
            \end{enumerate}
            We discard any guess that does not satisfy these requirements. Let $\mathcal K$ be the family of all guessed trees and let $I_{\mathcal K}\subseteq I$ be the union of their independent vertices.

            \bigparagraph{Completing a guess.}
            Let $K\in\mathcal K$ be a guessed pruned out-tree with root $r$. For every $z\in Z_C$, we denote by $\arr{K}{z}$ the time label of the last edge on the unique $r$-$z$ path in $K$, and set $\arr{K_z}{z}=-\infty$.
            A vertex $x\in I\setminus V(K)$ is \emph{attachable} to $K$ if there exists a cover vertex $z\in Z_C$ and a time $t\in\lambda(xz)$ such that $\arr{K}{z}<t$. In this case, $x$ can be added as a leaf with parent $z$ to $K$ while preserving that $K$ is a strict temporal out-tree rooted at~$r$. We say that an independent vertex $x\in I$ satisfies the \emph{leaf condition} if, for every $K\in\mathcal K$, either $x\in V(K)$ or $x$ is attachable to $K$.

            Next we determine which independent vertices can delegate their out-tree to the guessed representatives.
            Let $\sigma$ be a type for which at least one representative tree is guessed and set $\mathbf t_\sigma:=\mathbf t_{K_\sigma^1}(r_\sigma^1)$ as the common time label vector of the roots $r_\sigma^1$ and $r_\sigma^2$. We say that a vector $\mathbf t(x)=(t_1(x),\ldots,t_q(x))$ \emph{realizes $\sigma$ at $x$} if $t_1(x)<\cdots<t_q(x)$ and, for every $i\in[q]$ and $z\in\sigma_i$, the edge $xz$ is present at time $t_i(x)$.
            An independent vertex $x\in I$ satisfies the \emph{delegation condition} if either \begin{itemize} 
                \item $x$ is one of the guessed representatives, or 
                \item there exists a represented type $\sigma$ and a vector $\mathbf t(x)$ realizing $\sigma$ at $x$ such that \begin{itemize}
                    \item $\mathbf t(x)<_{\mathrm{lex}}\mathbf t_\sigma$, or 
                    \item $\mathbf t(x)=\mathbf t_\sigma$ and two representative trees were guessed for~$\sigma$.
                \end{itemize}
            \end{itemize}
            The algorithm now checks for every vertex in $I_{\mathcal K}$ whether it satisfies the leaf condition and the delegation condition. If this is not the case, the guess is discarded. Otherwise, we set $I_{\mathrm{add}}:=\{x\in I\setminus I_{\mathcal K}:x\text{ satisfies the leaf and delegation conditions}\}$, and record $C':=Z_C\cup I_{\mathcal K}\cup I_{\mathrm{add}}$ as a candidate solution.

            \bigparagraph{Soundness (the candidate $C'$ is closed temporally connected).}
            For every $K\in\mathcal K$, take $K$ and attach every independent vertex $x\in C'\setminus V(K)$ as a leaf using an edge that witnesses the leaf condition of $x$ for $K$. Denote the resulting temporal tree by $\widehat K$. By definition of the leaf condition, $\widehat K$ is a strict temporal out-tree with the same root as $K$ and spans $C'$. Since all attached edges are incident with a new independent leaf, every representative keeps its earliest-edge type and time label vector.

            For every $z\in Z_C$, the temporal out-tree $\widehat K_z$ is rooted at $z$ and directly witnesses that $z$ reaches every vertex of $C'$. Similarly, every independent representative $r_\sigma^j$ reaches every vertex of $C'$ in its temporal out-tree $\widehat K_\sigma^j$.

            Let $x\in C'\cap I$ not be a representative, and let $\sigma$ and $\mathbf t(x)$ witness its delegation condition.
            If $\mathbf t(x)<_{\mathrm{lex}}\mathbf t_\sigma$, let $h$ be the first coordinate in which these vectors differ. Then $t_i(x)=t_i(r_\sigma^1)$ for every $i<h$ and $t_h(x)<t_h(r_\sigma^1)$. Hence, the delegation argument from \Cref{prop:strict-vc-tree-certificate}, applied to the edges realizing $\mathbf t(x)$ and the tree $\widehat K_\sigma^1$, shows that $x$ reaches every vertex of~$C'$.
            Otherwise, $\mathbf t(x)=\mathbf t_\sigma$ and two representative trees were guessed for $\sigma$. Then the equal-vector argument from  \Cref{prop:strict-vc-tree-certificate}, applied to $\widehat K_\sigma^1$ and $\widehat K_\sigma^2$, shows that $x$ reaches every vertex of $C'$.
            Thus every vertex of $C'$ reaches every other vertex using only vertices of $C'$, and $C'$ is closed temporally connected.

            The candidate $C'$ does not have to be inclusion-maximal. However, it is contained in an inclusion-maximal closed temporally connected set, \ie in a ctcc.

            \bigparagraph{Completeness (a maximum closed tcc is found).}
            Let $C$ be a maximum closed tcc, let $Z_C:=C\cap Z$ and $I_C:=C\cap I$, and consider the out-trees from \Cref{prop:strict-vc-tree-certificate}. Take the pruned out-tree of every cover-rooted tree and every representative tree in this certificate. Pruning preserves all edges incident with an independent root. Thus, the pruned representative trees have the same earliest-edge types and time label vectors as the full trees. The algorithm eventually guesses $Z_C$ together with exactly these pruned out-trees.

            We show that every independent vertex of $C$ passes the two conditions for this guess. Fix $x\in I_C$ and $K\in\mathcal K$. If $x$ occurs in $K$, then the leaf condition for this tree holds directly. Otherwise, $x$ was removed when pruning the corresponding full out-tree and was therefore a leaf in that tree. Since $I$ is independent, the parent of $x$ is some cover vertex $z\in Z_C$. If the tree edge $zx$ has time $t$, then the strict temporal path from the root to $z$ arrives before $t$, and the same path is contained in $K$. Thus $\arr{K}{z}<t$ and $x$ is attachable to $K$. Since $K\in\mathcal K$ was arbitrary, $x$ satisfies the leaf condition.

            Now consider the delegation condition. Let $\sigma$ and $\mathbf t(x)$ be the earliest-edge type and the time label vector of $x$ in its out-tree from the certificate lemma. If $x$ is a representative, then the condition holds directly. Otherwise, if $x$ was removed in some $A_{i-1}\setminus A_i$, then $\mathbf t(x)$ is lexicographically smaller than the vector of the representative of type $\sigma$. If $x\in A_q$ was not chosen as a representative, then two representatives with the same vector as $x$ were chosen. In both cases, the edges in $F_x$ show that $\mathbf t(x)$ realizes $\sigma$ at $x$. Hence, every $x\in I_C$ satisfies the delegation condition.

            Consequently, $C\subseteq C'$ for this guess. We have already shown that $C'$ is closed temporally connected and is contained in an inclusion-maximal closed temporally connected set $D$, which is a closed tcc. The maximality of the size of $C$ therefore gives $|C'|\leq |D|\leq |C|$. Together with $C\subseteq C'$, this implies $C'=C$. Thus at least one of our guesses recovers a maximum closed tcc.

            \bigparagraph{Running time.}
            There are at most $2^k$ possibilities for $Z_C$ and at most $(k+1)^k$ earliest-edge types. We guess at most $|Z_C|+2(k+1)^k=2^{\mathcal O(k\log k)}$ pruned temporal out-trees, each on at most $2k$ vertices. Hence, all guessed trees together contain at most $2^{\mathcal O(k\log k)}$ temporal edges, and all guesses can be enumerated in $M^{2^{\mathcal O(k\log k)}}$ time.

            For a fixed guess, all requirements for the trees, roots, types, and equal time label vectors can be checked in time $2^{\mathcal O(k\log k)}(n+M)^{\mathcal O(1)}$. The values $\arr{K}{z}$ and the leaf conditions can be computed in the same time. For a vertex $x$ and a type $\sigma$, we can determine whether $x$ has a realizing vector that is lexicographically smaller than or equal to $\mathbf t_\sigma$ by scanning the available labels for the equivalence classes of $\sigma$. Thus, all delegation conditions can also be checked in time $2^{\mathcal O(k\log k)}(n+M)^{\mathcal O(1)}$. Therefore, \closedTCC on undirected strict temporal graphs can be solved in $M^{2^{\mathcal O(k\log k)}}(n+M)^{\mathcal O(1)}$ time.
        \end{proof}

    \subsection{\Wone-Hardness}\label{subsec:Whardness-vc-tpn}
        In this subsection, we show the complementing \Wone-hardness results by reducing from \textup{\textsc{Mixed 2-Sided Orthant Depth}}. We first recall the source problem and then give the reductions for undirected vertex cover and temporal path number.

            
A \emph{mixed 2-sided orthant} in $\mathbb R^d$ is a region $O=\{x\in\mathbb R^d:x_i\geq a,\ x_j\leq b\}$, for $i,j\in[d],\ i\neq j$.
We denote its lower-bounded dimension and value by $\mathsf{ldim}(O)=i$ and $\mathsf{lval}(O)=a$, and its upper-bounded dimension and value by $\mathsf{udim}(O)=j$ and $\mathsf{uval}(O)=b$.
In \textup{\textsc{Mixed 2-Sided Orthant Depth (M2SOD)}}, we are given a finite indexed family $\mathcal O$ of such regions and an integer $\ell$, and ask whether a point of $\mathbb R^d$ belongs to at least $\ell$ members of $\mathcal O$.
Repeated regions are allowed and counted separately.
This problem is \Wone-hard when parameterized by $d$; see \Cref{thm:M2SOD-wone-hard} in the appendix.

Two orthants intersect precisely when neither lower bound conflicts with the other orthant's upper bound in the same dimension. Thats is, two orthants $O_1,O_2$ intersect iff (1) $\mathsf{ldim(O_1)}\neq \mathsf{udim(O_2)}$ or (2) $\mathsf{ldim(O_1)}=\mathsf{udim(O_2)}$ and $\mathsf{lval(O_1)}\leq \mathsf{uval(O_2)}$, and (1) or (2) hold with reversed roles.
Moreover, a finite family of these regions has a common point if and only if its members intersect pairwise.
Replacing all bound values by their ranks preserves all lower--upper comparisons and hence all intersections and depths. Thus, for a nonempty family $\mathcal O$, we may assume that the bound values belong to $[n]$, where $n\leq2|\mathcal O|$ is the number of distinct bound values.

        \subsubsection{\openTCC and \closedTCC in undirected temporal graphs parameterized by vertex cover}\label{subsec:undirected-vc-hardness}
            We reduce \textup{\textsc{M2SOD}} to \openTCC and \closedTCC in undirected temporal graphs parameterized by vertex cover using the rank-comparison gadgets introduced in \Cref{subsec:directed-vc-hardness}.
        \undirectedVCHardness*
        Let $(\mathcal O,\ell)$ be an instance of \textup{\textsc{M2SOD}} in dimension $d$. We construct an undirected temporal graph $\mathcal G$ and an integer $k'$ such that a point of $\mathbb R^d$ is contained in at least $\ell$ orthants of $\mathcal O$ if and only if $\mathcal G$ has an open (resp. closed) tcc of size at least $k'$. We may assume that $\mathcal O\neq\emptyset$; moreover, $d\geq2$ by the definition of mixed 2-sided orthants.

        \bigparagraph{Intuition.}
        We represent every orthant $O$ by an \emph{orthant-vertex} $z_O$. For every dimension~$i$, we add a rank-comparison gadget that connects the vertices whose orthants are lower-bounded in dimension $i$ to those whose orthants are upper-bounded in dimension $i$ and have an upper value at least as large. The two cover vertices of each gadget are denoted by $L_i$ and $U_i$. To connect orthant-vertices whose relevant bounds lie in different dimensions $i\neq j$, an additional path from $L_i$ to $U_j$ provides reachability without comparing values.
        Lastly, we add auxiliary vertices so that the set of the rank-comparison vertices together with the auxiliary vertices is temporally connected and each of these vertices is mutually reachable with every orthant-vertex. As a result, every connected component will be both open and closed.

        \bigparagraph{Construction.}
        Recall that $\mathsf{ldim}(O)\neq\mathsf{udim}(O)$. Let $n\leq2|\mathcal O|$ be the number of distinct bound values. Replacing every bound value by its rank preserves all lower--upper comparisons, so we may assume that $\mathsf{lval}(O),\mathsf{uval}(O)\in[n]$ for every $O\in\mathcal O$. 
        For every occurrence of an orthant $O\in\mathcal O$, introduce an {\em orthant-vertex} $z_O$, and let $Z$ be the set of all orthant-vertices. Define the rank functions ${r_1},{r_2}:Z\to[n]$ by ${r_1}(z_O):=\mathsf{lval}(O)$ and ${r_2}(z_O):=\mathsf{uval}(O)$.
        
        For every $i\in[d]$, let $Z_i^L:=\{z_O\in Z:\mathsf{ldim}(O)=i\}$ and $Z_i^U:=\{z_O\in Z:\mathsf{udim}(O)=i\}$, and fix arbitrary orders on these sets to break ties. For every $i\in[d]$, add the rank-comparison gadget
        $\mathcal Q_i:=\mathcal Q_{\leq}(Z_i^L,Z_i^U;{r_1},{r_2})$ from \Cref{def:rank-comparison}. Denote the cover vertices of $\mathcal Q_i$ by $L_i,U_i$ and its threshold vertices by $p_{i,t}$ for $t\in[n]$. For every ordered pair $i,j\in[d]$ with $i\neq j$, introduce a \emph{portal vertex} $x_{i,j}$. Finally, introduce four global auxiliary vertices $a_1,a_2,b_1,b_2$, and set
        \begin{align*}
            A_0&:=\{L_i,U_i:i\in[d]\}
                \cup\{p_{i,t}:i\in[d],\ t\in[n]\}
                \cup\{x_{i,j}:i,j\in[d],\ i\neq j\},\\
            A&:=A_0\cup\{a_1,a_2,b_1,b_2\}.
        \end{align*}
        Thus, $V(\mathcal G)=Z\mathbin{\dot\cup}A$.

        The lifetime is partitioned into five consecutive blocks $T_1,T_2,T_3,T_4,T_5$, where every label in $T_r$ precedes every label in $T_{r+1}$. Every static edge receives exactly one globally unique label, so it suffices to specify its block and its position within that block.

        In $T_1$, place, in this order, all edges $ya_1$ with $y\in A_0$, the edges $b_1a_1$, $b_2a_2$, and $a_1a_2$, and then all edges $a_2y$ with $y\in A_0$. The order within either edge family is arbitrary.

        In $T_2$, place the edges of the $L^-$ block of every $\mathcal Q_i$ in the order specified in \Cref{def:rank-comparison}. Thus, for $t=1,\ldots,n$, first all edges $z_OL_i$ with ${r_1}(z_O)=t$ appear in the fixed tie-breaking order, and then the edge $L_ip_{i,t}$. The edge families for different dimensions occupy separate subblocks in an arbitrary order.

        Partition $T_3$ into two consecutive subblocks. First place all edges $L_ix_{i,j}$, and then all edges $x_{i,j}U_j$, for ordered pairs $i,j\in[d]$ with $i\neq j$. The order within each subblock is arbitrary.

        In $T_4$, place the edges of the $L^+$ block of every $\mathcal Q_i$ in their prescribed order. Thus, for $t=1,\ldots,n$, first the edge $p_{i,t}U_i$ appears and then all edges $U_iz_O$ with ${r_2}(z_O)=t$, in the fixed tie-breaking order. Again, the edge families for different dimensions occupy separate subblocks in an arbitrary order.

        Finally, in $T_5$, place, in this order, all edges $yb_1$ with $y\in A_0$, the edge $b_1a_2$, the edge $b_1b_2$, all edges $b_2y$ with $y\in A_0$, and the edge $b_2a_1$. The order within either edge family is arbitrary. Set $k':=|A|+\ell$.
        This completes the construction, which is computable in polynomial time. Refer to \Cref{fig:vc-Wh} for an illustration of the construction.
        \begin{figure}[t]
            \centering
            \includegraphics[width=0.5\textwidth]{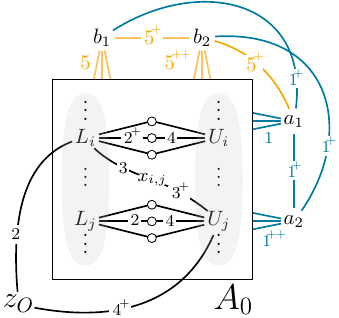}
            \caption{Schematic illustration of the five time blocks for two distinct dimensions $i$ and $j$. The time label $t$ indicates that the corresponding edges are in time block $T_t$, while superscripts $+$ and $++$ indicate their order within a block.\\
            The rank-comparison gadgets $\mathcal Q_i$ and $\mathcal Q_i$ connect $L_i$ to $U_i$, resp. $L_j$ to $U_j$, through their threshold vertices during time blocks $T_2$ and $T_4$. 
            Portal vertices connect distinct dimensions in $T_3$.
            The global vertices $a_1,a_2$ and $b_1,b_2$ connect to all but the orthant-vertices in the early block $T_1$ and the late block $T_5$, respectively, ensuring the global reachability.
            }
            \label{fig:vc-Wh}
        \end{figure}

        The temporal graph $\mathcal G$ is simple and proper: The vertex types in the comparison gadgets and portal paths are disjoint, and every vertex of $A_0$ has exactly one edge to each of $a_1,a_2,b_1,b_2$. The six pairs among the four global auxiliary vertices also appear exactly once: $a_1b_1,a_2b_2,a_1a_2$ in $T_1$, and $a_2b_1,b_1b_2,a_1b_2$ in $T_5$. Hence, no static edge is repeated, and the globally unique labels make the graph proper. Moreover, $C:=\{L_i,U_i:i\in[d]\}\cup\{a_1,a_2,b_1,b_2\}$ meets every edge. Consequently, $C$ is a vertex cover of the footprint and $\vc(G)\leq2d+4$.

        We first show that the auxiliary vertices can be included in every solution.
        \begin{lemma}\label{lem:vc-orthant-auxiliary}
            The temporal graph $\mathcal G[A]$ is temporally connected. Moreover, for every orthant-vertex $z_O\in Z$ and every auxiliary vertex $y\in A$, the vertices $z_O$ and $y$ are {mutually reachable} via temporal paths contained in $A\cup\{z_O\}$.
        \end{lemma}
        \begin{proof}
            Before $T_2$, every vertex of $A$ reaches every vertex $r\in A_0$. A vertex $y\in A_0$ uses the route $(y,a_1,a_2,r)$, while $b_1,b_2,a_1,a_2$ use the respective suffixes of $(b_1,a_1,a_2,r)$ and $(b_2,a_2,r)$.
            Taking $r=L_{\mathsf{ldim}(O)}$ and then the edge $L_{\mathsf{ldim}(O)}z_O$ in $T_2$ shows that every $y\in A$ reaches $z_O$ using only vertices of $A\cup\{z_O\}$.

            Conversely, for $z_O$ to reach every vertex in $A$, it first uses the edge $z_OU_{\mathsf{udim}(O)}$ in $T_4$. From any vertex $r\in A_0$, the edges in $T_5$ give routes to all of $A$: use $rb_1$ to reach $b_1$, continue directly to $a_2$ or $b_2$, continue through $b_2$ to any vertex of $A_0$, or finally use $b_2a_1$. Taking $r=U_{\mathsf{udim}(O)}$ proves that $z_O$ reaches every $y\in A$, again without using any other orthant-vertex.

            Finally, let $y,y'\in A$. Use the early route from $y$ to $L_1$, followed by the late route from $L_1$ to $y'$. This is a temporal walk entirely contained in $A$. Hence, $\mathcal G[A]$ is temporally connected.
        \end{proof}

        The second lemma establishes the intended correspondence between temporal reachability and intersection of orthants.
        \begin{lemma}\label{lem:vc-orthant-reachability}
            For any two distinct orthant-vertices $z_O,z_{O'}\in Z$, we have
            \[
                z_O\rightreach z_{O'}
                \quad\Longleftrightarrow\quad
                \mathsf{ldim}(O)\neq\mathsf{udim}(O')
                \ \text{or}\
                \mathsf{lval}(O)\leq\mathsf{uval}(O').
                \tag{$\ast$}
            \]
            Consequently, {$z_O$ and $z_{O'}$ are mutually reachable} if and only if $O\cap O'\neq\emptyset$.
            Whenever this reachability holds, it has a witnessing path whose internal vertices all belong to $A$. 
        \end{lemma}
        \begin{proof}
            Set $i:=\mathsf{ldim}(O)$ and $j:=\mathsf{udim}(O')$. If $i\neq j$, then $(z_O,L_i,x_{i,j},U_j,z_{O'})$ is a temporal path through $T_2,T_3,T_4$. If $i=j$ and $\mathsf{lval}(O)\leq\mathsf{uval}(O')$, then \Cref{lem:rank-comparison} gives a temporal path through $\mathcal Q_i$. More explicitly, for any $t\in[n]$ with
            $\mathsf{lval}(O)\leq t\leq\mathsf{uval}(O')$, use $(z_O,L_i,p_{i,t},U_i,z_{O'})$. In both cases, all internal vertices belong to $A$.

            Conversely, suppose that $i=j$ and that $z_O$ reaches $z_{O'}$. We show that $\mathsf{lval}(O)\leq\mathsf{uval}(O')$.
            No edge incident with an orthant-vertex appears in $T_1$, and after a temporal path enters $T_5$, no edge to an orthant-vertex remains. It therefore suffices to consider the edges in $T_2,T_3,T_4$.
            Observe that a temporal path can move from a vertex $L_h$ to a vertex $U_r$ in only three ways:
            \begin{enumerate}
                \item through a threshold vertex, along $(L_h,p_{h,t},U_h)$; this keeps the same dimension, so $r=h$;
                \item through a portal vertex, along $(L_h,x_{h,r},U_r)$; by construction, $r\neq h$; or
                \item through an orthant-vertex $z_{O''}$, along $(L_h,z_{O''},U_r)$, waiting between its edge in $T_2$ and its edge in $T_4$. Here $h=\mathsf{ldim}(O'')$ and $r=\mathsf{udim}(O'')$, so again $r\neq h$.
            \end{enumerate}
            Such a travel from $L_h$ to $U_r$ can occur at most once within a temporal path, since after the path reaches a vertex $U_r$, all edges that could take it to an $L$-vertex have already appeared: the relevant edges at the $L$-vertices lie in $T_2$ or in the first subblock of $T_3$, whereas a portal vertex reaches $U_r$ only in the second subblock of $T_3$, and the other two transitions reach $U_r$ in $T_4$.

            We analyze how $z_O$ reaches $z_{O'}$. Recall that any orthant-vertex $z_P$ has exactly two incident temporal edges: $z_PL_{\mathsf{ldim}(P)}$ in time block $T_2$ and $z_PU_{\mathsf{udim}(P)}$ in time block $T_4$. 
            Since $\mathsf{udim}(O)\neq i=\mathsf{ldim}(O)$, the $z_O$-$z_{O'}$-path cannot start with the $T_4$ edge $z_OU_{\mathsf{udim}(O)}$ since after taking it, the path cannot return to an $L$-vertex and hence can reach neither the edge $L_{\mathsf{ldim(O')}}z_{O'}$ nor $z_{O'}U_i$. Thus, it must start with $z_OL_i$ in $T_2$. Similarly, since $\mathsf{ldim}(O')\neq i=\mathsf{ldim}(O)$, the path cannot reach $z_{O'}$ through its $T_2$ edge and must end with $U_iz_{O'}$ in $T_4$.
            Consequently, the path must make its unique transition from $L_i$ to $U_i$. A portal vertex or an intermediate orthant-vertex changes the dimension, so the path must use a threshold vertex $p_{i,t}$. The order in $T_2$ forces $\mathsf{lval}(O)\leq t$, and the order in $T_4$ forces $t\leq\mathsf{uval}(O')$. Thus, $\mathsf{lval}(O)\leq\mathsf{uval}(O')$, proving~($\ast$).

            Applying~($\ast$) in both directions shows that $z_O$ and $z_{O'}$ are {mutually reachable} exactly when
            \begin{align*}
                &\mathsf{ldim}(O)\neq\mathsf{udim}(O')
                    \ \text{or}\
                    \mathsf{lval}(O)\leq\mathsf{uval}(O'), \text{ and}\\
                &\mathsf{ldim}(O')\neq\mathsf{udim}(O)
                    \ \text{or}\
                    \mathsf{lval}(O')\leq\mathsf{uval}(O).
            \end{align*}
            This precisely rules out the two possible lower--upper conflicts between $O$ and $O'$. Hence, they hold exactly when $O\cap O'\neq\emptyset$.
        \end{proof}

        By the intersection property, a pairwise-intersecting family of mixed 2-sided orthants has a common point.

        \begin{proof}[Proof of \Cref{thm:vc-whard}]
            Suppose that a point $q\in\mathbb R^d$ belongs to a subfamily $\mathcal F\subseteq\mathcal O$ of at least $\ell$ orthants, and let $Z_{\mathcal F}:=\{z_O:O\in\mathcal F\}$. By \Cref{lem:vc-orthant-reachability}, {any two vertices of $Z_{\mathcal F}$ are mutually reachable} via paths whose internal vertices belong to $A$. By \Cref{lem:vc-orthant-auxiliary}, all pairs involving a vertex of $A$ have witnessing paths in $A\cup Z_{\mathcal F}$. Consequently, $A\cup Z_{\mathcal F}$ is closed temporally connected. It is therefore contained in a closed tcc and, independently, in an open tcc, each of size at least
            $|A|+|\mathcal F|\geq k'$.

            Conversely, let $X$ be an open or closed tcc of size at least $k'$. Since there are only $|A|$ auxiliary vertices, $X$ contains at least $k'-|A|=\ell$ orthant-vertices. {Any two of these vertices are mutually reachable} in $\mathcal G$, so their orthants are pairwise intersecting by \Cref{lem:vc-orthant-reachability}. They therefore have a common point, and the \textup{\textsc{M2SOD}} instance is a yes-instance.

            The construction is polynomial, and the vertex cover number of its footprint is at most $2d+4$. It is thus a parameterized reduction. Since every static edge has one unique label, $\mathcal G$ is simple and proper, and strict and non-strict reachability coincide. This proves the undirected statement and completes the proof.
        \end{proof}

        \subsubsection{\openTCC parameterized by temporal path and walk number}\label{sec:open-tpn-Wh}
            We reduce \textup{\textsc{M2SOD}} to \openTCC in directed and undirected temporal graphs parameterized by temporal path number utilizing the ordering that each base-path provides.
                    We first prove the claim for directed graphs and then obtain a simple undirected construction by subdividing certain upper-path edges and forgetting the edge orientations.

        \bigparagraph{Intuition.}
        Each orthant $O$ is represented by an \textit{orthant-vertex} $z_O$, and the construction will ensure that distinct orthants are mutually reachable if and only if they intersect.
        For every dimension $i\in[d]$, we introduce one base-path $U_i$ for the upper values and one base-path $L_i$ for the lower values.
        For every value $\alpha\in[n]$, we introduce a \textit{switch vertex} $s_{i,\alpha}$ that lies on both $L_i$ and $U_i$; these switch vertices occur in the order of~$[n]$.
        We place a private \textit{wall vertex} immediately after each switch on $L_i$ and immediately before it on $U_i$. These walls prevent repeated edges between switches and will also block unwanted backward traversals in the undirected construction.
        The orthant-vertex $z_O$ is then placed on exactly two paths, namely $L_{\mathsf{ldim}(O)}$ and $U_{\mathsf{udim}(O)}$. On the lower path, it is placed before the switch corresponding to $\mathsf{lval}(O)$, while on the upper path, it is placed after the switch corresponding to $\mathsf{uval}(O)$.
        This allows a temporal path to switch from the lower path to the upper path within the same dimension exactly when the lower value is at most the upper value.
        To connect orthant-vertices whose relevant bounds lie in different dimensions $i\neq j$, we add an additional \textit{portal vertex} after the last orthant-vertex on $L_i$ and before the first orthant-vertex on $U_j$. These portal vertices provide reachability without comparing values.
        Lastly, to ensure that the orthant-vertices are preferred in a connected component, we replace each by sufficiently many consecutive copies and place them in opposite orders on their lower and upper paths, making the copies of each orthant-vertex temporally connected. Thus, every represented orthant contributes sufficiently many vertices to the connected component.

        This construction only works for open components. We cannot make it work for closed components as we did for vertex cover, because here we cannot simply add auxiliary vertices that connect everything.

        \openPathHardness*
        \begin{proof}[Proof of the directed case]
            We give a parameterized reduction from \textup{\textsc{M2SOD}} to \openTCC. Let $(\mathcal O,\ell)$ be an instance of \textup{\textsc{M2SOD}} in dimension $d$. We first construct a directed temporal graph $\gcal$ with $\tpn(\gcal)\leq 2d$ such that there is a point in $\mathbb R^d$ contained in at least $\ell$ orthants if and only if $\gcal$ has an open tcc of size at least $\ell M$ with $M$ defined below.

            For every orthant $O\in\mathcal O$, let $\mathsf{ldim}(O)$ and $\mathsf{lval}(O)$ denote its lower-bounded dimension and value, and $\mathsf{udim}(O)$ and $\mathsf{uval}(O)$ its upper-bounded dimension and value.
            Let $n\leq 2|\mathcal O|$ be the number of distinct bound values. Replacing every bound value by its rank preserves all lower--upper comparisons, so we may assume that $\mathsf{lval}(O),\mathsf{uval}(O)\in[n]$ for every $O\in\mathcal O$. 
            
            The vertex set $V=Z\cup P\cup S\cup W$ of \gcal is composed of four types of vertices.\begin{itemize}
                \item For every pair $i,j\in[d], i\neq j$, introduce a \emph{portal vertex} $p_{i\rightarrow j}$ and let $P:=\{p_{i\rightarrow j} \colon i,j\in[d], i\neq j\}$ denote the set of all portal vertices. For $i\in[d]$, set
                    $P_{i\rightarrow}:=\{p_{i\rightarrow j}\colon j\in[d]\setminus\{i\}\}$ and
                    $P_{\rightarrow i}:=\{p_{j\rightarrow i}\colon j\in[d]\setminus\{i\}\}$.
                \item For every $i\in[d]$ and $\alpha\in[n]$, introduce a \emph{switch vertex} $s_{i,\alpha}$, and let $S:=\{s_{i,\alpha}\colon i\in[d],\ \alpha\in[n]\}$ denote the set of all switch vertices.
                \item For every $i\in[d]$ and $\alpha\in[n]$, introduce two \emph{wall vertices} $w^L_{i,\alpha}$ and $w^U_{i,\alpha}$, and let
                    $W:=\{w^L_{i,\alpha},w^U_{i,\alpha}\colon i\in[d],\ \alpha\in[n]\}$.
                \item Set $M:=\lvert P\cup S\cup W\rvert+2$.
                For every orthant $O\in\mathcal O$, introduce $M$ copies of an \emph{orthant-vertex} $\{z_O^{(\gamma)}\}_{\gamma\in [M]}$, and let $Z:=\{z_O^{(\gamma)}\colon O\in\mathcal O, \gamma\in[M]\}$ denote the set of all orthant-vertices.
                For $i\in[d]$ and $\alpha\in[n]$, we denote by $Z^L_{i,\alpha}:=\{z_O^{(\gamma)}\colon \mathsf{ldim}(O)=i,\ \mathsf{lval}(O)=\alpha, \gamma\in[M]\}$ and $Z^U_{i,\alpha}:=\{z_O^{(\gamma)}\colon \mathsf{udim}(O)=i,\ \mathsf{uval}(O)=\alpha, \gamma\in[M]\}$ the sets of orthant-vertices with lower, respectively upper, dimension $i$ and value $\alpha$.
            \end{itemize}
            \begin{figure}[t]
                \centering
                \includegraphics[width=\linewidth]{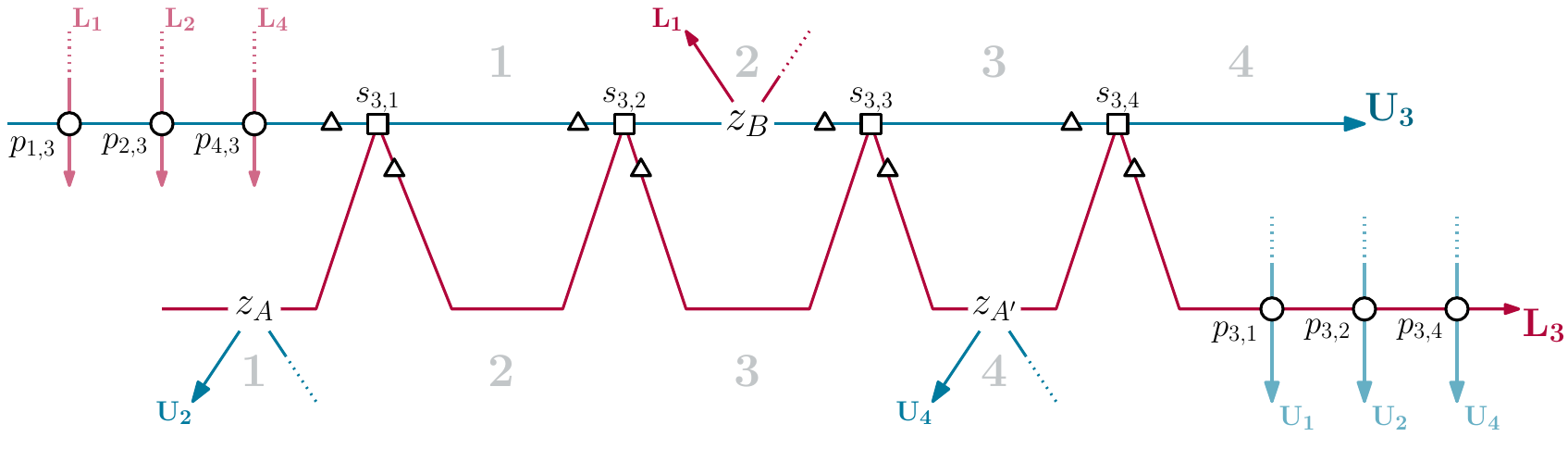}
                \caption{Illustration of the construction of temporal base-paths $L_3,U_3$ in an instance with $d=4$; $\mathsf{ldim}(A)=3$, $\mathsf{lval}(A)=1$, $\mathsf{udim}(A)=2$;
                    $\mathsf{ldim}(B)=1$, $\mathsf{udim}(B)=3$, $\mathsf{uval}(B)=2$; and 
                    $\mathsf{ldim}(A')=3$, $\mathsf{lval}(A')=4$, $\mathsf{udim}(A')=4$.
                The switch-vertices are squares, the portal-vertices are circles, and the wall-vertices are triangles (unlabeled for readability). 
                Since $t(L_i)<t(U_j)$, $z_A\rightreach z_B$ but $z_{A'}\not\rightreach z_B$ via $L_3,U_3$. This represents that in dimension 3, $A$ and $B$ have a common coordinate, while $A'$ and $B$ do not.
                The orthant-vertices are illustrated by one representative instead of the $M$ copies. For an illustration of this gadget refer to \Cref{fig:orthants4}.}
                \label{fig:orthants3}
            \end{figure}
            
            For every $i\in[d]$, we construct two temporal base-paths $L_i$ and $U_i$ visiting the following vertices in order
            \begin{align*}
                L_i:=&\bigl(Z^L_{i,1},s_{i,1},w^L_{i,1},Z^L_{i,2},s_{i,2},w^L_{i,2},\ldots,Z^L_{i,n},s_{i,n},w^L_{i,n},P_{i\rightarrow}\bigr),
                \\
                U_i:=&\bigl(P_{\rightarrow i},w^U_{i,1},s_{i,1},Z^U_{i,1},w^U_{i,2},s_{i,2},Z^U_{i,2},\ldots,w^U_{i,n},s_{i,n},Z^U_{i,n}\bigr),
            \end{align*}
            where each set of $M$ copies is consecutive. More precisely, fix an arbitrary ordering of the indexed family $\mathcal O$. Within every $Z^L_{i,\alpha}$, the corresponding copy sets occur in the induced order, and the copies of each orthant occur in increasing order of $\gamma$. Within every $Z^U_{i,\alpha}$, both orders are reversed. Refer to \Cref{fig:orthants3} for an illustration of the whole construction and to \Cref{fig:orthants4} for an illustration specifically of the $M$ copies of an orthant-vertex.
            \begin{figure}[t]
                \centering
                \includegraphics[width=0.75\linewidth]{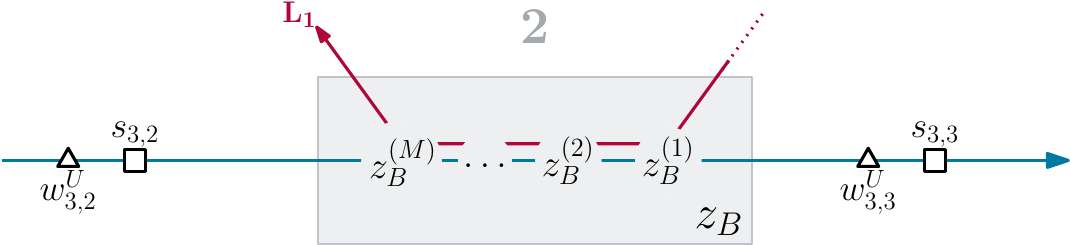}
                \caption{Illustration of the $M$ copies of the orthant-vertex $z_B$ from \Cref{fig:orthants3} and its appearance on its two base-paths $L_1,U_3$.}
                \label{fig:orthants4}
            \end{figure}

            We orient every edge in the displayed direction. The time labels strictly increase along each base-path, and every label is globally unique.
            The order between the base-paths is arbitrary, except that every edge of a lower path $L_i$ has a smaller label than every edge of an upper path $U_j$.

            This completes the construction, which is computable in polynomial time. The $2d$ displayed paths form an exact edge-cover, and hence $\tpn(\gcal)\leq 2d$. The private wall vertices prevent two shared switch vertices from being consecutive on both $L_i$ and $U_i$. Moreover, whenever two orthant-vertices are consecutive on both a lower and an upper path, their orders on these paths are opposite. Thus, no directed static edge is repeated. Consequently, \gcal is simple, and the globally unique labels make it proper. In the following let $A,B\in \mathcal O$.

            \smallskip
            First, we observe that the $M$ copies of an orthant-vertex $z_A$ (1) can reach one another and (2) have exactly the same inreachabilities and outreachabilities. Since the $M$ copies occur consecutively in increasing order on $L_{\mathsf{ldim}(A)}$ and in decreasing order on $U_{\mathsf{udim}(A)}$, they form a temporally connected set, so (1) holds. Since no other vertex occurs between these copies on either base-path, (2) holds. Consequently, in the rest of the proof, we will consider the $M$ copies $\{z_A^{(\gamma)}\}_{\gamma\in[M]}$ as one vertex $z_A$ with weight $M$.

            Second, we observe that any temporal path in $\mathcal{G}$ has to be either a subpath of a single base-path or consist of a subpath of some $L_i$ followed by a subpath of some $U_j$ for $i,j\in[d]$. That is, because for every $i,j\in[d], i\neq j$, we have $L_i\cap L_j=\emptyset$, $U_i\cap U_j=\emptyset$, and the largest time label on $L_i$ is smaller than the smallest time label on $U_j$.
            
            Third, we show that
            \[
                z_A\rightreach z_B
                \quad\Longleftrightarrow\quad
                \mathsf{ldim}(A)\neq\mathsf{udim}(B)
                \ \text{or}\ 
                \mathsf{lval}(A)\leq\mathsf{uval}(B).
                \tag{$\ast$}
            \]
            Let $i:=\mathsf{ldim}(A)$, $j:=\mathsf{udim}(B)$, $\alpha:=\mathsf{lval}(A)$, and $\beta:=\mathsf{uval}(B)$. If $i\neq j$, then $z_A$ follows $L_i$ to the portal $p_{i\rightarrow j}$ and subsequently follows $U_j$ to $z_B$. If $i=j$ and $\alpha\leq\beta$, then $z_A\in Z^L_{i,\alpha}$ follows $L_i$ to $s_{i,\beta}$ and subsequently follows $U_i$ to $z_B\in Z^U_{i,\beta}$ (e.g., see \Cref{fig:orthants3} with $i=3$).
            Lastly, assume $i=j$ and $\alpha>\beta$. A path from $z_A$ that starts along $L_i$ must leave $L_i$, since $z_B$ does not occur on $L_i$ (its lower dimension differs from its upper dimension~$i$). 
            Since a temporal path in this graph can only consist of at most two base-paths, it must therefore enter $U_i$. The only vertices shared by $L_i$ and $U_i$ are the switches $s_{i,r}$. Reaching such a switch on $L_i$ requires $r\geq\alpha$, whereas reaching $z_B$ from $s_{i,r}$ on $U_i$ requires $r\leq\beta$. No such $r$ exists. The only other base-path containing $z_A$ is $U_{\mathsf{udim}(A)}$. Since $\mathsf{udim}(A)\neq\mathsf{ldim}(A)=\mathsf{udim}(B)$, this upper path is distinct from $U_i$, and by the observation above, a temporal path cannot move between two upper paths. Hence $z_A\not\rightreach z_B$, which proves~($\ast$).

            By definition of mixed 2-sided orthants, $A$ and $B$ intersect if and only if $\mathsf{ldim}(A)\neq \mathsf{udim}(B)$ or $\mathsf{lval(A)}\leq\mathsf{uval(B)}$, and $\mathsf{udim}(A)\neq \mathsf{ldim}(B)$ or $\mathsf{uval(A)}\geq\mathsf{lval(B)}$.
            By~($\ast$), these two conditions are equivalent to $z_A\rightreach z_B$ and $z_B\rightreach z_A$, respectively. Therefore, we have $z_A\leftrightsquigarrow z_B$ if and only if $A\cap B\neq\emptyset$.

            It remains to prove that there is a point contained in at least $\ell$ orthants if and only if $\gcal$ has an \otcc of size at least $\ell M$.

            \bigparagraph{$(\Rightarrow)$}\quad
            Assume there is a point $z\in\mathbb R^d$ contained in at least $\ell$ orthants, and let $\mathcal O_z\subseteq\mathcal O$ be this family of orthants. Consider the set $X:=\{z_O^{(\gamma)}\colon O\in\mathcal O_z,\ \gamma\in[M]\}$ of all corresponding orthant-vertices. Since every pair of orthants in $\mathcal O_z$ intersects, {any two vertices of $X$ are mutually reachable}. Thus, $X$ is temporally connected and is contained in an \otcc. Since $|X|=|\mathcal O_z|M\geq\ell M$, \gcal has an \otcc of size at least~$\ell M$.

            \bigparagraph{$(\Leftarrow)$}\quad
            Assume there is an \otcc $X$ of size at least $\ell M$, and let $\mathcal O_X:=\{O\in\mathcal O\colon X\cap\{z_O^{(\gamma)}\colon \gamma\in[M]\}\neq\emptyset\}$ be the orthants represented in $X$.

            If $|\mathcal O_X|\leq\ell-1$, then $X$ contains at most $(\ell-1)M$ orthant-vertices. There are only $|P\cup S\cup W|=M-2$ other vertices. Hence $|X|\leq (\ell-1)M+(M-2)=\ell M-2$, which is a contradiction. Thus, $|\mathcal O_X|\geq\ell$.
            For every $A\in\mathcal O_X$, choose a copy $z_A^{(\gamma_A)}\in X$. Since $X$ is temporally connected, {any two of these copies are mutually reachable}. By the observation above, every two orthants in $\mathcal O_X$ must intersect, and thus the pairwise-intersecting family $\mathcal O_X$ has a common point. Hence there exists a point contained in at least $\ell$ orthants of~$\mathcal O$.
            The construction is polynomial, and $\tpn(\gcal)\leq2d$. Thus, this is a parameterized reduction to simple proper directed temporal graphs, proving the directed case for temporal path number. Since $\twn(\gcal)\leq\tpn(\gcal)\leq2d$, the same reduction establishes the temporal walk-number statement.
        \end{proof}

        When moving to undirected graphs, the general idea is to keep the same base-paths and forget the edge orientations. With undirected edges, since the time labels increase along each base-path, a temporal path can take at most one step backwards before it has to stop or switch to an upper path. However, the wall vertices ensure that such a step cannot bypass the comparison of the lower and upper values. Thus, the orthant-vertices still reach each other under exactly the same conditions.

        Unfortunately, this can result in a multilabeled graph when two orthant-vertices are consecutive on both a lower and an upper path, since forgetting the orientations gives the same undirected edge two distinct time labels.
        To remedy this and make the construction simple, we subdivide\footnote{This is very similar to the semaphore construction in \cite{casteigts_SimpleStrict_2024} and \cite{doring_SimpleStrict_2025}. Unfortunately, we cannot use the established hierarchy on induced-reachability equivalence for this result (by which we would only have to show hardness for directed strict simple temporal graphs), because the general proof strategy used to show the hierarchy does not preserve the \tpn parameter.} every edge between two orthant-vertices on an upper path. 
        \begin{lemma}\label{lem:undirected-open-tpn-hardness}
            \openTCC on simple proper undirected temporal graphs is \Wone-hard when parameterized by temporal path number, and consequently temporal walk number.
        \end{lemma}
        \begin{proof}
            Let $(\mathcal O,\ell)$ be an instance of \textup{\textsc{M2SOD}} in dimension $d$, and let $\gcal$ be the directed temporal graph constructed in the preceding proof, with $M=|P\cup S\cup W|+2$. We retain the target size $\ell M$.

            We obtain an undirected temporal graph $\gcal^\star$ as follows. First, subdivide every upper-path edge whose endpoints are both orthant-vertices with a private vertex, and let $R$ denote the set of these subdivision vertices. Then forget all edge orientations. More precisely, an edge $(u,v)$ with integer label $t$ that is subdivided is replaced by $ur,rv$ with labels $2t,2t+1$, respectively, where $u$ precedes $v$ on the upper path. Every unsubdivided edge with label $t$ receives label $2t$. This preserves the order of all original edge occurrences. The same $2d$ base-paths, with these subdivisions, still give $\tpn(\gcal^\star)\leq 2d$. Every repeated undirected edge of the original construction had two orthant-vertices as endpoints and occurred once on a lower and once on an upper path. Subdividing the upper occurrence therefore makes $\gcal^\star$ simple, and the globally unique labels make it proper.

            We show that $\gcal$ and $\gcal^\star$ have the same reachabilities between orthant-vertices. One direction is immediate: every directed temporal path in $\gcal$ gives a temporal path in $\gcal^\star$ by traversing the corresponding subdivisions in order.

            For the converse, consider copies of orthant-vertices $z_A$ and $z_B$ and let $i:=\mathsf{ldim}(A)$, $j:=\mathsf{udim}(B)$, $\alpha:=\mathsf{lval}(A)$, and $\beta:=\mathsf{uval}(B)$. It suffices to consider the case $i=j$ as otherwise, $z_A\rightreach z_B$ already holds in $\gcal$ by~($\ast$). All lower-path edges precede all upper-path edges, and distinct lower paths are vertex-disjoint, as are distinct upper paths. Thus, also in $\gcal^\star$, a temporal path can use at most one lower path followed by at most one upper path.

            The upper path containing $z_A$ is $U_{\mathsf{udim}(A)}$, which is distinct from $U_i$ because $\mathsf{udim}(A)\neq\mathsf{ldim}(A)=i$. Thus, a path from $z_A$ to $z_B\in U_i$ cannot start on an upper path and must start on~$L_i$.

            Suppose that a temporal path from $z_A$ to $z_B$ starts by traversing $L_i$ against its displayed order. Since the labels increase along $L_i$, this can only be its first lower-path edge. It either ends at a wall vertex, where it cannot continue, or at an orthant-vertex in the same block $Z^L_{i,\alpha}$ as $z_A$. In the latter case, the path can only continue on the upper path of this intermediate orthant, whose upper dimension differs from~$i$, and hence it cannot reach $z_B\in U_i$. The path must therefore follow $L_i$ in its displayed order.

            To reach $z_B$, the path must eventually enter $U_i$. The only common vertices of $L_i$ and $U_i$ are the switches. If it enters $U_i$ at $s_{i,r}$ against the displayed order, its next vertex is the private wall vertex $w^U_{i,r}$, where the path cannot continue. Hence it must follow $U_i$ in its displayed order as well. As in the directed construction, reaching $s_{i,r}$ from $z_A$ requires $r\geq\alpha$, and reaching $z_B$ from $s_{i,r}$ requires $r\leq\beta$. Therefore, $z_A\rightreach z_B$ in $\gcal^\star$ implies $\alpha\leq\beta$, and~($\ast$) holds for $\gcal^\star$ as well.

            Consequently, copies of two orthants are mutually reachable in $\gcal^\star$ if and only if the orthants intersect, and all copies of the same orthant remain mutually reachable.

            It remains to check that the subdivision vertices do not affect the size argument. Let $X$ be an open tcc of $\gcal^\star$. $X$ contains at most one vertex of $R$. Indeed, a temporal path starting at a subdivision vertex can only use its upper path, since all lower-path edges have already appeared and distinct upper paths are vertex-disjoint. On this upper path, a backward traversal can use only one edge. Since two subdivision vertices are never adjacent, the later one cannot reach the earlier one. Thus, no two distinct vertices of $R$ are mutually reachable.

            If a point belongs to at least $\ell$ orthants, the $M$ copies of each of these orthants together form a temporally connected set of size at least $\ell M$ in $\gcal^\star$, which extends to an open tcc.
            For the opposite argument, assume $|X|\geq\ell M$. Apart from orthant-vertices, $X$ contains at most $|P\cup S\cup W|+1=M-1$ vertices. Thus, if $X$ represented at most $\ell-1$ orthants, then $|X|\leq(\ell-1)M+(M-1)=\ell M-1$, a contradiction. It therefore represents at least $\ell$ orthants. These orthants intersect pairwise by the reachability equivalence, and hence have a common point.

            The subdividing-transformation is polynomial, and $\tpn(\gcal^\star)\leq2d$. This gives a parameterized reduction to simple proper undirected temporal graphs and establishes the undirected case of \Cref{thm:openTCC-wone-hard} for temporal path number. Since $\twn(\gcal^\star)\leq\tpn(\gcal^\star)\leq2d$, the same reduction establishes the temporal walk-number statement.
        \end{proof}

    \subsection{Hardness of \closedTCC for Constant Temporal Path and Walk Number}\label{sec:closed-tpn-pNPh}
                We present a parameterized construction from $k$-\clique to $k$-\closedTCC on a temporal graph constructed of 6 temporal paths.
        Before diving into the proof, we recall the classical \Wone-hardness reduction for \openTCC (not closed!) both parameterized by solution size \cite{bhadra_ComplexityConnected_2003,casteigts_SimpleStrict_2024}, as our reduction to \closedTCC builds on it directly.  
        
        The key idea is to encode the edges of the \clique instance~$H$ as temporal {mutual reachability} via the \emph{semaphore construction}~\cite{casteigts_SimpleStrict_2024,doring_SimpleStrict_2025}:
        \textit{(i)} replaces every static undirected edge $uv$ of $H$ with two directed temporal paths of length two, namely $(u,x_{uv},v)$ and $(v,x_{vu},u)$;
        \textit{(ii)}~label the edges $(u,x_{uv})$ and $(v,x_{vu})$ with time~1, and the edges $(x_{uv},v)$ and $(x_{vu},u)$ with time~2;
        \textit{(iii)}~to obtain unique labels, impose an arbitrary order on the time--$1$ edges and shift their labels accordingly, and do the same for the time--$2$ edges (all strictly after the last time--$1$ label).
        See \Cref{fig:hardness opentcc} (right) for an illustration of the resulting temporal graph.
        This labeling ensures that every temporal path has length at most two, which implies that a vertex can reach another if and only if they are adjacent in $H$. Consequently, a clique of size $k$ in $H$ corresponds exactly to an \otcc of size $k$ in the constructed temporal graph. This yields the desired reduction for $k$-\openTCC.
        \begin{figure}[ht]
            \centering
            \includegraphics[width=0.75\linewidth]{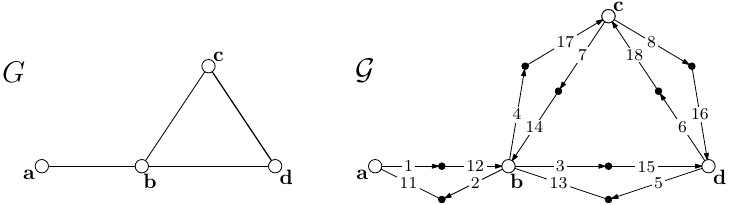}
            \caption{Illustration of the classical \Wone-hardness reduction for \openTCC.}
            \label{fig:hardness opentcc}
        \end{figure}
        
        We will extend this construction to obtain \NP-hardness for \closedTCC on \kpathgraphs{6}.  
        There are two obstacles to using this construction for a \closedTCC hardness result when parameterizing by \tpn.
        First, the subdivision vertices cannot belong to any closed component of size larger than one as they only reach their direct neighbor. Second, the temporal path number of the constructed graph is unbounded.
        So we must force the maximum component to contain  all semaphore subdivision vertices (and thus become closed), while using constantly many temporal base-paths (six will suffice).
        Crucially, we must ensure that along these base-paths no unwanted reachabilities are created. To control this, we insert special pairs of vertices that can never belong to a non-trivial \ctcc; we call them \emph{bridges}.
        Intuitively, a bridge creates a temporal ``bottleneck'': it can be crossed in one direction only, so any set containing a bridge-vertex cannot be temporally connected unless it is of size one.
        \begin{definition}[Bridges]
            A \emph{bridge} is a pair of consecutive vertices on some temporal path that do not appear in any other path. We call these vertices \emph{bridge-vertices}. 
        \end{definition}
        \begin{figure}[ht]
            \centering
            \includegraphics[width=0.85\linewidth]{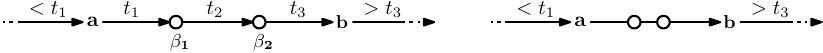}
            \caption{The left shows an example of a bridge $(\beta_1,\beta_2)$ containing the bridge-vertices $\beta_1$ and $\beta_2$ with $t_1<t_2<t_3$. The right shows how we depict  bridges in \Cref{fig:ctcc hardness big overview}.}
            \label{fig:bridge-vertices}
        \end{figure}
        By definition, bridge-vertices are too restricted in their temporal reachability to belong to any nontrivial \ctcc.
        \begin{observation}\label{lem:no-bridges-in-ctcc}
            For any bridge-vertex $\beta\in V$, the only \ctcc containing $\beta$ is the trivial component $\{\beta\}$.
        \end{observation}
        \begin{proof}
            Refer to the bridge illustrated in \Cref{fig:bridge-vertices}.  
            Any \ctcc containing $\beta_1$ and some vertex $x\neq \beta_1$ must also contain $\beta_2$, since every vertex reachable from $\beta_1$ is reached via $\beta_2$.  
            However, $\beta_2$ cannot reach $\beta_1$: The earliest time at which $\beta_2$ can reach any other vertex is after $t_3$, whereas the latest time at which any vertex can reach $\beta_1$ is $t_1 < t_3$.  
            Analogously, any nontrivial \ctcc containing $\beta_2$ would also have to contain $\beta_1$, which is impossible.  
            Hence, no nontrivial \ctcc can contain a bridge-vertex.
        \end{proof}
        From this, it follows that bridges can be removed from the temporal graph without changing the maximum \ctcc of a \kpathgraph{\paraPaths}.  
        \begin{lemma}\label{lem:bridges-irrelevant-for-ctcc}
            Let $\gcal$ be a \kpathgraph{\paraPaths} and let $\gcal'$ be the temporal graph obtained from $\gcal$ by deleting all bridges. 
            A set $C\subseteq V$ with $\lvert C\rvert>1$ is a \ctcc in $\gcal$ if and only if it is a \ctcc in $\gcal'$.
        \end{lemma}
        \begin{proof}
            Let $C$ with $\lvert C\rvert>1$ be a \ctcc in $\gcal$.  
            By \Cref{lem:no-bridges-in-ctcc}, $C$ contains no bridge-vertices. By the definition of a closed component, for every~$u,v\in C$ there exists a temporal $u$-$v$-path entirely within~$C$. Since $C$ contains no bridge-vertices, this $u$-$v$-path is also contained in~$\gcal'$. Hence, $C$ is also a \ctcc in~$\gcal'$.
            For the opposite direction observe that a \ctcc $C'$ in $\gcal'$ is a \ctcc in $\gcal$, because $\gcal'$ is a subgraph of $\gcal$.
         \end{proof}
         With these preliminaries in place, we can now present our construction.

        \closedPathHardness*
        \begin{proof}
            Let $(H=(V_H,E_H),\paraCliqueSize)$ be an instance of \clique. The task is to decide whether there exists a set $V'\subseteq V_H$ with $|V'|=\paraCliqueSize$ such that $uv\in E_H$ for all distinct $u,v\in V'$.
            
            We will construct a temporal graph $\gcal$ out of 6 temporal base-paths such that $\gcal$ contains a \ctcc of size~\mbox{$\paraCliqueSize\cdot (\lvert V_H\rvert-1) + 2\lvert E_H\rvert$} if and only if $H$ contains a clique of size~$s$.
            We first describe the construction using 10 temporal base-paths, prove its correctness, and then show how to merge the 10 paths into 6.
            \bigparagraph{Construction. (see \Cref{fig:ctcc hardness big overview} for an illustration)}
            \label{constr: closed hardness}
            We describe the construction base-path by base-path.
            Rather than giving explicit time labels for the edges of the base-paths, we impose an \textit{implicit temporal order}: For any two paths $P_i$ and $P_j$ with $i < j$, all edges of $P_i$ occur
            strictly earlier than all edges of $P_j$ ; within each path, edge labels increase strictly from start to end. This guarantees a globally consistent temporal labeling.
            \begin{enumerate}
                \item \textbf{Sub-vertices.}\quad 
                    For each vertex $v_i \in V_H$, create a \emph{vertex-gadget} in $\gcal$ consisting of $n-1$ sub-vertices $v_i^j$.
                    Let $V^{sub} = \{v_i^j \colon i,j \in [n],\, i\neq j\}$ denote the set of all such sub-vertices.
                
                \item \textbf{Paths for {mutual reachability} of sub-vertices $P_1^V,P_2^V,P_9^V,P_{10}^V$.}\quad
                    Construct a temporal base-path $P_1^V$ traversing the sub-vertices $v_i^j\in V^{sub}$ in lexicographic order on $(i,j)$ (ordered first by $i$, then by $j$) and insert a bridge between the last sub-vertex of $v_i$ and the first sub-vertex of $v_{i+1}$.
                    Let $P_2^V$ be the reverse of $P_1^V$, and let $P_9^V$ and~$P_{10}^V$ be additional copies of $P_1^V$ and~$P_2^V$, respectively. Note that each base-path uses unique bridge-vertices.
                \item \textbf{Semaphore vertices.}\quad
                    For each edge $v_i v_j \in E_H$, create an \emph{edge-gadget} in $\gcal$ consisting of two bidirected, subdivided edges $(v_i^j, x_{ij}, v_j^i)$ and $(v_j^i, x_{ji}, v_i^j)$, where $v_i^j$ and $v_j^i$ are sub-vertices of the vertex-gadgets of $v_i$ and $v_j$, respectively. The subdivision vertices $x_{ij}$ and $x_{ji}$ are called \emph{semaphore vertices} (or \emph{sem-vertices} for short).
                    Let $S$ denote the set of all sem-vertices, so $\lvert S\rvert = 2\lvert E\rvert$ by construction.  
                \item \textbf{Paths for adjacency encoding $P_5^{out}$, $P_6^{in}$.}\quad
                    Collect every edge that leaves a sub-vertex to a sem-vertex into base-path $P_5^{out}$ by ordering the edges $(v_i^j, x_{ij})$ in lexicographic order of $(i,j)$ and inserting a bridge between each pair of consecutive edges.
                    Similarly, collect all edges $(x_{ij}, v_j^i)$ going from a sem-vertex to a vertex-gadget into base-path $P_6^{in}$.  \\
                    This construction is well defined: the semaphore technique of \cite{casteigts_FindingStructure_2018} guarantees a proper temporal labeling of these edge-gadgets, which implies a strict total order in which all edges used in $P_5^{out}$ occur earlier than those in $P_6^{in}$.
                \item \textbf{{Paths for mutual reachability of sem-vertices} $P_3^{sem},P_4^{sem},P_7^{sem},P_8^{sem}$.}\quad
                    Connect $S$ using two temporal base-paths $P_3^{sem}$ and $P_4^{sem}$ that traverse all sem-vertices $x_{ij}$ in lexicographic order, once forwards and once backwards.
                    Construct $P_7^{sem}$ and $P_8^{sem}$ as duplicates of $P_3^{sem}$ and~$P_4^{sem}$, respectively.
                \end{enumerate}
            In summary, the temporal graph is given by  
            \[\gcal = P_1^V \;\cup\; P_2^V \;\cup\; P_3^{sem} \;\cup\; P_4^{sem} \;\cup\; P_5^{out} \;\cup\; P_6^{in} \;\cup\; P_7^{sem} \;\cup\; P_8^{sem} \;\cup\; P_9^V \;\cup\; P_{10}^V.\quad \]
            
            We show that $H$ contains a clique of size $\paraCliqueSize$ if and only if $\gcal$ contains a \ctcc of size $\paraCliqueSize\cdot (\lvert V_H\rvert-1) + 2\lvert E_H\rvert$. Since \clique is \NP-hard and the reduction runs in polynomial time while producing a \kpathgraph{10} (later merged into a \kpathgraph{6}), the claim follows.
            Recall that $S$ is the set of sem-vertices.
            \begin{figure}[t]
                \centering
                \includegraphics[width=\linewidth]{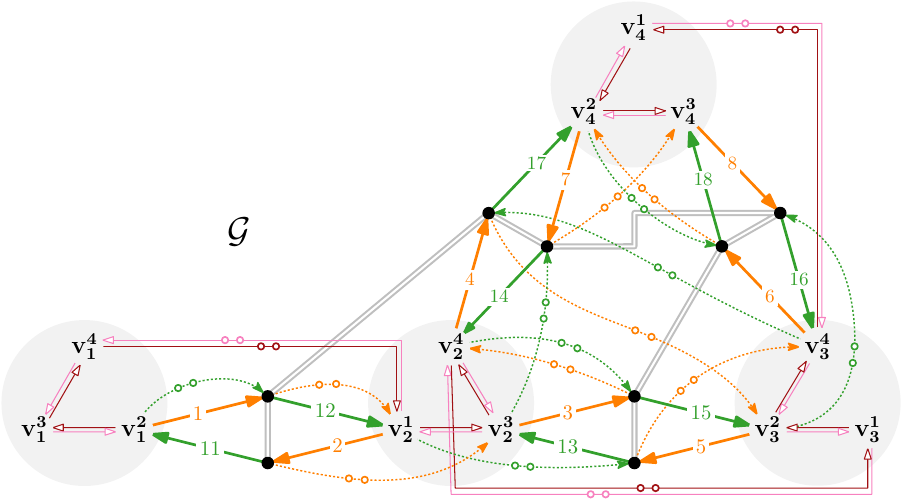}
                \caption{
                    Illustration of the construction for \Cref{thm:closed-tpn-paranp-hard}.
                    Bridges are indicated by pairs of circles, e.g., on the red and pink paths between vertex gadgets. The gadgets are highlighted by gray circular regions.  
                    The red and pink paths ($P_1^V$, $P_2^V$, $P_9^V$, $P_{10}^V$) enable {mutual reachability} within each vertex gadget while the bridges ensure this does not create {arbitrary mutual reachability between vertices of different gadgets}.
                    The gray paths ($P_3^{sem}$, $P_4^{sem}$, $P_7^{sem}$, $P_8^{sem}$) ensure that all sem-vertices (black dots) are {mutually reachable} with every vertex in \gcal.
                    The orange $P_5^{out}$ and green path $P_6^{in}$ encode the adjacency relation of the \clique instance.  
                    Orange temporal edges (labels 1–8) are connected via the dotted orange arcs (with bridges), while green temporal edges (labels 11–18) are connected via the dotted green arcs (with bridges).  
                    }
                \label{fig:ctcc hardness big overview}
            \end{figure}
            
            \bigparagraph{$(\Rightarrow)$}\quad 
            Let $C \subseteq V_H$ be a clique of size $\paraCliqueSize$ in $H$. We claim that $\mathcal{C} = S \cup \{v_i^j \in V^{sub} : v_i \in C\}$ is a \ctcc in $\gcal$. By construction, $\lvert\mathcal{C}\rvert = \paraCliqueSize\cdot (\lvert V_H\rvert-1) + 2\lvert E_H\rvert$.
            
            First, {any two sem-vertices from $S$ are mutually reachable} via the base-paths $P_3^{sem}$ and $P_4^{sem}$.
            Next, every sem-vertex can reach every sub-vertex from $V^{sub}$.
            Since $P_6^{in}$ occurs after both~$P_3^{sem}$ and $P_4^{sem}$, an arbitrary sem-vertex can first traverse the sem-paths to some sem-vertex~$x_{ij}$ incident with $v_i$, then enter the gadget of~$v_i$ via $P_6^{in}$, and finally reach all sub-vertices of $v_i$ through $P_9^V$ and $P_{10}^V$:  
            \[
            \text{sem-vertex} \;\to\; P_3^{sem}/P_4^{sem} \;\to\; P_6^{in} \;\to\; P_9^V/P_{10}^V \;\to\; \text{sub-vertex}.
            \]
            Conversely, every sub-vertex $v_i^j$ of some $v_i \in V_H$ can reach every sem-vertex.  
            Within its gadget it uses $P_1^V$ or $P_2^V$ to reach a designated sub-vertex $v_i^j$, then follows $P_5^{out}$ to a sem-vertex~$x_{ij}$, and from there $P_7^{sem}$ and $P_8^{sem}$ provide reachability to all others:  
            \[
            \text{sub-vertex} \;\to\; P_1^V/P_2^V \;\to\; P_5^{out} \;\to\; P_7^{sem}/P_8^{sem} \;\to\; \text{sem-vertex}.
            \]
            Finally, for any two distinct $v_i,v_j\in C$, the clique property ensures an edge $v_iv_j\in E_H$ in $H$. In $\gcal$ this guarantees that the sub-vertex $v_i^j$ of $v_i$ is connected to the sub-vertex $v_j^i$ of $v_j$ via a sem-vertex $x_{ij}$. 
            Hence, all sub-vertices of $v_i$ can reach all sub-vertices of $v_j$ via  
            \[
            \text{sub-vertex of } v_i \;\to\; P_1^V/P_2^V \;\to\; P_5^{out} \;\to\; x_{ij} \;\to\; P_6^{in} \;\to\; P_9^V/P_{10}^V \;\to\; \text{sub-vertex of } v_j,
            \]
            and symmetrically in the reverse direction. 
            All these paths use only vertices in~$\mathcal{C}$, so $\mathcal{C}$ is a valid \ctcc.
    
            \bigparagraph{$(\Leftarrow)$}\quad 
            Let $\mathcal{C}$ be a \ctcc of size $\paraCliqueSize\cdot (\lvert V_H\rvert-1) + 2\lvert E_H\rvert$ in $\gcal$.  
            Since $\lvert S\rvert = 2\lvert E_H\rvert$ and no bridge-vertex can be contained in $\mathcal{C}$ according to \Cref{lem:no-bridges-in-ctcc}, it follows that $\mathcal{C}$ contains at least~$\paraCliqueSize\cdot (\lvert V_H\rvert-1)$ many sub-vertices.
            Hence, $\lvert C = \{v_i \in V_H \colon v_i^j \in \mathcal{C} \text{ for some }j\}\rvert \geq \paraCliqueSize$ by the pigeonhole principle.  
            We claim that $C$ induces a clique in $G$.
            
            Consider the reduced temporal graph $\gcal'$ obtained by deleting all bridges from \gcal as per \Cref{lem:bridges-irrelevant-for-ctcc}.  
            In $\gcal'$ each base-path of type $P_i^V$ is split into segments confined to a single vertex-gadget, while $P_5^{out}$ and $P_6^{in}$ are split into isolated edges.
            
            Take any two distinct $v_i,v_j\in C$. Then there exists $x,y\in [n]$ with $v_i^{x} , v_j^{y} \in \mathcal{C}$.  
            Since $\mathcal{C}$ is a \ctcc, there must be temporal paths between $v_i^{x}$ and $v_j^{y}$ in $\gcal'$ using only $\mathcal{C}$.  
            In $\gcal'$, the only way to leave the vertex-gadget of $v_i$ is via an edge in $P_5^{out}$ to a sem-vertex, and the only way to enter the vertex-gadget of $v_j$ is via an edge in $P_6^{in}$ coming from a sem-vertex.
            Since $P_5^{out}$ and $P_6^{in}$ consist of isolated edges in $\gcal'$ (separated by bridges in \gcal), any temporal path can use at most one edge of each type.  
            Hence, the temporal path from $v_i^{x}$ to $v_j^{y}$ must traverse the sem-vertex $x_{ij}$, which exists if and only if $v_iv_j\in E$ by construction.
            Therefore, every pair in $C$ is adjacent and $C$ is a clique of size at least $\paraCliqueSize$.
    
            \bigparagraph{Merging ten base-paths into six.}\quad
            Observe that $P_1^V$ and $P_3^{sem}$ are vertex-disjoint and that their relative order is irrelevant for our arguments.  
            Hence, we can concatenate them two into a single temporal base-path, inserting a bridge between them to avoid unwanted reachabilities. We denote this concatenation by $P_1^V \circ^\star P_3^{sem}$.
            The same reasoning applies to $P_2^V$ with $P_4^{sem}$, $P_9^V$ with $P_7^{sem}$, and $P_{10}^V$ with $P_8^{sem}$.  
            Hence, \gcal can equivalently be constructed as
            \[
                \gcal = (P_1^V \circ^\star P_3^{sem}) \;\cup\; (P_2^V \circ^\star P_4^{sem}) \;\cup\; (P_5^{out} \;\cup\; P_6^{in}) \;\cup\; (P_7^{sem}\circ^\star P_9^V) \;\cup\; (P_8^{sem}\circ^\star P_{10}^V)\ifshort.\quad\qedhere\else,\fi
            \]
            which consists of a total of 6 temporal base-paths.
    
            \begin{remark}[Undirected and non-strict versions.]
            The reduction remains valid if all temporal edges are made undirected. To avoid confusion, we will refer to a temporal path through \gcal, which does not have to be one of the 6 base-paths of the construction, as a \textit{temporal trip}.
            
            Since the labeling of \gcal is proper, no two incident edges have the same time label. Thus, every temporal base-path of the construction keeps its temporal direction.
            
            Making the edges undirected does create additional local reachabilities (e.g., a sub-vertex can reach its incident sem-vertex along the “in” edge at its late time), but they do not create {any new pairs of mutually reachable vertices belonging to different gadgets}. In particular:
            \begin{enumerate}
                \item Within a vertex-gadget, {any two sub-vertices are mutually reachable}, while the base-paths $P_i^V$ cannot be used to move between different gadgets because of the bridges.
                \item Between vertex-gadgets, {mutual reachability} is still achieved only via $P_5^{out}$ and $P_6^{in}$, based on the adjacency in the \clique instance.
                \item The semaphore edges satisfy \emph{out before in}: All edges of $P_5^{out}$ occur strictly earlier than all edges of $P_6^{in}$.
                Hence, a temporal trip cannot enter a gadget and later leave to another gadget (that would require taking an “out” edge after an “in” edge).
                While a trip may reach a sem-vertex via an undirected “in” edge at a late time (e.g., $v_2^1$ taking the green edge at time step 12 in the wrong direction in \Cref{fig:ctcc hardness big overview}), any continuation to another gadget would either violate time order or require following $P_6^{in}$ across bridges, which is not allowed in a closed component.
            \end{enumerate}
            The bridge definition and implications (\Cref{lem:no-bridges-in-ctcc,lem:bridges-irrelevant-for-ctcc}) depend only on the increasing time labels $t_1<t_2<t_3$, not on the edge orientation, and are thus unchanged. Consequently, {the pairs of mutually reachable vertices, and thus the correctness of the reduction, coincide} in the directed and undirected versions of the construction.
    
            Finally, because the labeling is proper, the strict and non-strict interpretations of~$\gcal$ have the same reachabilities and are thus reachability-equivalent (cf.~\cite{casteigts_SimpleStrict_2024,doring_SimpleStrict_2025}).
            \end{remark}
            Since $\twn(\gcal)\leq\tpn(\gcal)=6$, the same construction also proves hardness for temporal walk number at most six.
        \end{proof}

                Since $\tpn \geq \tdegree$, the reduction above directly implies \paraNP-harness of \closedTCC parameterized by \tdegree. For \openTCC, the classical \clique reduction (see \Cref{fig:hardness opentcc}) can be adjusted to ensure bounded temporal total degree: Replace every vertex $v$ with a binary tree whose number of leaves equals the degree of $v$ in $H$, and enforce full pairwise reachability in this tree before the first and after the last time step of the semaphore edges. Then, in a maximum open tcc that binary tree is included instead of $v$. Combining both observations yields the following.
        \begin{corollary} \label{cor:open-closed-tempdegree}
            \openTCC and \closedTCC on (un)directed, (non-)strict temporal graphs are \NP-hard even on graphs with $\tdegree=6$ for \closedTCC and $\tdegree=4$ for \openTCC.
        \end{corollary}
        Additionally, since the temporal total degree is always larger than the static degree, we get:
        \begin{corollary} \label{cor:open-closed-staticdegree}
            \openTCC and \closedTCC on (un)directed, (non-)strict temporal graphs are \NP-hard even on graphs with $\Delta=6$ for \closedTCC and $\Delta=4$ for \openTCC.
        \end{corollary}

\section{Combining Footprint and Temporal Parameters}
\label{subsec: FPT for tw+deg or tw+lifetime}
    In this section we present our results for computing temporal connected components when parameterized by treewidth plus a temporal parameter: temporal total degree or lifetime.
    We present \FPT results for both $\tw+\tdegree$ and $\tw+\lifetime$ in \Cref{sec: FPT MSO} and show that they do not admit a polynomial kernel in \Cref{sec:kernel}.
        
    \subsection{\FPT via MSO Formulations}\label{sec: FPT MSO}
        We show that open and closed tccs can be encoded using \emph{monadic second-order logic} (MSO). This yields the existence of a fixed-parameter tractable algorithm when parameterized by the treewidth~$\tw$ combined either with the maximum temporal degree~$\tdegree$ or with the lifetime~$\lifetime$.  
    
        \combinedTemporalFPT*
        The bounds in \Cref{thm:combined-temporal-fpt} follow from \Cref{thm: fpt tw plus degree,thm: fpt tw plus lifetime,cor: FPT tw + tpn}.

                \bigparagraph{MSO on static graphs.}
        MSO is a logical formalism with two types of quantifiers ranging over individual elements and sets of such elements.  
        A classical \textit{static} graph $G=(V,E)$ can be encoded as a relational structure $(U, V, E, \adj, \inc)$ with universe $U=V\cup E$, where
        \begin{itemize}
            \item $V(\cdot)$ and $E(\cdot)$ are unary predicates identifying vertices and edges,
            \item $\adj(u,v)$ is a binary relation expressing adjacency of vertices,
            \item $\inc(v,e)$ is a binary relation expressing incidence between a vertex and an edge.
        \end{itemize}
        Formulas are built from atomic statements of the form $x=y$, $R(x,y)$, or $R'(x)$ (for $R\in\{\adj,\inc\}$ and $R'\in\{V,E\}$), combined by Boolean connectives $\neg,\vee,\wedge,\rightarrow,\leftrightarrow$ and quantifiers $\forall,\exists$ over elements or sets of elements. For details, see~\cite{courcelle_GraphStructure_2012}.  
        
        \bigparagraph{Algorithmic meta-theorem.}
        To connect MSO logic to parameterized complexity, we rely on Courcelle’s theorem, which states that every MSO-definable graph property can be decided efficiently on graphs of bounded treewidth. We use the following optimization variant.  
        \begin{theorem}[\cite{arnborg_Easyproblems_1991,courcelle_GraphStructure_2012}] \label{thm: courcelle optimization}
            There exists an algorithm that, given 
            \begin{enumerate}
                \item an MSO formula $\varphi$ with free monadic variables $X_1,\dots,X_r$,
                \item an affine function $\alpha(x_1,,\dots,x_r)$, and
                \item a graph $G$,
            \end{enumerate} 
            computes the minimum (or maximum) of $\alpha(\lvert X_1\rvert,\dots,\lvert X_r\rvert)$ over all evaluations of $X_1,\dots,X_r$ that satisfy $\varphi$ on $G$, in time $f(\lvert \varphi \rvert, \tw(G)) \cdot n$, where $f$ is a computable function.
        \end{theorem}
        Note that the runtime of the algorithm depends on the treewidth of the graph and the length of the formula.

        \bigparagraph{MSO on temporal graphs.}
        The first application of this framework to temporal graphs was given by Arnborg et al.~\cite{arnborg_Easyproblems_1991}, who encoded time labels as bit strings in an \textit{edge-labeled graph} and employed the classical MSO on static graphs. This was used by Zschoche et al.~\cite{zschoche_ComplexityFinding_2020} to show that separating two vertices $s$ and $z$ is FPT when parameterized by $\tw+\lifetime$.  
    
        In this work we follow a more general approach: A temporal graph is encoded as a relational structure with a universe (usually containing the vertices, static edges, temporal edges, and time steps), unary predicates for each set in the universe, and suitable relations such that the treewidth of the associated \emph{Gaifman graph} is bounded by the chosen parameter $\mathsf{para}$.
        The Gaifman graph of a relational structure $\mathcal{S}=(U, R_1,\dots,R_m)$ is the undirected graph with vertex set $U$ in which two distinct elements $u,v\in U$ are connected by an edge whenever there exists a relation $R_i$ and a tuple $(x,y)$ in $R_i$ with $u,v\in\{x,y\}$.
        We refer to such an encoding of \gcal as $\lfloor \gcal \rfloor_\mathsf{para}$. If the temporal property can be expressed in $\lfloor \gcal \rfloor_\mathsf{para}$ by a formula of length bounded in $\mathsf{para}$, then Courcelle’s theorem yields the existence of an \FPT algorithm.
        
        This approach was also used by Enright et al.~\cite{enright_DeletingEdges_2021} to obtain an FPT algorithm parameterized by $\tw+\tdegree$, and by Haag et al.~\cite{haag_FeedbackEdge_2022} to obtain an FPT algorithm parameterized by $\tw+\lifetime$. An overview of these developments is provided in the survey of temporal treewidth notions by Fluschnik et al.~\cite[Section 5]{fluschnik_timegoes_2020}.
        
        \newcommand{\MSOlifetime}{\ensuremath{\lfloor \gcal \rfloor_{\tw+\lifetime}}\xspace}
        \newcommand{\MSOtdegree}{\ensuremath{\lfloor \gcal \rfloor_{\tw+\tdegree}}\xspace}

        \subsubsection{MSO formulation under \texorpdfstring{$\tw+\tdegree$ and $\tw+\tpn$}{tw + temporal total degree and tw + tpn}}
        \label{subsec: MSO tw plus degree FPT}
                We first present the encoding for \openTCC/\closedTCC on undirected strict temporal graphs, and then explain the minor adjustments needed for directed or non-strict temporal graphs, without fully restating the definitions and proofs.
        \begin{definition}
            A relational structure $\MSOtdegree=(U,V,E,\ecal,\inc, \edgeTedge, \psuc)$ has universe $U=V\cup E \cup\ecal$, unary predicates $V(\cdot),E(\cdot),\ecal(\cdot)$ identifying vertices, static edges, and temporal edges, respectively, and binary relations  
            \begin{itemize}
                \item $\inc\subseteq E\times V$ where $\inc(e,v) \Leftrightarrow v\in e$,
                \item $\edgeTedge\subseteq\ecal\times E$ where $\edgeTedge((e,t),e')\Leftrightarrow e=e'$,
                \item $\psuc\subseteq \ecal\times\ecal$ where $\psuc((e_1,t_1),(e_2,t_2))\Leftrightarrow \big( e_1\cap e_2\neq\emptyset\text{ and } t_1<t_2\big)$.
            \end{itemize}  
        \end{definition}
        By~\cite[Lemma 5.3]{enright_DeletingEdges_2021}, the treewidth of the Gaifman graph is bounded by $\tw+\tdegree$.
        \begin{lemma}[\cite{enright_DeletingEdges_2021}]
            The treewidth of the Gaifman graph of the structure \MSOtdegree representing a temporal graph \gcal is bounded by $\tw+\tdegree$.
        \end{lemma}
        With this, we are ready to prove our theorem.
        \begin{restatable}{theorem}{twtdegFPT} \label{thm: fpt tw plus degree}
            \openTCC and \closedTCC on (un)directed, (non-)strict temporal graphs are in \FPT parameterized by $\tw+\tdegree$; can be solved in time $\mathcal{O}(f(\tw,\tdegree)\cdot n)$ for a computable function $f$.
        \end{restatable}
        \begin{proof}
            We define the optimization variant of \openTCC and \closedTCC as an MSO formula. 
            Let $\mathcal{I}=(\gcal,k)$ be an instance of \openTCC or \closedTCC interpreted as a relational structure \MSOtdegree.
            First we express incidence of a temporal edge with a vertex: 
            \begin{align*}
                \mathrm{inc_t}(\varepsilon,v) &:= \exists e \in E \big(\edgeTedge(\varepsilon,e)\wedge \inc(e,v)\big).
                \intertext{Next, we express for a temporal edge $\varepsilon$ the existence of one (or no) pre-/successor in a set $P$:}
                \deg_{\text{out}=1}(\varepsilon,P) &:= \exists \varepsilon'\in \ecal (\varepsilon'\in P \wedge \psuc(\varepsilon,\varepsilon')) \;\wedge \\
                    & \hspace{1.6em}\forall \varepsilon_1,\varepsilon_2 \Big( \big(\varepsilon_1,\varepsilon_2\in P \wedge \mathrm{pos\_suc}(\varepsilon,\varepsilon_1)\wedge \mathrm{pos\_suc}(\varepsilon,\varepsilon_2) \big)\to \varepsilon_1=\varepsilon_2\Big), \\
                \deg_{\text{in}=1}(\varepsilon,P) &:= \exists \varepsilon'(\varepsilon'\in P\wedge \mathrm{pos\_suc}(\varepsilon',\varepsilon)) \;\wedge \\
                    &\hspace{1.6em}\forall \varepsilon_1,\varepsilon_2\Big(\big(\varepsilon_1,\varepsilon_2\in P \wedge \mathrm{pos\_suc}(\varepsilon_1,\varepsilon)\wedge \mathrm{pos\_suc}(\varepsilon_2,\varepsilon) \big)\to \varepsilon_1=\varepsilon_2\Big), \\
                \deg_{\text{out}=0}(\varepsilon,P) &:= \neg\exists \varepsilon'(\varepsilon'\in P\wedge \mathrm{pos\_suc}(\varepsilon,\varepsilon')), \\
                \deg_{\text{in}=0}(\varepsilon,P) &:= \neg\exists \varepsilon'(\varepsilon'\in P\wedge \mathrm{pos\_suc}(\varepsilon',\varepsilon)).\end{align*}
                Now, we express open and closed temporal paths with a first and last temporal edge:
                \begin{align*}
                \pathPred(P,\varepsilon_s,\varepsilon_t) &:= (\varepsilon_s\in P \wedge \varepsilon_t\in P)\;\wedge
                    \deg_{\text{in}=0}(\varepsilon_s,P)\wedge \deg_{\text{out}=1}(\varepsilon_s,P)\;\wedge\\ 
                    &\hspace{1.6em}\deg_{\text{in}=1}(\varepsilon_t,P) \wedge \deg_{\text{out}=0}(\varepsilon_t,P)\;\wedge \\
                    &\hspace{1.6em} \forall \varepsilon \in P \Big( \big(\varepsilon\neq \varepsilon_s \wedge \varepsilon\neq \varepsilon_t \big) 
                    \to \deg_{\text{in}=1}(\varepsilon,P)\wedge \deg_{\text{out}=1}(\varepsilon,P)\Big),\\
                \pathPredX(P,\varepsilon_s,\varepsilon_t,X) &:= 
                        \pathPred(P,\varepsilon_s,\varepsilon_t) \;\wedge\; \forall \varepsilon \in P\; \forall v\in V \big( \mathrm{inc_t}(\varepsilon,v) \to v\in X\big).
                \intertext{Using these formulas, we can express temporal reachability between two vertices:}
                \mathrm{reach}(u,v) &:= \exists P\subseteq \ecal \;\exists \varepsilon_s,\varepsilon_t\in \ecal 
                    \big(\pathPred(P,\varepsilon_s,\varepsilon_t)\wedge
                    \mathrm{inc_t}(\varepsilon_s,u)\wedge \mathrm{inc_t}(\varepsilon_t,v)\big).\\
                \mathrm{reach_X}(u,v) &:= \exists P\subseteq \ecal \;\exists \varepsilon_s,\varepsilon_t\in \ecal 
                    \big(\pathPredX(P,\varepsilon_s,\varepsilon_t)\wedge
                    \mathrm{inc_t}(\varepsilon_s,u)\wedge \mathrm{inc_t}(\varepsilon_t,v)\big).
                \intertext{Finally, we express a set $X\subseteq V$  being an open or closed component as before:}
                \varphi_{open}(X)&:= X\subseteq V\wedge \forall u,v\in X \big(\mathrm{reach}(u,v)\wedge \mathrm{reach}(v,u) \big),\\
                \varphi_{closed}(X)&:= X\subseteq V\wedge \forall u,v\in X \big(\mathrm{reach_X}(u,v,X)\wedge \mathrm{reach_X}(v,u,X) \big).
        \end{align*}
            Using the affine goal function $\alpha(x)=x$, \Cref{thm: courcelle optimization} implies that the optimization variant of \openTCC and \closedTCC can be solved in time $f(\tw,\tdegree)\cdot n$ for a computable function~$f$.
            
            \bigparagraph{Adjustments for directed / non-strict variants.}\quad  We describe the adjustments necessary to formulate the component properties in directed or non-strict temporal graphs. They all keep the Gaifman treewidth bound of $\tw+\tdegree$ unchanged.
                
                \smallparagraph{Directed}
                In the encoding, replace $\inc$ with the two binary relations $\source,\target\subseteq E\times V$ where $\source((y,x),v)\Leftrightarrow x=v$ and $\target((y,x),v)\Leftrightarrow y=v$; and replace $\psuc$ with
                \[
                \psuc^{\to}((e_1,t_1),(e_2,t_2))\Leftrightarrow \exists v\in V\,\big(\target(e_1,v)\wedge\source(e_2,v)\big)\wedge t_1<t_2.
                \]
                In the formulas, replace the temporal incidence $\inc_t$ with
                \begin{align*}
                \inc_t^{-}(\varepsilon,u)&:=\exists e\in E\,\big(\edgeTedge(\varepsilon,e)\wedge\source(e,u)\big),\\
                \inc_t^{+}(\varepsilon,v)&:=\exists e\in E\,\big(\edgeTedge(\varepsilon,e)\wedge\target(e,v)\big);
                \intertext{and replace $\mathrm{reach}$ (and $\mathrm{reach_X}$ analogously) with}
                \mathrm{reach}(u,v)&:=\exists P\subseteq\ecal,\exists\varepsilon_s,\varepsilon_t\in\ecal
                \big(\pathPred(P,\varepsilon_s,\varepsilon_t)\wedge \inc_t^{-}(\varepsilon_s,u)\wedge \inc_t^{+}(\varepsilon_t,v)\big).
                \end{align*}
                
                \smallparagraph{Non-strict}  In the encoding, replace the use of $t_1<t_2$ in $\psuc$ by $t_1\leq t_2$, yielding $\pnsuc$.
                The formulas remain unchanged.  
        \end{proof}
        Since each base-path contributes at most two temporal edges incident with any vertex, $\tdegree\leq2\tpn$. This implies:
        \begin{corollary} \label{cor: FPT tw + tpn}
            \openTCC and \closedTCC on (un)directed, (non-)strict temporal graphs are in \FPT parameterized by $\tw + \tpn$.
        \end{corollary}
        \begin{proof}
            The bound $\tdegree\leq2\tpn$ allows us to apply \Cref{thm: fpt tw plus degree}.
        \end{proof}


    \subsubsection{MSO formulation under \texorpdfstring{$\tw+\lifetime$}{tw + lifetime}}\label{subsec: tw plus lifetime FPT}
                We first present the encoding for \openTCC/\closedTCC on undirected strict temporal graphs, and then explain the minor adjustments needed for directed or non-strict temporal graphs, without fully restating the definitions and proofs.
    \begin{definition}
        A relational structure $\MSOlifetime=(U,V,E,\ecal,T,\inc, \timeTedge, \edgeTedge, \pres)$ has universe $U=V\cup E \cup\ecal \cup T$, unary predicates $V(\cdot),E(\cdot),\ecal(\cdot),T(\cdot)$ identifying vertices, static edges, temporal edges, and time steps, respectively, and binary relations  
        \begin{itemize}
            \item $\inc\subseteq E\times V$ where $\inc(v,e) \Leftrightarrow v\in e$,
            \item $\timeTedge\subseteq \ecal\times T$ where $\timeTedge((e,t),t')\Leftrightarrow t=t'$,
            \item $\edgeTedge\subseteq\ecal\times E$ where $\edgeTedge((e,t),e')\Leftrightarrow e=e'$,
            \item $\pres\subseteq E\times T$ where $\pres(e,t)\Leftrightarrow (e,t)\in\ecal$.
        \end{itemize}  
    \end{definition}
    By~\cite[Theorem 23]{haag_FeedbackEdge_2022}, the Gaifman graph has treewidth bounded by $\tw+\lifetime$.
    \begin{lemma}[\cite{haag_FeedbackEdge_2022}]
        The treewidth of the Gaifman graph of a structure representing a temporal graph \gcal in \MSOlifetime is bounded by $\tw+\lifetime$.
    \end{lemma}
    With this, we are ready to prove our theorem.
    \begin{restatable}{theorem}{twlifetimeFPT} \label{thm: fpt tw plus lifetime}
        \openTCC and \closedTCC on (un)directed, (non-)strict temporal graphs are in \FPT parameterized by $\tw+\lifetime$; can be solved in time $\mathcal{O}(f(\tw,\lifetime)\cdot n)$ for a computable function $f$.
    \end{restatable}
    \begin{proof}
        We define the optimization variant of \openTCC and \closedTCC as an MSO formula. 
        Let $\mathcal{I}=(\gcal,k)$ be an instance of \openTCC or \closedTCC interpreted as a relational structure \MSOlifetime.
        First, we express adjacency of two vertices $v$ and $w$ at time step $t$:
        \begin{align*}
            \tadj(v,w,t)&:=\exists e\in E \big( \inc(e,v) \wedge \inc(e,w) \wedge \pres(e,t)\big).\\
        \intertext{Next, we express open and closed temporal paths between two vertices:}
            \pathPred(u,v)&:=\exists x_0,\dots,x_{\lifetime}\in V
                \Big(x_0=u \wedge x_{\lifetime}=v \wedge
                \bigwedge_{t=0}^{\lifetime-1}
                \big( x_t=x_{t+1} \vee \tadj(x_t,x_{t+1},t)\big)
                \Big),\\
            \pathPredX(u,v,X)&:=\exists x_0,\dots,x_{\lifetime}\in X
                \Big(x_0=u \wedge x_{\lifetime}=v \wedge
                \bigwedge_{t=0}^{\lifetime-1} \big(x_t=x_{t+1} \vee \tadj(x_t,x_{t+1},t)\big)
                \Big).\\
        \intertext{The formula $\pathPred(u,v)$ checks whether there is a strict temporal path from $u$ to $v$ in $G$, while $\pathPredX(u,v,X)$ additionally restricts the path to visit only $X$. Both formulas have length upper-bounded by $2^{\mathcal{O}(\lifetime)}$.
        Finally, we express a set $X\subseteq V$  being an open/closed tcc:}
            \varphi_{open}(X)&:= X\subseteq V\wedge \forall u,v\in X 
            \big(\pathPred(u,v)\wedge \pathPred(v,u)
            \big),\\
        \varphi_{closed}(X)&:= X\subseteq V\wedge \forall u,v\in X 
            \big(\pathPredX(u,v,X)\wedge \pathPredX(v,u,X)
            \big).
        \end{align*}
        Using the affine goal function $\alpha(x)=x$, \Cref{thm: courcelle optimization} implies that the optimization variant of \openTCC and \closedTCC can be solved in time $f(\tw,\lifetime)\cdot n$ for some computable function $f$.

            \bigparagraph{Adjustments for directed / non-strict variants.}\quad
            We describe the adjustments necessary to formulate the component properties in directed or non-strict temporal graphs.
            They all keep the Gaifman treewidth bound of $\tw+\lifetime$ unchanged.
            
            \smallparagraph{Directed}
            In the encoding, replace $\inc$ with the two binary relations $\source,\target\subseteq E\times V$ where $\source((x,y),v)\Leftrightarrow x=v$ and $\target((x,y),v)\Leftrightarrow y=v$.
            In the formulas, replace the temporal adjacency $\tadj$ with
            \[
            \tadj^{\to}(u,v,t)\;:=\;\exists e\in E\;\big(\source(e,u)\wedge \target(e,v)\wedge \pres(e,t)\big).
            \]
            
            \smallparagraph{Non-strict}
            We allow multiple hops within a single snapshot by replacing $\tadj$ (two vertices are connected by an edge at time $t$) with an
            MSO-definable ``path-inside-snapshot'' predicate (two vertices are connected by a path at time $t$). Define the formula
            \begin{align*}
            \pathPred_t(u,v):=&\exists P\subseteq E\,\exists X\subseteq V\,\Big(
            u,v\in X\ \wedge\\ &\forall e\in P\,\big(\pres(e,t)\wedge\forall w\in V\,(\inc(e,w) \to w\in X)\big)\,\wedge\\
            &\deg_{=1}(u)\ \wedge\ \deg_{=1}(v)\ \wedge\ \forall w\in X\setminus\{u,v\}\,\big(\deg_{=2}(w)\big)\Big),
            \end{align*}
            where the degree formulas are:
            \begin{align*}
                \deg_{=0}(w):=&\neg\exists e\in P\,\big(\inc(e,w)\big),\\
                \deg_{=1}(w):=&\exists e\in P\,\big(\inc(e,w)\big)\ \wedge\ \forall e_1,e_2\in P\,
                \big(\inc(e_1,w)\wedge\inc(e_2,w)\to e_1=e_2\big),\\
                \deg_{=2}(w):=
                &\exists e_1,e_2\in P\, \big(e_1\neq e_2 \wedge \inc(e_1,w)\wedge\inc(e_2,w)\big)\ \wedge\ \\
                & \forall e_1,e_2,e_3\in P\,\Big(\bigwedge_{i\in[3]} \inc(e_i,w)\to \big(e_1=e_2\vee e_1=e_3\vee e_2=e_3\big)\Big).
            \end{align*}
            Accordingly, define $\pathPred_{t,X}$.
    \end{proof}

\subsection{Kernelization Lower Bounds}
\label{sec:kernel}
    Since \openTCC and \closedTCC are both in \FPT parameterized by $\tw+\tpn$, $\tw+\tdegree$ and $\tw+\lifetime$, it is natural to ask whether they admit a polynomial kernel. We show that this is not the case. The proof is a rather standard, straightforward proof based on the framework of \emph{cross-composition} introduced by Bodlaender, Jansen and Kratsch \cite{bodlaender_KernelizationLower_2014}, (see also the book by Fomin et al.~\cite{fomin_kernelizationtheory_2019}).
    \combinedKernelBounds*
We prove the two parameterizations in \Cref{thm:kernel-lower-bounds} separately in \Cref{thm:no_kernel_lifetime,thm:no_kernel_tpn}.

    \iflong \begin{definition}[Polynomial equivalence relation~\cite{fomin_kernelizationtheory_2019}]
    An equivalence relation $R$ on the set $\Sigma^*$ is called a \emph{polynomial equivalence relation} if the following conditions are satisfied:
    \begin{enumerate}
        \item There exists an algorithm that, given strings $x, y \in \Sigma^*$, resolves whether $x$ is equivalent to $y$ in time polynomial in $|x| + |y|$.
        \item For any finite set $S \subseteq \Sigma^*$ the equivalence relation $R$ partitions the elements of $S$ into at most $(\max_{x\in S} |x|)^{\bigoh(1)}$ classes. 
    \end{enumerate}
    \end{definition}
    \begin{definition}[OR-cross-composition~\cite{fomin_kernelizationtheory_2019}]\label{def:cross_composition}
        Let $L \subseteq \Sigma^*$ be a language and $Q \subseteq \Sigma^* \times \mathbb{N}$ be a parameterized language. 
        We say that $L$ \emph{cross-composes} into $Q$ if there exists a polynomial equivalence relation $R$ and an algorithm $\mathbb{A}$, called a \emph{cross-composition}, satisfying the following conditions. 
        The algorithm $\mathbb{A}$ takes as input a sequence of strings $x_1, x_2,\ldots, x_t \in \Sigma^*$ that are equivalent with respect to $R$, runs in time polynomial in $\sum_{i=1}^t |x_i|$, and outputs one instance $(y, k) \in \Sigma^* \times N$ such that 
        \begin{enumerate}
            \item $k\le p(\max^t_{i=1}|x_i|+\log t)$ for some polynomial $p(\cdot)$, and 
            \item $(y,t)\in Q$ if and only if there exists at least one index $i\in [t]$ such that $x_i\in L$.
        \end{enumerate}
    \end{definition}
    \begin{theorem}[\cite{fomin_kernelizationtheory_2019}]\label{thm:cross_composition_kernel_lowerbound}
    Let $L\subseteq \Sigma^*$ be an \NP-hard language. If $L$ cross-composes into parameterized problem $Q$ and $Q$ has a polynomial kernel, then $\coNP\subseteq \NP/poly$    
    \end{theorem}
    
    With this, we show that unless $\coNP\subseteq \NP/poly$, none of the studied problems admit a polynomial kernel parameterized by $\tw+\tdegree+\lifetime$.
    \fi
    The main observation is that each of $\tw$, $\tdegree$, and $\lifetime$ of a disjoint union of temporal graphs is bounded by their respective maximum in a single temporal graph in the union. \ifshort
    Moreover, $\tw$ and $\tdegree$ stay bounded, even if we connect the graphs in a chain.
    \fi
    \begin{theorem}\label{thm:no_kernel_lifetime}
        \openTCC and \closedTCC on (un)directed, (non-)strict graphs do not admit a polynomial kernel parameterized by $\tw+\tdegree+\lifetime$, unless $\coNP\subseteq \NP/poly$.
    \end{theorem}
    \iflong
    \begin{proof}
         Note that all versions of \openTCC and \closedTCC are \NP-hard even on instances with lifetime $\lifetime =2$~\cite{costa_ComputingLarge_2023}. We give a simple, very standard cross-composition from the problem on instances with lifetime $2$ to itself, parameterized by $\tw+\tdegree+\lifetime$.
         Let $R$ be an equivalence relation on the instances such that  $(\gcal_1=(V_1,E_1,\lambda_1), \paraCompSize_1)$ and $(\gcal_2=(V_2,E_2,\lambda_2), \paraCompSize_2)$, both with lifetime~$2$ are equivalent according to $R$ if and only if
         \begin{itemize}
             \item $|V_1| = |V_2|$, 
             \item $|\ecal_1| = |\ecal_2|$, and 
             \item $\paraCompSize_1 = \paraCompSize_2$.
         \end{itemize}  
         All instances with lifetime other than $2$ or strings that do not form a valid instance of the problem form another equivalence class. 
         It is easy to see that $R$ is a polynomial equivalence relation. 
         
         Now we give a cross-composition for instances belonging to the same equivalence class. For the equivalence class containing the invalid or non-lifetime-2 instances, we output a trivial no-instance. Let $(\gcal_1=(V_1,E_1,\lambda_1), \paraCompSize), (\gcal_2=(V_2,E_2,\lambda_2), \paraCompSize),\ldots, (\gcal_t=(V_t,E_t,\lambda_t), \paraCompSize)$ be valid instances of the same equivalence class such that $|V_i| = n$ for all $i\in [t]$. We simply output the instance $(\gcal = (V,E,\lambda),s)$, where $\gcal$ is the disjoint union of the temporal graphs $\gcal_1,\ldots, \gcal_t$. That is, $V = \biguplus_{i\in [t]}V_i$, $E = \biguplus_{i\in [t]}E_i$ and for all $e\in E$, if $e\in E_i$, then $\lambda(e) = \lambda_i(e)$. First, note that treewidth of the footprint of $\gcal$ is at most $n$, the maximum temporal total degree is at most $2n-2$ and lifetime is $2$. Hence, $\tw+\tdegree+\lifetime$ is bounded by $4n$. It remains to show that $(\gcal,s)$ is yes-instance if and only if there exists $i\in [t]$ such that $(\gcal_i, s)$ is yes-instance. Since, there is no (temporal) path between two vertices in different connected components of the footprint of $\gcal$, it is easy to see that if $X$ is an \otcc or \ctcc, then $X$ is fully contained in the vertices of exactly one of the original instances. That is, there exists $i\in [t]$ such that $X\subseteq V_i$. Hence if there exists $X\subseteq V$ such that $|X|\ge s$ and $X$ is an \otcc or a \ctcc in $\gcal$, then there exists $i\in [t]$ such that $X\subseteq V_i$ and $X$ is an \otcc or a \ctcc in $\gcal_i$. Similarly, if there exists $i\in [t]$ such that $\gcal_i$ contains an \otcc or a \ctcc of size $s$, then clearly $\gcal$ contains an \otcc or a \ctcc of size $s$, as $\gcal_i$ is fully contained in $\gcal$. Hence, $(\gcal, s)$ is yes-instance if and only if there is at least one $i\in [t]$ such that $(\gcal_i,s)$ is yes-instance.
         In conclusion, all versions of the problem cross-compose to themselves parameterized by $\tw+\tdegree+\lifetime$ and by Theorem~\ref{thm:cross_composition_kernel_lowerbound} do not admit a polynomial kernel unless $\coNP\subseteq\NP/poly$.
    \end{proof}

    Unfortunately, if we wish to bound $\tpn$ of the instance, we cannot take a disjoint union of the instances. However, we can chain them so we connect the ends of temporal paths in the $i$-th instance to the starts of temporal paths in the $(i+1)$-st instance. This still keeps \tw and \tdegree bounded by the size of a single instance; however, $\lifetime$ now becomes unbounded, since we do not wish to introduce new connectivity. This is indeed unavoidable, as a (strict) $k$-path graph with lifetime $\lifetime$ and temporal path number $\tpn$ has at most $k\cdot\lifetime$ temporal edges and hence a kernel of size $k\cdot\lifetime$.
    \fi
    \begin{theorem}\label{thm:no_kernel_tpn}
        \openTCC and \closedTCC on (un)directed, (non-)strict graphs do not admit a polynomial kernel parameterized by $\tw+\tpn$, and consequently by $\tw+\twn$, unless $\coNP\subseteq \NP/poly$.
    \end{theorem}
    \iflong
    \begin{proof}
        The proof is very similar to the proof of Theorem~\ref{thm:no_kernel_lifetime}. We give a cross-composition from the same variant of the problem with lifetime $2$. Let $(\gcal_1=(V_1,E_1,\lambda_1), \paraCompSize), (\gcal_2=(V_2,E_2,\lambda_2), \paraCompSize),\ldots, (\gcal_t=(V_t,E_t,\lambda_t), \paraCompSize)$ be valid instances of the same equivalence class such that $|V_i| = n$ and $|\ecal_i|=m$ for all $i\in [t]$, that is all instances have the same number of vertices and the same number temporal edges. First, if $\paraCompSize \le 2$, we can solve each instance in polynomial time by trying all subsets of vertices of size at most $\paraCompSize$ and checking reachability inside these components and output either a trivial yes-instance or a trivial no-instance. So, from now on, we assume $s\ge 3$.
        Let us arbitrarily order the temporal edges in each instance. Moreover, if the edges are undirected, let us, for each edge, pick an arbitrary direction on that edge and refer to one endpoint as the head and the other as the tail. 

        We now construct the instance $(\gcal = (V,E,\lambda),s)$ as follows. We start by taking the disjoint union of the footprints of the input instances, and for an edge $e\in E_i$, we let $\lambda(e) = \{\ell+(2m+2)(i-1)\mid \ell\in \lambda_i(e)\}$. That is, we shift the labels of edges, such that edges from $\gcal_1$ have labels $1$ and $2$, edges from $\gcal_2$ have labels $2m+3$ and $2m+4$, and edges from $\gcal_i$ have labels $(2m+2)(i-1)+1$ and $(2m+2)(i-1)+2$. Now let $(x_j^iy_j^i, t_j)$ be the $j$-the edge in $\ecal_i$, where $x_j^i$ is the head of the edge and $y_j^i$ the tail, and $(x_j^{i+1}y_j^{i+1}, t_j)$ be the $j$-the edge in $\ecal_{i+1}$ with head $x_j^{i+1}$ and tail $y_j^{i+1}$. We add to $\gcal$ a vertex $w_j^i$, an edge $y_j^iw_j^i$ with label $\lambda(y_j^iw_j^i)=(2m+2)(i-1)+2+j$ and an edge $w_j^ix_j^{i+1}$ with label $\lambda(w_j^ix_j^{i+1})=(2m+2)(i-1)+2+m+j$. This finishes the construction. 
        
        It follows rather straightforwardly from the construction that $\gcal$ is an $m$-path graph. Indeed, we can define $j$-th path to start in the head of the $j$-th edge in $\ecal_1$, and finish in the tail of the $j$-th path in $\ecal_t$ such that it contains all $j$-th edges in all $\ecal_i$'s, it passes the $j$-th edge in $\ecal_i$ at time step either $(2m+2)(i-1)+1$ or $(2m+2)(i-1)+2$, at time step $(2m+2)(i-1)+2+j$ it passes from the tail of the $j$-th edge in $\ecal_i$ to $w_j^i$ and at time step $(2m+2)(i-1)+2+m+j$ it passes from $w_j^i$ to the head of the $j$-th edge in $\ecal_{i+1}$. The treewidth is at most $2n$. This can be easily seen by taking a tree decomposition, which is a path on $m\cdot (t-1)$ vertices and letting the $((i-1)\cdot m + j)$-th bag be $V_i\cup V_{i+1}\cup \{w_j^i\}$ if $0<j<m$ and $V_{i-1}\cup V_{i}\cup \{w_m^{i-1}\}$ if $j=0$. Finally, the temporal total degree is at most $2n+2m$, since we added at most $2m$ edges to each vertex in the original instances, and the vertices $w_j^i$, for $j\in [m]$ and $i\in [t-1]$, have temporal total degree two. 

        It remains to show that $\gcal$ has a tcc of size $s$ if and only if there exists $i\in [t]$ such that $\gcal_i$ has a tcc of size $s$. First, let $X$ be a tcc in $\gcal_i$, it is easy to see that $X$ is a tcc in $\gcal$, as all temporal edges in $\gcal$ are shifted by the same constant $(2m+2)(i-1)$, so the difference between time steps of two edges in $\gcal_i$ remains the same. 
        
        On the other hand, let $v_i\in V_i$ and $v_j\in V_j$ be two vertices with $i < j$. By our construction, the latest time step on an edge incident on $v_i$ can be $(2m+2)(i-1)+1+m = (2m+2)i-m-1$ and the earliest edge incident on $v_j$ can be $(2m+2)(j-2)+2+m+1=(2m+2)(j-1)-m+1$. Since $i\le j-1$, it follows that the latest edge incident on $v_i$ is earlier than the earliest edge incident on $v_j$, and so there cannot be a temporal path from $v_j$ to $v_i$, and they cannot be in the same temporal component. Similarly, $w_j^i$ can only be in connected components of size $2$ with either of its neighbors. Note that if $\gcal$ is directed, it is easy to see that each $w_j^i$ is a strong component in the footprint, so the maximum size of a tcc that contains $w_j^i$ is $1$. Let us now consider the case when $\gcal$ is undirected. First, let us see what vertices $w_j^i$ can reach by starting with the edge to its neighbor $y_j^i$ in $V_i$ at time $(2m+2)(i-1)+2+j$. All edges in $E_i$ have labels $(2m+2)(i-1)+1$ or $(2m+2)(i-1)+2$, so it cannot continue on any of those edges. Similarly, all edges from some $w_{j'}^{i-1}$ have label at most $(2m+2)(i-2)+2+2m=(2m+2)(i-1)$. It can, however, use an edge to some $w_{j'}^i$ if $j' > j$ and $y_{j'}^i=y_j^i$. Now, by starting with the edge to its neighbor $x_j^{i+1}$ in $V_{i+1}$ at time $(2m+2)(i-1)+2+m+j$, it can potentially reach any vertex $w_{j'}^{i'}$ or any vertex in $V_{i'}$, for $j'\in [m]$ and $i'\ge i+1$. However, to reach a vertex $w_{j'}^i$ for $j'\in [m]$, it is necessary that $x_{j'}^i=x_j^i$ and $j' > j$, as the only way to reach $w_{j'}^i$ from $V_{i+1}$ is at time step $(2m+2)(i-1)+2+m+j'$ by the edge from $x_{j'}^{i+1}$. Hence, $w_j^i$ can only reach $w_{j'}^{i'}$ if either $i < i'$ or $i'=i$ and $j < j'$. Since this holds for an arbitrary $j\in [m]$ and $i\in [t-1]$, it follows that $w_j^i$ cannot be in a tcc with another $w_{j'}^{i'}$. Moreover, the only vertex in $\bigcup_{1\le p \le  i}V_p$ that $w_{j}^i$ can reach is $y_j^i$. By an analogous argument, we get that the only vertex in $\bigcup_{i+1\le p \le  t}V_p$ that can reach $w_j^i$ is $x_{j}^{i+1}$. Since $x_{j}^{i+1}$ cannot reach $y_{j}^i$, the largest component that can contain $w_{j}^i$ has size $2$. Since $s \ge 3$, it follows that if $X$ is a tcc of size at least $s$ in $\gcal$, it is fully contained in $\gcal_i$ for some $i\in [t]$. 
        
        For the case of the \closedTCC, it remains to show that all connections are also achieved by a path fully in $\gcal_i$. Since we already argued that $w_j^i$, for some $j\in [m]$, can reach only a single vertex in $V_{i}$ and only by a direct edge at time step $(2m+2)(i-1)+2+j$, and $w_j^i$ cannot be reached before this time step, $w_j^i$ cannot be on any temporal path between two vertices in $V_i$. Similarly, for $w_{j}^{i-1}$ for $j\in [m]$, only way to enter $w_{j}^{i-1}$ from $V_i$ is by temporal edge edge at $(2m+2)(i-2)+2+m+j$, but only way to leave $w_{j}^{i-1}$ after arriving at this time is if we have a non-strict temporal graph and we take the same edge, so $w_{j}^{i-1}$ also cannot be on a temporal path between two vertices in $V_i$. Hence all temporal paths between vertices in $X$ are fully contained in $\gcal_i$, just with labels shifted by $(2m+2)(i-1)$.
        Since $\twn(\gcal)\leq\tpn(\gcal)$, the output parameter $\tw+\twn$ is bounded by the same polynomial as $\tw+\tpn$. Thus, the same cross-composition also rules out a polynomial kernel for $\tw+\twn$ unless $\coNP\subseteq\NP/poly$.
    \end{proof}
    \fi

\section{Extended Outlook}
Our results nearly settle the parameterized complexity of maximum open and closed temporal connected components for the footprint parameters considered: only vertex cover admits \XP algorithms, and only in undirected temporal graphs. Treewidth, pathwidth, and feedback vertex number leave both problems \NP-hard even for small constant values. Combining treewidth with lifetime or maximum temporal total degree yields \FPT algorithms, but these temporal parameters tend to be large in practice.

Temporal path number offers a different structural perspective: it measures how many temporal routes suffice to describe all interactions. Such a description is natural for networks organized around a small number of journeys, for example in transportation networks. Unlike the footprint parameters considered, bounded \tpn permits \XP algorithms for open components in both directed and undirected temporal graphs; this extends even to temporal walk number. The algorithms do not require a decomposition into base-paths or base-walks as part of the input: the existence of such a decomposition bounds the VC-dimension of the family of maximal open components and thereby enables efficient enumeration.

These results make temporal path number a promising parameter for studying temporal connectivity beyond restrictions on the footprint. Nevertheless, its algorithmic benefits have limits. Open components remain \Wone-hard parameterized by \tpn, while closed components are \NP-hard already for $\tpn=6$. This contrasts with undirected vertex cover, where temporal out-tree certificates yield \XP algorithms even for closed components albeit only on undirected graphs. Thus, restricting the number of temporal routes controls open connectivity, but does not by itself suffice for computing closed components efficiently.

Several questions remain open. For the footprint parameters, the remaining gaps concern small constant values, including undirected feedback vertex number one, pathwidth two, and directed vertex cover number two or three. For closed components, it would also be interesting to determine the smallest temporal path number for which \NP-hardness holds; our construction gives an upper bound of six.

Beyond closing these gaps, a promising direction is to restrict how the base-paths interact. Our hardness constructions use paths that visit shared vertices in conflicting orders. Requiring pairwise monotone paths---whose orders agree on their shared vertices---may yield stronger tractability results while retaining the interpretation of the graph as a collection of temporal routes. Does this restriction make open components \FPT parameterized by \tpn, or even permit \XP algorithms for closed components? More generally, which restrictions on the interaction of a bounded number of temporal paths suffice for tractability?

Other structural parameters also should be investigated, including the feedback edge number of the footprint and temporal feedback vertex or edge numbers, which measure the minimum number of vertices or edges whose removal eliminates all temporal cycles \cite{haag_FeedbackEdge_2022}. Such parameters offer further ways to constrain temporal structure beyond the footprint restrictions considered here.

\bibliography{references,MyLibrary}
\appendix

\section{\Wone-hardness for Mixed 2-Sided Orthant Depth} \label{sec:app-m2od}
In the general \textup{\textsc{Box Depth}} problem, we are given a set of $n$ axis-aligned boxes in $\mathbb{R}^d$ together with an integer~$\ell$, and the task is to decide whether there is a point $x\in\mathbb{R}^d$ that is contained in at least $\ell$ boxes. A \emph{2-sided orthant} is an axis-aligned region that is unbounded in $d-2$ dimensions and bounded from one side in each of two distinct dimensions; that is, it is of the form $\{x\in\mathbb{R}^d : (x_i\,{\mathord ?}\,a)\land(x_j\,{\mathord ?}\,b)\}$ for distinct $i,j\in[d]$, where each occurrence of ${\mathord ?}$ may independently be $\leq$ or $\geq$~\cite{chan_kleesmeasure_2013}. We consider the further restricted \textup{\textsc{Mixed 2-Sided Orthant Depth (M2SOD)}} problem, in which every region has one upper-bounded and one lower-bounded coordinate. Thus, every region is of the form $\{x\in\mathbb{R}^d : x_i\leq a,\ x_j\geq b\}$ for some distinct $i,j\in[d]$.



We show that allowing two bounds per region already makes \textsc{Box Depth} W[1]-hard parameterized by the dimension, even when every region has exactly one lower bound and one upper bound on distinct coordinates and all bounds are integers. Our reduction from Multicolored Clique uses two coordinates per color class and polynomially many explicitly listed orthants with integer bounds.

Formally, an instance of \textup{\textsc{Mixed 2-Sided Ortant Depth}} consists of a dimension $d\geq2$, an explicitly listed indexed family $\mathcal O=(O_1,\ldots,O_N)$ of mixed 2-sided orthants with rational bounds encoded in binary, and a nonnegative integer threshold $\ell$ encoded in binary. The question is whether there exists $z\in\mathbb R^d$ such that $\operatorname{depth}_{\mathcal O}(z):=\bigl|\{r\in[N]:z\in O_r\}\bigr|\geq\ell$. 


Note that we allow repeated regions: the input is an indexed family of orthants, and depth counts each occurrence separately. The multiplicities used in our reduction are implemented by explicitly listing polynomially many copies. An equivalent construction with unique orthants would be possible by slightly offsetting the bound values in duplicate orthants; with integer bounds, all values can be scaled up accordingly to make space for such small offsets.

We write $[n]=\{1,\ldots,n\}$. Throughout, all bounds are nonstrict and all regions are closed. For ease of reading, we give names to dimensions instead of enumerating them with integers, and denote dimension $d$ of a point $z$ as $z[d]$.

\bigparagraph{Intuition.} We reduce from \textup{\textsc{Multicolored Clique}}.
For each color class $i$, we use two dimensions $x_i,y_i$ to encode the choice of a vertex in this class. Specifically, $z[x_i]=s=-z[y_i]$ will indicate choosing the vertex with index $s$ in color class $i$. We will call every point that satisfies $z[x_i]=s=-z[y_i]$ for every dimension $i$ a \textit{valid choice}.
The two coordinates let us express a bound on the selected index using either coordinate: for a point $z$ representing a valid choice, $z[x_i]\leq a$ is equivalent to $z[y_i]\geq-a$, and $z[x_i]\geq a$ is equivalent to $z[y_i]\leq-a$.
To ensure that every solution is a valid choice we introduce \textit{choice orthants} with sufficiently large multiplicity. 
For each edge of the input graph, four further orthants contribute a baseline of one at every valid choice, and a contribution of two precisely when the two endpoints of that edge are selected.
Reaching the target depth will then require an edge between every pair of selected vertices.
The construction uses only $2k$ dimensions for $k$ color classes, and every orthant has one lower and one upper bound on distinct coordinates.

\begin{theorem}\label{thm:M2SOD-wone-hard}
    \textup{\textsc{M2SOD}} is \Wone-hard when parameterized by the dimension.
\end{theorem}
\begin{proof}
    We give a parameterized reduction from \MCclique, which is \Wone-hard when parameterized by the number of color classes~\cite{fellows_parameterizedcomplexity_2009}. 
    Let $(H=(V_H,E_H), (V_1,\dots,V_k))$ be an instance of \MCclique for a simple undirected graph $H$. 
    Instances with $k<2$ can be decided directly, so we assume $k\geq2$. We additionally assume that every color class has the same size $n\geq1$, and  add isolated vertices where necessary.
    Write $V_i=\{v_{i,1},\ldots,v_{i,n}\}$ and, for $i<j$, let $E_{ij}$ denote the set of edges between $V_i$ and $V_j$.

    We construct an instance $(\mathcal O,\ell)$ of \textup{\textsc{M2SOD}} in dimension $d=2k$. For each $i\in[k]$, introduce two dimensions $x_i$ and $y_i$. A point $z\in\mathbb{R}^{2k}$ is a \emph{valid choice} if for every $i\in[k]$ there is an $\alpha_i\in[n]$ with $z[x_i]=\alpha_i$ and $z[y_i]=-\alpha_i$. Thus, a valid choice encodes the selection of one unique vertex $v_{i,\alpha_i}$ from each color class.
    The family of orthants $\mathcal O$ is composed as follows. For every vertex $v_{i,\alpha}$ with $i\in[k]$ and $\alpha\in[n]$, introduce two \textit{choice orthants}:\begin{align*}
        U_{i,\alpha}=\{z\in\mathbb{R}^{2k}:z[x_i]\geq \alpha,\ z[y_i]\leq -\alpha\},
        \quad\quad
        L_{i,\alpha}=\{z\in\mathbb{R}^{2k}:z[x_i]\leq \alpha,\ z[y_i]\geq -\alpha\}.
    \end{align*}
    Each of these orthants will be added with multiplicity $M$, defined below. See \Cref{fig:orthants1} for an illustration of the choice orthants for a single vertex $v_{i,\alpha}$.
    \begin{figure}[t]
        \centering
        \includegraphics[width=0.45\linewidth]{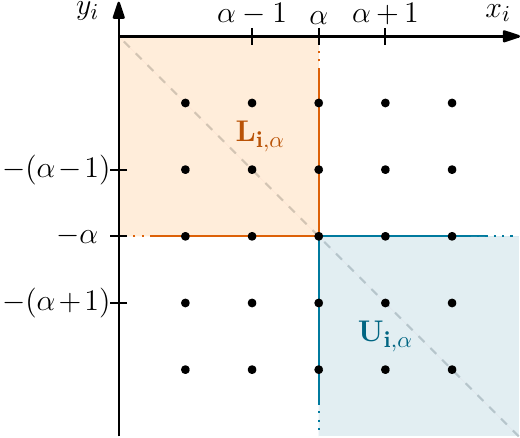}
        \caption{Illustration of the choice orthants $L_{i,\alpha}$ and $U_{i,\alpha}$ for a single color class $V_i$. The valid choices are the points $z$ with $z[x_i]=-z[y_i]\land z[x_i]\in [n]$ on the diagonal line (dashed gray).}
        \label{fig:orthants1}
    \end{figure}

    For every edge $v_{i,\alpha}v_{j,\beta}\in E_{ij}$ with $i,j\in[k], i<j$, introduce four \textit{edge orthants}:
    \begin{align*}
        &D_1=\{z:z[x_i]\leq \alpha,\ z[x_j]\geq \beta\},&&
        D_2=\{z:z[y_i]\leq -\alpha,\ z[y_j]\geq -\beta\},\\
        &D_1'=\{z:z[x_i]\leq \alpha-1,\ z[y_j]\geq -(\beta-1)\},&&
        D_2'=\{z:z[y_i]\leq -(\alpha+1),\ z[x_j]\geq \beta+1\}.
    \end{align*}
    See \Cref{fig:orthants2} for an illustration of the four orthants associated with an edge $v_{i,\alpha}v_{j,\beta}$.
    \begin{figure}[t]
        \centering
        \includegraphics[width=0.45\linewidth]{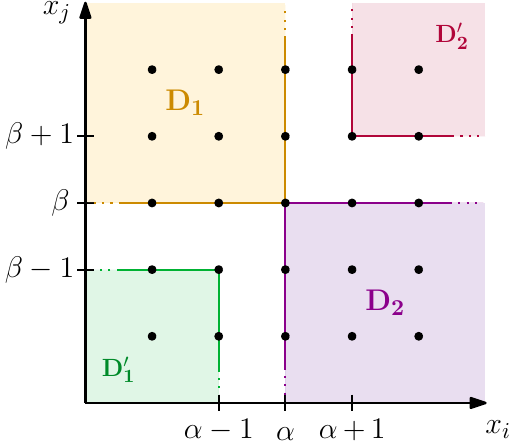}
        \caption{Illustration of the edge orthants associated with an edge $v_{i,\alpha}v_{j,\beta}$. Orthants are shown along the dimensions $x_i$ and $x_j$; assuming only valid choices, we represent a bound along $y_i$ (resp. $y_j$) via the inverted bound on $x_i$ (resp. $x_j$). 
        }
        \label{fig:orthants2}
    \end{figure}

    The total number of edge orthants is $4\lvert E_H\rvert$, and we set $M:=4\lvert E_H\rvert+1$.
    Finally, we set the target to $\ell=kM(n+1)+\lvert E_H\rvert+\binom{k}{2}$. This completes the construction.

    {Consider any point $z\in \mathbb R^d$.} Among the choice orthants $U_{i,1},\ldots,U_{i,n}$ for some $i\in [k]$, precisely the ones with index at most ${\min\{z[x_i],-z[y_i]\}}$ contain $z$. Analogously, among $L_{i,1},\ldots,L_{i,n}$, precisely those with index at least ${\max\{z[x_i],-z[y_i]\}}$ contain $z$. Hence, $z$ is contained in at most $n+1$ of these $2n$ orthants. This maximum of $n+1$ orthants is reached by $z$ if and only if ${z[x_i]=-z[y_i]=\alpha_i}$ for some $\alpha_i\in[n]$, that is, if $z$ encodes a valid choice for color $i$.
    Now fix an edge $v_{i,\alpha}v_{j,\beta}\in E_{ij}$ and let $z\in\mathbb R^d$ be a valid choice, selecting $v_{i,{s}}\in V_i$ and $v_{j,{t}}\in V_j$. The four edge orthants $D_1,D_2,D_1',D_2'$ associated with $v_{i,\alpha}v_{j,\beta}$ contain $z$ under the following conditions:
    \begingroup\interdisplaylinepenalty=10000
    \begin{align*}
        z\in D_1 &\quad\Longleftrightarrow\quad {s}\leq\alpha \text{ and } {t}\geq\beta, &
        z\in D_2 &\quad\Longleftrightarrow\quad {s}\geq\alpha \text{ and } {t}\leq\beta,\\
        z\in D_1' &\quad\Longleftrightarrow\quad {s}<\alpha \text{ and } {t}<\beta, &
        z\in D_2' &\quad\Longleftrightarrow\quad {s}>\alpha \text{ and } {t}>\beta.
    \end{align*}\endgroup
    Thus, $D_1,D_2,D_1',D_2'$ contain every pair $({s},{t})\in[n]^2$ exactly once, except $(\alpha,\beta)$, which is contained in both $D_1$ and $D_2$.
    Consequently, the four edge orthants associated with $v_{i,\alpha}v_{j,\beta}$ contribute $2$ if ${s}=\alpha$ and ${t}=\beta$, and $1$ otherwise, to the depth of $z$.
    It follows that the edge orthants associated with $E_{ij}$ contribute $\lvert E_{ij}\rvert+1$ if $v_{i,{s}}v_{j,{t}}\in E_{ij}$, and $\lvert E_{ij}\rvert$ otherwise.

    It remains to prove that $H$ has a multicolored clique if and only if there is a point contained in at least $\ell$ orthants.

    \bigparagraph{$(\Rightarrow)$}\quad
    Assume $\{v_{1,\alpha_1},\ldots,v_{k,\alpha_k}\}$ is a multicolored clique in $H$, and consider the valid choice $z$ with $z[x_i]=\alpha_i\land z[y_i]=-\alpha_i$ for all $i\in [k]$.
    Then $z$ is contained in exactly $kM(n+1)$ choice orthants. Moreover, every pair $i<j$ contributes $\lvert E_{ij}\rvert+1$ edge orthants, since $v_{i,\alpha_i}v_{j,\alpha_j}\in E_H$. Hence, the total depth of $z$ is $kM(n+1)+\sum_{i<j}(\lvert E_{ij}\rvert+1)=kM(n+1)+\lvert E_H\rvert+\binom{k}{2}=\ell$.

    \bigparagraph{$(\Leftarrow)$}\quad
    Assume there is a point $z\in\mathbb R^d$ with depth at least $\ell$. 
    If $z$ does not encode a valid choice for some color $i\in[k]$, then the choice orthants can contribute at most $kM(n+1)-M$ in total. Since all edge orthants together contribute at most $4\lvert E_H\rvert$, the total depth of $z$ is at most $kM(n+1)-M+4\lvert E_H\rvert=kM(n+1)-1<\ell$, which is a contradiction.
    Hence $z$ must be a valid choice, selecting $v_{i,\alpha_i}\in V_i$ for $i\in[k]$.
    By the above observation, its depth equals $kM(n+1)+\lvert E_H\rvert+\sum_{i<j}\mathds{1}_{v_{i,\alpha_i}v_{j,\alpha_j}\in E_H}$.
    Reaching $\ell$ forces every one of the $\binom{k}{2}$ selected vertex pairs to be adjacent. Thus, the selected vertices form a multicolored clique in $H$.

    The construction has dimension $d=2k$ and produces exactly $N=2kn(4|E_H|+1)+4|E_H|$ orthants. Every orthant uses two distinct coordinates, one lower bound and one upper bound, where all bounds use integers in $[-(n+1),n+1]$. Thus the construction has polynomial size. Since $d=2k$ and \MCclique is \Wone-hard by $k$, this proves \Wone-hardness of \textup{\textsc{M2SOD}} by $d$.
\end{proof}

\section{Polynomial-Time Algorithms on Forest Footprints}\label{sec:tree-footprint-algorithms}
    Jarry and Lotker showed that a maximum temporal connected component can be found in polynomial time on undirected temporal graphs whose footprint is a tree~\cite[Section~4]{jarry_ConnectivityEvolving_2004}. We first recall their model and explain the idea of their algorithm. For completeness, we then prove that the algorithm can directly be extended to \openTCC and \closedTCC directed temporal graphs whose underlying undirected footprint is a forest. Lastly, we explain why that approach fails for treewidth 2 and feedback vertex number 1.

    \bigparagraph{The (interval) evolving-graph model.}
    A \textit{timed evolving graph} 
    is a graph $G=(V,E)$, an ordered sequence of its subgraphs $S_G=(G_1,G_2,\dots,G_\tau)$ such that $\bigcup_{i=1}^{\tau}G_i=G$, and a function $\xi$ specifying the nonnegative duration of each edge traversal.
    Given times $t_1<\cdots<t_{\tau+1}$, the graph $G_i$ describes the available vertices and edges during $[t_i,t_{i+1})$.
    A \emph{journey} is a path together with departure times that respect edge availability and allow each traversal to finish before the next one starts. Waiting is allowed.
    Although this definition uses a sequence of snapshots, the input-size measure of Jarry and Lotker counts changes in availability \cite[Section 2.3]{jarry_ConnectivityEvolving_2004}. They bound the size of this description by $\delta\cdot(|V(G)|+|E(G)|)$, where $\delta$ denotes the \textit{dynamics} of the graph, the maximum number of activations and deactivations of any vertex or edge. 
    Equivalently, one can describe each vertex and edge by its \textit{availability intervals}, storing only the interval endpoints. Each availability interval contributes one activation and one deactivation, counting both endpoints. Thus, the bound counts interval endpoints, rather than individual time steps at which an edge is available. This is an interval representation of temporal graphs, as also considered here~\cite{aubian_complexitygaps_2026}.

    Besides this difference in representation, the evolving model allows travel times. When departures are restricted to integer times, the strict temporal graph model corresponds to uniform travel time $1$ and the non-strict model to uniform travel time $0$. 
    Here, a temporal edge specifies an available departure; in the interval model, a traversal must also respect the availability interval.

    The algorithm below assumes an explicit list of temporal edges, with $M:=|\mathcal E|$, and its running-time bound is polynomial in $|V|+M$. Thus, the bound also applies when the input explicitly lists $\tau$ discrete snapshots, since $M=\sum_{i=1}^{\tau}|E(G_i)|\leq\tau|E(G)|$. For a compact interval representation, however, the number of temporal edges may be exponential in the encoded input size: an interval with binary-encoded endpoints can contain exponentially many integer time steps. Our bound therefore does not by itself establish polynomial time for that representation. Jarry and Lotker's algorithm avoids this expansion by working directly with critical arrival and departure times.

    \bigparagraph{The idea of the tree algorithm.}
    Root the tree at an arbitrary vertex $r$. A subtree $T_x$ rooted at vertex $x$ communicates with the rest of the tree only through $x$. For a partial solution containing $x$, we therefore record two times: the latest earliest-arrival at $x$, and the earliest latest-departures from $x$. They describe when every vertex can reach $x$ and when $x$ can still reach every vertex. 
    In that sense, $x$ is a pivot vertex of the partial solution.
    Jarry and Lotker use this observation to build a dynamic program. They consider only arrival and departure times which are necessary to reach some vertex from $x$/for some vertex to reach $x$, obtaining at most $|V(T_x)|$ times of each kind for the subtree $T_x$ rooted at $x$. Another useful observation is that at most one child part can have its arrival time later than its departure time: two such parts would prevent mutual reachability through $x$. Their algorithm considers this possible exceptional child separately~\cite[Section~4]{jarry_ConnectivityEvolving_2004}. The proof below uses the same two-time description, but combines the children one at a time and considers the departure and arrival times of all temporal edges. This gives a simpler polynomial-time implementation for proving the directed extension.

    For the proof below, we equip every temporal edge $\varepsilon=(uv,t)\in\mathcal E$ with a nonnegative rational travel time $\xi(\varepsilon)$, given as part of the input. Traversing $\varepsilon$ starts at time $t$ and ends at time $t+\xi(\varepsilon)$. A temporal path respects these travel times if every edge is entered at or after the arrival along the preceding edge. Waiting is allowed, and travel times may differ even between occurrences of the same arc.

    \bigparagraph{The directed adjustment.}
    Here, a tree footprint always means that the \emph{underlying undirected graph} is a tree. Both arcs $(x,y)$ and $(y,x)$ may be present, and their departure times and travel times may differ. The root is chosen only to organize the dynamic program; it imposes no restriction on the directions of the arcs.

    We extend an observation of Jarry and Lotker~\cite[Claim~1]{jarry_ConnectivityEvolving_2004} to directed temporal graphs.
    \begin{observation}\label{obs:tree-open-closed}
        In every directed temporal graph with tree footprint and nonnegative travel times, every open tcc is closed. 
        Hence, a solution to \openTCC is also a solution to \closedTCC and vice versa.
    \end{observation}
    \begin{proof}
        Let $C$ be a nonempty open tcc and let $T_C$ be the smallest subtree containing $C$. If $|C|=1$, the statement is immediate. Otherwise, every leaf of $T_C$ belongs to $C$. For any distinct $u,v\in V(T_C)$, choose leaves $p,q$ such that the unique $p$-$q$-path contains $u$ and $v$ in this order. A temporal path from $p$ to $q$ must follow this unique footprint path, so it contains a temporal path from $u$ to $v$. Likewise, a temporal path from $q$ to $p$ contains one from $v$ to $u$. All these paths lie in $T_C$. Hence, $V(T_C)$ is closed temporally connected. 
        Since $C$ is maximal, we have $C=V(T_C)$ and it must be a closed tcc. 
    \end{proof}

    We now extend the polynomial-time result of Jarry and Lotker~\cite[Theorems~2 and~3]{jarry_ConnectivityEvolving_2004} to directed temporal graphs with arbitrary nonnegative travel times.
    \begin{theorem}\label{thm:tcc-tree-footprint}
        \openTCC and \closedTCC are polynomial-time solvable on directed and undirected temporal graphs with arbitrary nonnegative rational travel times whose underlying undirected footprint is a forest, given an explicit list of temporal edges and their travel times. 
        
        In particular, \openTCC and \closedTCC are polynomial-time solvable on (non-)strict (un)directed temporal graphs with forest footprint.
    \end{theorem}
    \begin{proof}
        We give the algorithm for directed graphs. This algorithm can be applied to undirected graph by replacing every undirected temporal edge by two opposite arcs with the same departure time and travel time, which preserves temporal reachability within every vertex set. We treat each tree of the forest footprint separately, since a connected component can only be within one footprint component. Let $T$ be one such tree, root it arbitrarily at $r$, and let $T_x$ denote the subtree rooted at vertex $x\in V$. Throughout the proof, paths can be concatenated whenever the first arrives no later than the second departs.

        \bigparagraph{States.}
        Let $L:=\{t:(e,t)\in\mathcal E\}\cup\{t+\xi(e,t):(e,t)\in\mathcal E\}$ contain all departure and arrival times of temporal edges, and let $q:=|L|\leq 2M$. For a closed temporally connected set $C\subseteq V(T_x)$ containing $x$, define $a_x(C)$ as the maximum, over $v\in C\setminus\{x\}$, of the earliest arrival time of a temporal $v$-$x$ path in $\gcal[C]$. Define $d_x(C)$ as the minimum, over $v\in C\setminus\{x\}$, of the latest departure time of a temporal $x$-$v$ path in $\gcal[C]$. For $C=\{x\}$, set $a_x(C):=-\infty$ and $d_x(C):=+\infty$.

        The entry $\mathsf{DP}_x[a,d]$ stores the maximum size of a set $C$ with $a_x(C)=a$ and $d_x(C)=d$. Impossible entries have value $-\infty$. Every finite arrival or departure time belongs to $L$, so each table has at most $(q+1)^2$ entries. We keep exact times so that every later {feasibility test for combining entries} depends only on the table indices.

        \bigparagraph{Passing from a child to its parent.}
        Let $y$ be a child of $x$ and consider a feasible entry $\mathsf{DP}_y[{a_y},{d_y}]$. The corresponding set $C_y$ contains $y$. To use it in a solution containing $x$, all its vertices must reach $x$ and be reachable from $x$. Define $A:=\min\{t+\xi((y,x),t):t\in\lambda(y,x),\ {a_y}\leq t\}$ and $D:=\max\{t:t\in\lambda(x,y),\ t+\xi((x,y),t)\leq{d_y}\}$. If either set is empty, discard this entry when considering an attachment to $x$. A nonexisting arc has an empty label set.
        These are the points where direction matters: arrivals at $x$ use $(y,x)$, whereas departures from $x$ use $(x,y)$. In the undirected case, both directions have the same available departures and travel times. 
        The values $A$ and $D$ are exactly the arrival and departure times of the whole child part, measured at $x$. Indeed, every path from $C_y$ to $x$ must end with $(y,x)$. The vertex whose earliest arrival at $y$ is ${a_y}$ determines the earliest arrival at $x$ that \textbf{all} vertices can achieve: every other vertex can reach $y$ by ${a_y}$ and wait for the departure at time ${a_y}$.
        Equivalently, arriving at $y$ by ${d_y}$ allows all vertices of $C_y$ to be reached, and the vertex with latest departure time ${d_y}$ shows that a later arrival cannot suffice. Thus, $D$ is the latest departure from $x$ that works for all vertices.

        \bigparagraph{Combining the children.}
        At $x$, start with the singleton $\{x\}$, represented by the entry $(-\infty,+\infty)$ of value $1$. Process the children one at a time. An intermediate table describes solutions using $x$ and vertices from the children processed so far. For each child, we may leave out its entire subtree. Otherwise, combine an intermediate entry $(a,d)$ with a feasible child entry $({a_y},{d_y})$, whose times at $x$ are $(A,D)$ as above, if and only if $a\leq D$ and $A\leq d$. The resulting entry is $(\max\{a,A\},\min\{d,D\})$, and its value is the sum of the two entry values. For each pair of indices, keep only the largest value. After processing all children, this is $\mathsf{DP}_x$.
        To justify the test, let $C_0$ be the solution represented by the intermediate entry. Every path from $C_0\setminus\{x\}$ to $C_y$ passes through $x$. All such pairs are reachable exactly when the latest earliest arrival from $C_0\setminus\{x\}$ is no later than a feasible departure into every vertex of $C_y$, which is $a\leq D$. The reverse direction gives $A\leq d$. These conditions also suffice, since the two path segments lie in different child subtrees and intersect only at $x$. If $C_0=\{x\}$, both comparisons hold automatically. Reachability between $x$ and $C_y$ was already checked when computing $A,D$.

        \bigparagraph{Correctness and running time.}
        By induction over the processed children, every table entry represents a closed temporally connected set with exactly the recorded times. Conversely, any closed temporally connected set containing $x$ induces a subtree of $T$. Its nonempty intersection with a child subtree contains that child and is itself closed temporally connected, because the footprint paths are unique. Thus, its parts occur in the child tables, and their mutual reachability guarantees that they pass the two comparisons when combined. Replacing a part by a larger one with the same times preserves every comparison. Hence, keeping only the largest value for each pair of indices loses no optimum.
        Any maximum closed tcc has a unique vertex $x$ closest to the chosen root and is contained in $T_x$. Taking the largest table value over all vertices therefore gives its size. Storing the choices made in each transition also recovers the vertex set. By \Cref{obs:tree-open-closed}, this solves \openTCC as well.

        There are $\mathcal O((q+1)^2)$ entries in each intermediate or child table. For each parent-child pair, we try at most $\mathcal O((q+1)^4)$ pairs of entries. For each arc, sorting its temporal edges by departure time and storing suffix minima of their arrival times permits the computation of $A$ for every ${a_y}$. Sorting by arrival time and storing prefix maxima of departure times similarly gives $D$ for every ${d_y}$. Scanning these lists together with the sorted set $L$ precomputes all transitions across arcs in $\mathcal O(M\log(M+2)+n(q+1))$ arithmetic and comparison operations. The total number of operations is therefore $\mathcal O(M\log(M+2)+n(q+1)^4)$, where $q\leq 2M$. Every time used is an input departure time or the sum of such a time and one input travel time, so its encoding length is polynomial in the input size as well.
        Finally, assigning travel time $0$ to every temporal edge gives exactly non-strict paths. For integer time labels, assigning travel time $1$ gives exactly strict paths, since $t+1\leq t'$ if and only if $t<t'$. These assignments preserve reachability within every vertex set and therefore preserve both open and closed tccs. This proves the stated consequence for both variants.
    \end{proof}



\end{document}